\documentclass[11pt, a4paper]{article}
\pdfoutput=1
\usepackage[utf8]{inputenc}
\usepackage{jheppub}
\hypersetup{pdfencoding=unicode, bookmarksopen=true, bookmarksnumbered}
\usepackage{microtype}
\usepackage{amsthm,amsbsy,amsfonts,mathrsfs,enumerate,float,wrapfig,amsmath,mathtools,graphicx,framed,tcolorbox}
\usepackage{tikz}
\usepackage{tikz-cd}
\usepackage{subcaption}
\usepackage{graphbox}
\tikzcdset{
    scale cd/.style={every label/.append style={scale=#1},
    cells={nodes={scale=#1}}}}
\tikzset{
    partial ellipse/.style args={#1:#2:#3}{
        insert path={+ (#1:#3) arc (#1:#2:#3)}
}}

\DeclareMathOperator{\PE}{PE}
\DeclareMathOperator{\rank}{rank}
\DeclareMathOperator{\Tr}{Tr}
\newcommand{\tauftn}[2]{\tau\genfrac[]{0pt}{1}{#1}{#2}}
\newcommand{\tauftnd}[2]{\tau_{\mathrm{6d}}\genfrac[]{0pt}{1}{#1}{#2}}

\theoremstyle{definition}
\newtheorem*{lemma}{Lemma}

\usetikzlibrary{positioning}
\usetikzlibrary{arrows}

\title{Bootstrapping bilinear relations of discrete Painlev\'e systems from 5d gauge theories}

\author[a]{Minsung Kim,}
\author[b,c]{Xin Wang}
\affiliation[a]{Quantum Universe Center, Korea Institute for Advanced Study, Seoul 02455, Korea}
\affiliation[b]{Interdisciplinary Center for Theoretical Study, University of Science and Technology of China, Hefei, Anhui 230026, China}
\affiliation[c]{Peng Huanwu Center for Fundamental Theory, Hefei, Anhui 230026, China}

\emailAdd{minsung1@kias.re.kr}
\emailAdd{wxin@ustc.edu.cn}

\abstract{
    We study tau-functions of elliptic and $ q $-Painlev\'e equations through their correspondence with 5d $ \mathcal{N}=1 $ supersymmetric field theories. Defining tau-functions associated with 5d field theories in terms of partition functions on the $ \Omega $-background, we propose a bootstrap method that determines their bilinear relations from the effective prepotential and the perturbative contributions, together with the global symmetry of the theory. Applying the method to 5d $ \mathrm{SU}(2) $ gauge theories coupled to $ N_f \leq 8 $ fundamental hypermultiplets, we obtain the bilinear relations of the elliptic and $ q $-Painlev\'e equations. This method can be applied to general 5d supersymmetric field theories beyond $ \mathrm{SU}(2) $ gauge theories.
}

\begin{document}
\preprint{\shortstack[r]{KIAS-Q26008\\USTC-ICTS/PCFT-26-53}}
\maketitle

\section{Introduction}

The Painlev\'e equations arose from the classification of nonlinear second-order ordinary differential equations with the Painlev\'e property \cite{Painleve:1900, Painleve:1902, Gambier:1910}. This classification led to six irreducible Painlev\'e equations I--VI. They can also be understood in terms of isomonodromic deformations of linear systems on the punctured Riemann sphere, in which the \emph{tau-function} plays a central role \cite{Jimbo:1981tov,MR866050}. One advantage of the tau-function formulation is that the nonlinear Painlev\'e equations can be recast as Hirota-type bilinear relations among tau-functions \cite{MR1243688,MR1396258}, thereby providing a useful framework for studying their solutions. Tau-functions and Hirota bilinear relations were first introduced in the study of the Korteweg--de Vries (KdV) equation
\begin{align}
    \frac{\partial^3 u}{\partial x^3} + 6u \frac{\partial u}{\partial x} + \frac{\partial u}{\partial t} = 0 \, ,
\end{align}
which describes waves in shallow water. By introducing the tau-function via $u=2\partial_x^2\log\tau$, the KdV equation is transformed into the bilinear equation
\begin{align}
    \left(D_x^4+D_xD_t\right)(\tau\cdot\tau) = 0 \, ,
\end{align}
where $ D $ is the Hirota derivative operator defined as
\begin{align}
    D_x (f \cdot g) = \left( \partial_{x_1} - \partial_{x_2} \right) f(x_1) g(x_2) \big|_{x_1=x_2=x} \, .
\end{align}
A simple ansatz $\tau=1+e^{k(x-k^2t)+\delta}$ with two arbitrary parameters $k$ and $\delta$ yields the one-soliton solution of the KdV equation,
\begin{align}
    u(x,t)=\frac{k^2}{2} \operatorname{sech}^2\left(\frac{k(x-k^2t)+\delta}{2}\right) \, .
\end{align}
Multi-soliton solutions can be constructed systematically by extending the tau-function to a finite sum of exponential functions \cite{Hirota:1971zz}. This example illustrates that bilinear relations are not merely alternative expressions of nonlinear equations, but also provide a direct method for constructing their exact solutions. Hirota's bilinear method has proven to be an effective and systematic way to find multi-soliton solutions of a broad class of integrable systems, including the sine-Gordon, Toda lattice, and Kadomtsev–Petviashvili (KP) equations \cite{Hirota:1972, Hirota:1973, Date:1981, Jimbo:1983if}. This motivates the direct study of the bilinear relations satisfied by Painlev\'e tau-functions.

In recent years, the study of Painlev\'e equations has been connected to 4d $ \mathcal{N}=2 $ supersymmetric gauge theories \cite{Gamayun:2012ma, Gamayun:2013auu, Bonelli:2016qwg}. It was proposed that the tau-function of the Painlev\'e VI equation can be expressed as a sum of four-point Virasoro conformal blocks at central charge $ c=1 $ \cite{Gamayun:2012ma}. By the AGT correspondence \cite{Alday:2009aq}, these conformal blocks agree, up to standard prefactors, with Nekrasov instanton partition functions of the 4d $\mathrm{SU}(2)$ theory with four fundamental hypermultiplets in the self-dual $\Omega$-background \cite{Nekrasov:2002qd, Nekrasov:2003rj}. The corresponding tau-function can therefore be identified with a Nekrasov--Okounkov dual partition function. This relation has been extended to the other Painlev\'e equations, leading to the Painlev\'e/gauge correspondence for rank-one 4d $\mathcal N=2$ theories \cite{Gamayun:2013auu,Bonelli:2016qwg}. See \cite{Gavrylenko:2016zlf,Gavrylenko:2017lqz,Grassi:2018spf,Bonelli:2019boe,Bonelli:2019yjd,Bonelli:2021rrg,Francois:2023trm} for further development.

The Painlev\'e equations also admit discrete generalizations, in which the differential equation is uplifted to a difference equation \cite{Grammaticos:1991zz, Ramani:1991zz, Jimbo:1995}. The discrete Painlev\'e equations are organized geometrically by Sakai's classification \cite{Sakai:2001,Yamada:2020hqd,Mizoguchi:2002kg} in terms of rational surfaces obtained by blowing up nine points of $ \mathbb{P}^2 $. Each equation is characterized by its symmetry type and by the class of its discrete time evolution: elliptic difference, $ q $-difference and additive difference. As a generalization of the correspondence between differential Painlev\'e equations and 4d $ \mathcal{N}=2 $ theories, the class of $ q $-Painlev\'e systems is related to 5d $\mathcal N=1$ $\mathrm{SU}(2)$ gauge theories coupled to $ N_f<8 $ fundamental hypermultiplets, which are low-energy descriptions of 5d $ E_n $ SCFTs with $ n=N_f+1 $ \cite{Seiberg:1996bd, Morrison:1996xf}, while the elliptic Painlev\'e equation is associated with the 5d $ \mathrm{SU}(2) $ gauge theory coupled to $ N_f=8 $ fundamental hypermultiplets, obtained from the circle compactification of the 6d rank-one E-string theory \cite{Mizoguchi:2002kg}. For the cases with $ N_f \leq 4 $, the expressions for the tau-functions of the associated $ q $-Painlev\'e system in terms of the conformal blocks of the $ q $-deformed Virasoro algebra and their bilinear relations have been established \cite{Bershtein:2016aef,Jimbo:2017ael,Matsuhira:2018qtx}. Through the five-dimensional AGT correspondence \cite{Awata:2009ur}, these tau-functions can be written in terms of Nekrasov instanton partition functions of the 5d theories in the self-dual $ \Omega $-background. See also \cite{Bershtein:2014yia, Bershtein:2018zcz, Shchechkin:2020ryb, Gavrylenko:2020gjb,Jeong:2020uxz, Nekrasov:2020qcq, Stoyan:2025qya} for the connections to the blowup equation \cite{Nakajima:2003pg, Nakajima:2005fg, Gottsche:2006bm} and derivations for specific cases using it. Related realizations of tau-functions based on the matrix models, grand-canonical topological string partition functions and the quantum mirror curves have also been proposed \cite{Mironov:2017sqp, Bonelli:2017gdk, Nosaka:2020tyv,Bonelli:2022dse,Mironov:2022xew,Gavrylenko:2023ewx,Moriyama:2023mjx}.

Although several case-by-case derivations for the tau-functions and their bilinear relations are known, a uniform gauge-theoretic construction applicable to a broader class of 5d theories is still lacking. In this paper, we propose a bootstrap method to obtain the bilinear relations for general 5d $ \mathcal{N}=1 $ field theories. Application to 5d $ \mathrm{SU}(2) $ gauge theory yields the bilinear relations associated with the $ q $-difference and elliptic Painlev\'e equations. We define tau-functions as the discrete Fourier transformation of the 5d instanton partition function on $ \mathbb{C}^2 \times S^1 $ in the self-dual $ \Omega $-background $ (\epsilon_1,\epsilon_2)=(\epsilon,-\epsilon) $. More precisely, motivated by the blowup formula \cite{Huang:2017mis, Kim:2020hhh} for the 5d field theory, the tau-function is defined as the grand-canonical partition function on $ \hat{\mathbb{C}}^2 \times S^1 $ in the limit $ (\epsilon_1,\epsilon_2) \to (\epsilon,0) $, where $ \hat{\mathbb{C}}^2 $ is obtained by blowing up the origin of $ \mathbb{C}^2 $ to a $ \mathbb{P}^1 $. Since $ \hat{\mathbb{C}}^2 $ admits magnetic fluxes supported on the $ \mathbb{P}^1 $, the tau-function is labelled by the fractional part $ r $ of the dynamical magnetic flux associated with the gauge symmetry $ G $ of the 5d theory, together with a background magnetic flux $ B $ associated with the flavor symmetry:
\begin{align}
    \tauftn{r}{B}(s,\phi,m,\epsilon) = \sum_{n\in\mathbb Z^{\rank G}+r}s^n\, \hat Z(\phi+n\epsilon,m+B\epsilon;\epsilon)\, .
\end{align}
Here, $ \phi $ denotes the Coulomb moduli, $ m $ denotes the collection of mass parameters including the instanton counting parameters, and $ s $ is the fugacity for the dynamical magnetic flux. The partition function $ \hat{Z} $ contains the classical and perturbative contributions, as well as the instanton partition function in the self-dual $ \Omega $-background. We then consider a bilinear relation ansatz of the form 
\begin{align}
    \Lambda_1\tauftn{r_1}{B_1}\tauftn{r_1'}{B_1'}+\Lambda_2\tauftn{r_2}{B_2}\tauftn{r_2'}{B_2'}+\Lambda_3\tauftn{r_3}{B_3}\tauftn{r_3'}{B_3'}=0\, ,
\end{align} 
where we require the coefficients $\Lambda_i$ to depend on $m$ and $\epsilon$, but not on the Coulomb moduli $\phi$. Our main proposal is that the fluxes allowed by the blowup equations \cite{Huang:2017mis,Kim:2019uqw, Gu:2019pqj,Kim:2020hhh} and the classical and perturbative parts of the partition function $ \hat{Z} $ restrict the possible flux assignments $\{r_i,r_i',B_i,B_i'\}$ and fix the coefficients $\Lambda_i$ up to an overall factor. By further using the global symmetry data or the spinning spectrum of BPS states carrying the minimal nonzero gauge charges, a set of bilinear relations can be bootstrapped. Applying this method to 5d $ \mathrm{SU}(2) $ gauge theories, we obtain bilinear relations of the elliptic and $ q $-Painlev\'e equations.

The remainder of this paper is organized as follows. In section~\ref{sec:Painleve-gauge}, we review the correspondence between differential Painlev\'e equations and 4d $ \mathcal{N}=2 $ theories, together with its generalization relating the discrete Painlev\'e hierarchy and 5d $ \mathcal{N}=1 $ gauge theories. In section~\ref{sec:5dbootstrap}, we define tau-functions associated with 5d field theories from the partition function on $ \mathbb{C}^2 \times S^1 $ and the blowup equation, and develop the bootstrap algorithm for the bilinear relations they satisfy. We also give a minimal example that illustrates the bootstrap algorithm. Section~\ref{sec:applications} presents the results of the bootstrap program for 5d $ \mathrm{SU}(2) $ gauge theories coupled to $ N_f $ fundamental hypermultiplets, which correspond to the elliptic and $ q $-Painlev\'e equations. We then conclude with some future directions in section~\ref{sec:conclusion}. Appendix~\ref{app:lemma} discusses a lemma used in our analysis of the elliptic Painlev\'e equations, and Appendix~\ref{app:bilinear} contains a brief outline of how the bilinear relations are related to the Painlev\'e equations. Appendix~\ref{app:adhm} summarizes the computation of the instanton partition functions of 5d $ \mathrm{SU}(2) $ gauge theories, which we use to verify the bilinear relations obtained from the bootstrap program.

\section{Painlev\'e-gauge correspondence} \label{sec:Painleve-gauge}

In this section, we review the correspondence between Painlev\'e equations and supersymmetric gauge theories. We first recall the realization of the differential Painlev\'e equations through isomonodromic deformations and their relation to 4d $ \mathcal{N}=2 $ rank-one theories. We then turn to the discrete Painlev\'e hierarchy and its correspondence with 5d $ \mathcal{N}=1 $ theories.

\subsection{Painlev\'e equations and 4d \texorpdfstring{$ \mathcal{N}=2 $}{N=2} theories}

The Painlev\'e equations first appeared in the theory of nonlinear ordinary differential equations (ODEs). A singularity of a solution of an ODE is called movable if its location depends on the initial conditions. An ODE is said to have the Painlev\'e property if none of its solutions has movable branch points, although the solutions may have movable poles, movable single-valued essential singularities, or fixed singularities. One of the central problems in the theory of ODEs is the classification of ODEs admitting the Painlev\'e property. Painlev\'e and Gambier \cite{Painleve:1900, Painleve:1902, Gambier:1910} established that there are only 50 types of second-order ODEs admitting the Painlev\'e property of the form
\begin{align}
    \frac{d^2 q}{dt^2} = F\left( \frac{dq}{dt}, q, t \right) \, ,
\end{align}
where $ F $ is a rational function of $ dq/dt $ and $ q $ and locally analytic in $ t $, up to a M\"obius transformation
\begin{align}
    \tilde{q} = \frac{a(t)q+b(t)}{c(t)q + d(t)} \, , \quad
    \tilde{t} = \varphi(t) \, ,
\end{align}
for locally analytic functions $ a(t) $, $ b(t) $, $ c(t) $, $ d(t) $, and $ \varphi(t) $. Among the 50 classes of equations, 44 can be solved in terms of known functions such as elliptic functions. The remaining six equations are novel nonlinear ODEs known as the Painlev\'e equations. Their explicit forms are given by
\begin{alignat}{2} \label{eq:Painleve}
    &\mathrm{P}_\mathrm{I} &&\qquad \frac{d^2q}{dt^2} = 6q^2 + t \, , \\
    &\mathrm{P}_\mathrm{II} &&\qquad \frac{d^2q}{dt^2} = 2q^3 + tq + \alpha \, , \nonumber \\
    &\mathrm{P}_\mathrm{III} &&\qquad \frac{d^2q}{dt^2} = \frac{1}{q} \left( \frac{dq}{dt}\right)^2 - \frac{1}{t} \frac{dq}{dt} + \frac{1}{t} \left( \alpha q^2 + \beta \right) + \gamma q^3 + \frac{\delta}{q} \, , \nonumber \\
    &\mathrm{P}_\mathrm{IV} &&\qquad \frac{d^2q}{dt^2} = \frac{1}{2q} \left( \frac{dq}{dt}\right)^2 + \frac{3}{2} q^3 + 4tq^2 + 2(t^2-\alpha)q + \frac{\beta}{q} \, , \nonumber \\
    &\mathrm{P}_\mathrm{V} &&\qquad \frac{d^2q}{dt^2} = \left( \frac{1}{2q} + \frac{1}{q-1} \right) \left( \frac{dq}{dt}\right)^2 - \frac{1}{t} \frac{dq}{dt} + \frac{(q-1)^2}{t^2} \left( \alpha q + \frac{\beta}{q} \right) + \gamma \frac{q}{t} + \delta \frac{q(q+1)}{q-1} \, , \nonumber \\
    &\mathrm{P}_\mathrm{VI} &&\qquad \frac{d^2q}{dt^2} = \frac{1}{2} \left( \frac{1}{q} + \frac{1}{q-1} + \frac{1}{q-t} \right) \left( \frac{dq}{dt}\right)^2 - \left( \frac{1}{t} + \frac{1}{t-1} + \frac{1}{q-t} \right) \frac{dq}{dt} \nonumber \\
    & && \qquad \qquad + \frac{q(q-1)(q-t)}{t^2(t-1)^2} \left( \alpha + \beta\frac{t}{q^2} + \gamma \frac{t-1}{(q-1)^2} + \delta \frac{t(t-1)}{(q-t)^2} \right) \, , \nonumber
\end{alignat}
where $ \alpha $, $ \beta $, $ \gamma $, $ \delta $ are complex-valued constants.\footnote{\label{footnote:PIII'}It is also common to use an alternative form $ \mathrm{P}_{\mathrm{III}}' $ of the Painlev\'e III equation given by \begin{align*} \frac{d^2q}{dt^2} = \frac{1}{q} \left( \frac{dq}{dt}\right)^2 - \frac{1}{t} \frac{dq}{dt} + \frac{q^2}{4t^2} (\gamma q + \alpha) + \frac{\beta}{4t} + \frac{\delta}{4q} \, . \end{align*}} The solutions of these equations are called the Painlev\'e transcendents. The Painlev\'e equations I-V can be obtained through successive limiting procedures from Painlev\'e VI, known as \emph{coalescence}. For example, reparametrizing $ q\to \epsilon q + \epsilon^{-5} $, $ t \to \epsilon^2t - 6\epsilon^{-10} $, $ \alpha \to 4\epsilon^{-15} $ for Painlev\'e II and taking the $ \epsilon \to 0 $ limit yields the Painlev\'e I equation. The coalescence diagram for the six Painlev\'e equations is shown as follows:
\begin{equation}\label{eq:coalescence}
    \begin{tikzcd}
        \mathrm{P}_{\mathrm{VI}} \arrow{r} & \mathrm{P}_{\mathrm{V}} \arrow{r} \arrow{rd} & \mathrm{P}_{\mathrm{III}_1} \arrow{r} \arrow{rd} & \mathrm{P}_{\mathrm{III}_2} \arrow{r} \arrow{rd} & \mathrm{P}_{\mathrm{III}_3} \\
        & & \mathrm{P}_{\mathrm{IV}} \arrow{r} & \mathrm{P}_{\mathrm{II}} \arrow{r} & \mathrm{P}_{\mathrm{I}}
    \end{tikzcd}
\end{equation}
Here, the Painlev\'e III equation is further classified into three cases: $ \mathrm{P}_{\mathrm{III}_1} $ corresponds to $ \gamma\delta \neq 0 $; $ \mathrm{P}_{\mathrm{III}_2} $ corresponds to $ \gamma=0 $, $ \alpha \delta \neq 0 $; and $ \mathrm{P}_{\mathrm{III}_3} $ corresponds to $ \gamma=\delta=0 $, $ \alpha\beta \neq 0 $. The Painlev\'e VI equation has four parameters, and the number of free parameters decreases by one at each limiting step corresponding to an arrow in the diagram. For example, although the Painlev\'e III equation appears to have four free parameters $ (\alpha,\beta,\gamma,\delta) $, when $ \gamma\delta \neq 0 $, one can set $ \gamma=-\delta=\mathrm{constant} $ without loss of generality by rescaling $ q $ and $ t $.

The Painlev\'e equations can be realized through the framework of \emph{isomonodromic deformations}. This problem concerns the monodromy of solutions to a system of linear ordinary differential equations
\begin{align}\label{eq:fuchsian}
    \frac{dY(z)}{dz} = A(z) Y(z) \, , \quad
    A(z) = \sum_{\nu=1}^n \sum_{k=0}^{r_\nu} \frac{A_{\nu,-k}}{(z-z_\nu)^{k+1}} \, ,
\end{align}
where $ z $ is a coordinate on the $ n $-punctured sphere $ \mathbb{P}^1\setminus \{z_1,\cdots,z_n\} $, $ A(z) \in \mathfrak{sl}(2,\mathbb{C}) $ is a matrix-valued function, and the non-negative integer $ r_\nu $ is called the Poincar\'e rank of the singularity. If $ r_\nu=0 $, the linear system has a regular singularity at $ z=z_\nu $, whereas singularities with $ r_\nu>0 $ are called irregular. A linear system of this type without irregular singularities is called a Fuchsian system. For a given system of differential equations \eqref{eq:fuchsian}, one can determine the monodromy of the solution $ Y(z) $ around the singularities $ z_\nu $. Conversely, fixing the monodromy data does not single out a unique linear system but rather a family of systems sharing the same monodromy. A deformation of $ A(z) $ that preserves the monodromy data is called an isomonodromic deformation. Such a deformation is governed by a system of equations for $ Y $ given by \cite{Jimbo:1981tov, Jimbo:1981-2}
\begin{align}\label{eq:isomonodomic}
    \frac{\partial Y}{\partial z} = A(z,t) Y \, , \quad
    \frac{\partial Y}{\partial t_i} = B_i(z,t) Y \, ,
\end{align}
where $ t=(t_i) $ denotes the deformation parameters. The matrix-valued functions $ A $ and $ B_i $ form a Lax pair, and $ B_i(z,t) $ can be determined from $ A(z,t) $. The commutativity condition $ \partial_z \partial_{t_i} Y = \partial_{t_i} \partial_z Y $ implies that $ A $ and $ B_i $ satisfy the following compatibility equations: 
\begin{align}
    \frac{\partial A}{\partial t_i} - \frac{\partial B_i}{\partial z} + [A, B_i] = 0 \, , \quad
    \frac{\partial B_i}{\partial t_j} - \frac{\partial B_j}{\partial t_i} + [B_i, B_j] = 0 \, .
\end{align}
For suitable choices of the singularity structure and deformation parameters, these compatibility conditions give rise to the Painlev\'e equations.

For the Painlev\'e VI equation, which sits at the top of the coalescence cascade, \eqref{eq:fuchsian} corresponds to a Fuchsian system with four regular singularities. By using a M\"obius transformation, the singularities can be placed at $ 0 $, $ 1 $, $ t $ and $ \infty $. The Lax pair $ A $ and $ B $ can be written as
\begin{align}
    A(z,t) = \frac{A_0}{z} + \frac{A_1}{z-1} + \frac{A_t}{z-t} \, , \quad
    B(z,t) = -\frac{A_t}{z-t} \, ,
\end{align}
where the position $ t $ of the singularity serves as the deformation parameter, while $ A_0 $, $ A_1 $, and $ A_t $ are the elements of $ \mathfrak{sl}(2,\mathbb{C}) $ satisfying the compatibility conditions
\begin{align}\label{eq:VI-compat}
    \frac{d A_0}{dt} + \frac{[A_0,A_t]}{t} = 0 \, , \quad
    \frac{d A_1}{dt} - \frac{[A_1,A_t]}{1-t} = 0 \, , \quad
    \frac{d A_t}{dt} - \frac{[A_0,A_t]}{t} + \frac{[A_1,A_t]}{1-t} = 0 \, .
\end{align}
The explicit form of $ A $ is given in \cite{Jimbo:1981-2}, and \eqref{eq:VI-compat} yields the Painlev\'e VI equation. The compatibility conditions of the Fuchsian system are called the Schlesinger equations, which can be formulated as the Hamiltonian system
\begin{align}\label{eq:PVI-H}
    H_t = \frac{\Tr A_0 A_t}{t} + \frac{\Tr A_1 A_t}{t-1} \, , \quad
    \frac{\partial A_i}{\partial t} = \{ A_i, H_t \}_{\mathrm{PB}} \, ,
\end{align}
with the Poisson bracket
\begin{align}
    \{ (A_i)_{ab}, (A_j)_{cd} \}_{\mathrm{PB}}
    = \delta_{ij} \left( \delta_{ad} (A_i)_{cb} - \delta_{bc} (A_j)_{ad} \right) \, .
\end{align}
The other Painlev\'e equations can be treated in a similar manner using isomonodromic deformation problems involving irregular singularities. In this formulation, the coalescence process can be understood as the collision of punctures. For further mathematical details, we refer the reader to \cite{Jimbo:1981tov, Jimbo:1981-2, Its:2018}.

This isomonodromic construction of the Painlev\'e equations is closely related to their correspondence with 4d $ \mathcal{N}=2 $ supersymmetric field theories. First, the punctured Riemann sphere appearing in the isomonodromic deformation problem also arises in the class $ \mathcal{S} $ construction of the 4d $ \mathcal{N}=2 $ theories. For example, compactifying the 6d $ \mathcal{N}=(2,0) $ $ A_1 $ theory on a Riemann sphere with four regular punctures, corresponding to the setup of the Painlev\'e VI isomonodromic deformation problem, gives a 4d $ \mathcal{N}=2 $ $ \mathrm{SU}(2) $ gauge theory coupled to four fundamental hypermultiplets \cite{Gaiotto:2009we, Gaiotto:2009hg}. Second, the spectral curve defined by $ y^2 = \frac{1}{2} \Tr A^2 $ in the isomonodromic deformation problem is identified with the Seiberg-Witten curve of the corresponding 4d $ \mathcal{N}=2 $ theory. From this perspective, the six Painlev\'e equations can be related to 4d $ \mathcal{N}=2 $ rank-1 theories as summarized in Table~\ref{table:painleve-4d}. The Painlev\'e equations in the first line of the coalescence cascade \eqref{eq:coalescence} correspond to 4d $ \mathrm{SU}(2) $ gauge theories with $ N_f $ fundamental hypermultiplets, whereas those in the second line correspond to the $ H_0=(A_1,A_2) $, $ H_1=(A_1,A_3) $, and $ H_2=(A_1,D_4) $ Argyres-Douglas (AD) theories \cite{Argyres:1995xn, Argyres:1995jj}. Under this identification, the free parameters $ \alpha $, $ \beta $, $ \gamma $, $ \delta $ of the Painlev\'e equations correspond to mass parameters of the 4d theory, and the Painlev\'e time $ t $ is identified with the gauge coupling for $ \mathrm{SU}(2) $ gauge theories or with the relevant deformation coupling for Argyres-Douglas theories. More details on the correspondence between Painlev\'e equations and 4d $ \mathcal{N}=2 $ theories can be found in \cite{Bonelli:2016qwg}.

\begin{table}
    \centering
    \begin{tabular}{cccccccc}
        $ \mathrm{P}_{\mathrm{VI}} $ & $ \mathrm{P}_{\mathrm{V}} $ & $ \mathrm{P}_{\mathrm{III}_1} $ & $ \mathrm{P}_{\mathrm{III}_2} $ & $ \mathrm{P}_{\mathrm{III}_3} $ & $ \mathrm{P}_{\mathrm{IV}} $ & $ \mathrm{P}_{\mathrm{II}} $ & $ \mathrm{P}_{\mathrm{I}} $ \\ \hline
        $ N_f=4 $ & $ N_f=3 $ & $ N_f=2 $ & $ N_f=1 $ & $ N_f=0 $ & $ H_2 $ & $ H_1 $ & $ H_0 $
    \end{tabular}
    \caption{Correspondence between Painlev\'e equations and 4d $ \mathcal{N}=2 $ theories. The Painlev\'e III, V and VI equations correspond to $ \mathrm{SU}(2) $ gauge theories coupled to $ N_f $ hypermultiplets in the fundamental representation, while the Painlev\'e I, II, and IV equations are associated with the $ H_0 $, $ H_1 $, and $ H_2 $ Argyres-Douglas theories, respectively. The number of mass parameters in each 4d theory coincides with the number of free parameters in the associated Painlev\'e equation.} \label{table:painleve-4d}
\end{table}

We now introduce the notion of the \emph{isomonodromic tau-function} of the Painlev\'e equations \cite{Jimbo:1981tov}. It is defined as the generating function of the Hamiltonian $ H(t) $ appearing in the isomonodromic deformation problem as
\begin{align}
    \frac{d}{dt} \log \tau(t) = H(t) \, .
\end{align}
For Painlev\'e VI, the isomonodromic Hamiltonian is given by \eqref{eq:PVI-H}. A remarkable feature of the Painlev\'e VI tau-function is that it can be expressed as a Fourier series of 4-point Virasoro conformal blocks at central charge $ c=1 $ \cite{Gamayun:2012ma, Iorgov:2014vla, Bershtein:2014yia}. By the AGT correspondence \cite{Alday:2009aq}, this conformal block is identified with the Nekrasov partition function of the 4d $ \mathrm{SU}(2) $ gauge theory with four fundamental hypermultiplets in the self-dual $ \Omega $-background $ \epsilon_1=-\epsilon_2=\epsilon $. Similar expressions for the other Painlev\'e equations can be obtained by taking the appropriate coalescence limits, and the Painlev\'e tau-function can be connected to the partition function of the corresponding 4d $ \mathcal{N}=2 $ rank-1 theory as \cite{Gamayun:2013auu, Bonelli:2016qwg}
\begin{align}\label{eq:dualpf}
    \tau(t | \rho, \sigma, m) = \sum_{n\in \mathbb{Z}} e^{in\rho} Z(t, \sigma+n, m) \, ,
\end{align}
where $ Z $ is the Nekrasov(-like) partition function in the self-dual $ \Omega $-background, including the classical, one-loop and instanton contributions. Here, $ m $ denotes the mass parameters of the 4d theory, and $ \sigma=a/\epsilon $ is the dimensionless Coulomb branch parameter. For $ \mathrm{SU}(2) $ gauge theories, $ t $ is identified with the exponentiated gauge coupling, while it corresponds to the relevant deformation coupling in the case of Argyres-Douglas theories \cite{Bonelli:2016qwg}. The summation over integer shifts of the Coulomb modulus can be regarded as a discrete Fourier transformation and is referred to as the (Nekrasov-Okounkov) \emph{dual partition function} \cite{Nekrasov:2003rj}. The Fourier dual variable $ \rho $ conjugate to $ \sigma $ can be interpreted as the dual period $ a_D/\epsilon $ of the 4d theory in the Seiberg-Witten limit $ \epsilon\to 0 $. From the perspective of the Painlev\'e equation, the mass parameters $ m $ correspond to the parameters $ (\alpha,\beta,\gamma,\delta) $ in \eqref{eq:Painleve}, while $ (\sigma, \rho) $ play the role of integration constants.

For example, the tau-function of the Painlev\'e $ \text{III}_{\text{3}} $ near $ t=0 $ can be written as
\begin{align}\label{eq:PIII-tau}
    \tau_{\text{III}_3}(t|s,\sigma) = \sum_{n\in \mathbb{Z}} s^n Z_{\mathrm{SU}(2)}(t, \sigma+n) = \sum_{n\in \mathbb{Z}} s^n t^{(\sigma+n)^2} C(\sigma+n) B(t, \sigma+n) \, ,
\end{align}
where $ Z_{\mathrm{SU}(2)} $ is the Nekrasov partition function of the 4d pure $ \mathrm{SU}(2) $ gauge theory, $ s=e^{i\rho} $ and
\begin{align}
    C(\sigma) &= \frac{1}{G(1+2\sigma) G(1-2\sigma)} \, , \\
    B(t, \sigma) &= \sum_{Y_1, Y_2} t^{|Y_1|+|Y_2|} \prod_{i=1}^2 \prod_{\square \in Y_i} \frac{1}{\prod_{j=1}^2 (\sigma_i-\sigma_j + h_i(\square) + v_j(\square)+1)^2} \\
    &= 1 + \frac{t}{2\sigma^2} + \frac{(1+8\sigma^2)t^2}{4\sigma^2(1-4\sigma^2)^2} + \mathcal{O}(t^3) \, , \nonumber
\end{align}
for the Barnes G-function $ G(z) $ satisfying $ G(1+z) = \Gamma(z) G(z) $, and Young diagrams $ Y_1 $, $ Y_2 $. Here, $ \sigma_1=-\sigma_2=\sigma $, $ \square\in Y_i $ is a box in the Young diagram $ Y_i $, $ h_i(\square) $ is the arm length of the box $ \square $ in $ Y_i $, and $ v_j(\square) $ is the leg length in $ Y_j $. The factors $ C(\sigma) $ and $ B(t,\sigma) $ are respectively the 1-loop vector multiplet contribution and the instanton partition function of the pure $ \mathrm{SU}(2) $ theory, while the factor $ t^{\sigma^2} $ is the classical contribution.

Another interesting feature of the tau-function is that it provides a \emph{bilinearization} of the Painlev\'e equations. The Painlev\'e equations admit a group of \emph{B\"acklund transformations} that act on the solutions $ q(t) $ and the parameters $ (\alpha,\beta,\gamma,\delta) $, thereby mapping solutions with one set of parameters to those with another set. Applying B\"acklund transformations to the tau-function produces a family of tau-functions $ \{\tau_k\} $. The non-linear Painlev\'e equations can then be replaced by bilinear differential relations among the tau-functions of the form
\begin{align}
    P_{kl}(D_{[\log t]}, t, m)(\tau_k \cdot \tau_l) = 0 \, ,
\end{align}
where $ P_{kl} $ is a polynomial in the Hirota derivative $ D_{[\log t]} $ defined as
\begin{align}
    f(e^{\alpha}t) g(e^{-\alpha}t) = \sum_{n=0}^\infty D_{[\log t]}^n(f(t),g(t)) \frac{\alpha^n}{n!} \, .
\end{align}
Such relations are referred to as \emph{Hirota bilinear relations}. The solution $ q(t) $ of the Painlev\'e equation can then be recovered as a ratio of tau-functions. For example, the Painlev\'e $ \text{III}_3 $ tau-function \eqref{eq:PIII-tau} satisfies the bilinear relations \cite{Bershtein:2016uov}
\begin{align}\label{eq:PIII3-bilinear}
    \left\{ \begin{array}{l}
            D^2_{[\log t]}(\tau_1, \tau_1) + 2\sqrt{t} \tau_2^2 = 0 \\[1ex]
            D^2_{[\log t]}(\tau_2, \tau_2) + 2\sqrt{t} \tau_1^2 = 0 
    \end{array}\right. \quad \Leftrightarrow \quad
    \left\{ \begin{array}{l}
            \tau_1 \mathcal{D}^2\tau_1 - (\mathcal{D}\tau_1)^2 + \sqrt{t} \tau_2^2 = 0 \\[1ex]
            \tau_2 \mathcal{D}^2\tau_2 - (\mathcal{D}\tau_2)^2 + \sqrt{t} \tau_1^2 = 0 
        \end{array}\right.
\end{align}
where $ \mathcal{D} = t \frac{d}{dt} $ and
\begin{align}
    \tau_1(t|s,\sigma) = \tau_{\text{III}_3}(t|s,\sigma) \, , \quad
    \tau_2(t|s,\sigma) = \sqrt{s} \tau_{\text{III}_3}(t|s,\sigma+1/2) \, .
\end{align}
Then $ q(t) = \sqrt{t} \tau_1^2 / \tau_2^2 $ satisfies the $ \mathrm{P}_{\mathrm{III}}' $ equation given in footnote~\ref{footnote:PIII'}, with $ \gamma=\delta=0 $ and $ \alpha=-\beta=8 $. In this way, the solutions of the Painlev\'e equations are expressed in terms of the partition functions of the corresponding 4d $ \mathcal{N}=2 $ theories.

\subsection{Generalization to discrete Painlev\'e equations and 5d \texorpdfstring{$ \mathcal{N}=1 $}{N=1} theories} \label{sec:discreteP-5d}

The correspondence between Painlev\'e equations and 4d $ \mathcal{N}=2 $ theories admits a natural lift to a correspondence between discrete generalizations of Painlev\'e equations and 5d $ \mathcal{N}=1 $ theories. Just as the ordinary Painlev\'e equations arise from isomonodromic deformations of linear differential equations, the $ q $-difference Painlev\'e equations are derived from connection preserving deformations of linear $ q $-difference equations \cite{Jimbo:1995, Murata:2008, Yamada:2008}. Analogous to \eqref{eq:fuchsian}, one considers a linear $ q $-difference system
\begin{align}
    Y(qz) = A(z) Y(z) \, ,
\end{align}
where $ A(z) $ is a rational matrix-valued function. The role of monodromy in the differential case is played by the connection matrix which intertwines two solutions of the linear system near $ z=0 $ and $ z=\infty $. A deformation of the linear system that preserves the connection matrix is called a \emph{connection preserving deformation} or a \emph{q-isomonodromic deformation}. As in \eqref{eq:isomonodomic}, such a deformation implies that $ Y $ satisfies
\begin{align}\label{eq:monodromy-preserving}
    Y(qz, t) = A(z, t) Y(z, t) \, , \quad
    Y(z, qt) = B(z, t) Y(z, t) \, ,
\end{align}
where $ t $ is a deformation parameter and $ B $ is a matrix-valued rational function \cite{Jimbo:1995}. The compatibility condition for the system \eqref{eq:monodromy-preserving} is
\begin{align}\label{eq:monodromy-compatibility}
    A(z, qt) B(z, t) = B(qz, t) A(z, t) \, .
\end{align}
This condition yields a $ q $-difference equation that generalizes the Painlev\'e equation. For example, \cite{Jimbo:1995} considers a particular two-dimensional linear system and shows that the compatibility condition \eqref{eq:monodromy-compatibility} is equivalent to the nonlinear $ q $-difference equations for $ f=f(t) $ and $ g=g(t) $, which parameterize $A(z,t)$,  given by
\begin{align}\label{eq:qVI}
    \frac{f\overline{f}}{a_3a_4} = \frac{(\overline{g}-tb_1) (\overline{g}-tb_2)}{(\overline{g}-b_3) (\overline{g}-b_4)} \, , \quad
    \frac{g\overline{g}}{b_3b_4} = \frac{(f-ta_1) (f-ta_2)}{(f-a_3) (f-a_4)} \, , 
\end{align}
where $ \overline{f}=f(qt) $ denotes the discrete time evolution and $ a_i $ and $ b_i $ are parameters satisfying
\begin{align}
    \frac{b_1b_2}{b_3b_4} = q \frac{a_1a_2}{a_3a_4} \, .
\end{align}
Due to the rescaling freedom in $ f $, $ g $, and $ t $, there are only four independent parameters, and the continuous limit $ q\to 1 $ reproduces the differential Painlev\'e VI equation. This $ q $-difference system is referred to as the $ q $-Painlev\'e VI equation.

The $q$-Painlev\'e VI equation is one example of a broader family of discrete Painlev\'e equations. A characteristic property of these equations is \emph{singularity confinement} \cite{Grammaticos:1991zz}, which can be regarded as a discrete analogue of the Painlev\'e property: a singularity produced by special initial data disappears after finitely many iterations and the dependence on the initial data is recovered. These equations are classified geometrically in terms of rational surfaces obtained by blowing up nine points of the projective plane $ \mathbb{P}^2 $, or equivalently eight points of $ \mathbb{P}^1 \times \mathbb{P}^1 $, which serve as the spaces of initial conditions for the discrete time evolution \cite{Sakai:2001}. For a generic configuration of the nine points in $ \mathbb{P}^2 $, the space of initial conditions admits the largest possible symmetry, of type $ E_8^{(1)} $. The extended affine Weyl group of $ E_8^{(1)} $ acts on the parameters and solutions as a symmetry of the equation, namely as B\"acklund transformations, while the time evolution is realized as a translation in the $ E_8^{(1)} $ affine root lattice. When the configuration of blown-up points \emph{degenerates}, part of the symmetry is reduced, leading to discrete Painlev\'e equations with smaller symmetry types. In this way, each discrete Painlev\'e equation is naturally characterized by its affine Weyl symmetry type.\footnote{In the literature, discrete Painlev\'e equations are often labelled by both their \emph{surface} type and their \emph{symmetry} type. Throughout this paper, we label each equation only by its symmetry type.}

This classification further splits the discrete Painlev\'e equations into three classes according to the type of discrete time evolution: \emph{elliptic}, \emph{multiplicative} (or $ q $-difference), and \emph{additive}. The elliptic Painlev\'e equation is the most general member of the hierarchy, with symmetry type $ E_8^{(1)} $; its time evolution corresponds to a translation on a smooth elliptic curve passing through the nine blown-up points of $ \mathbb{P}^2 $. Successive degeneration of the elliptic Painlev\'e equation yields multiplicative $ q $-difference Painlev\'e equations and additive difference Painlev\'e equations, whose time evolutions are $ t\mapsto qt $ and $ t \mapsto t+\delta $, respectively. We present the explicit expressions of the elliptic and $ q $-Painlev\'e equations in Appendix~\ref{app:PainleveExpression}. Finally, the differential Painlev\'e equations \eqref{eq:Painleve} are obtained in the continuous limits $ q\to 1 $ and $ \delta\to 0 $ of $ q $-difference and additive difference equations, respectively. The degeneration pattern of the Painlev\'e hierarchy is summarized in Figure~\ref{fig:PainleveClassify}. In the figure, each equation is labelled by its symmetry type and time evolution type, $ e $-, $ q $-, $ d $- or $ c $-, denoting the elliptic, multiplicative, additive difference, and continuous (differential) types, respectively, and the arrows denote possible degenerations \cite{Rains:2013, Dzhamay:2018}. The $ q $-Painlev\'e VI equation \eqref{eq:qVI} and the differential Painlev\'e VI equation in \eqref{eq:Painleve} correspond to $ q\text{-}D_5^{(1)} $ and $ c\text{-}D_4^{(1)} $, respectively, while the coalescence cascade \eqref{eq:coalescence} of the differential Painlev\'e equations corresponds to the eight cases at the bottom right of Figure~\ref{fig:PainleveClassify} labelled by $ c $. The terminal $ q\text{-}A_0^{(1)} $ node in the second row is an exceptional case: the corresponding symmetry does not admit a non-trivial translation that generates a discrete time evolution. This node is included only to display the terminal geometry in the degeneration pattern, rather than to indicate a non-trivial $ q $-Painlev\'e equation. For the same reason, the continuous Painlev\'e equations with $ A_0^{(1)} $ symmetry type, which are Painlev\'e $ \text{III}_3 $ and I, have no additive difference counterparts.

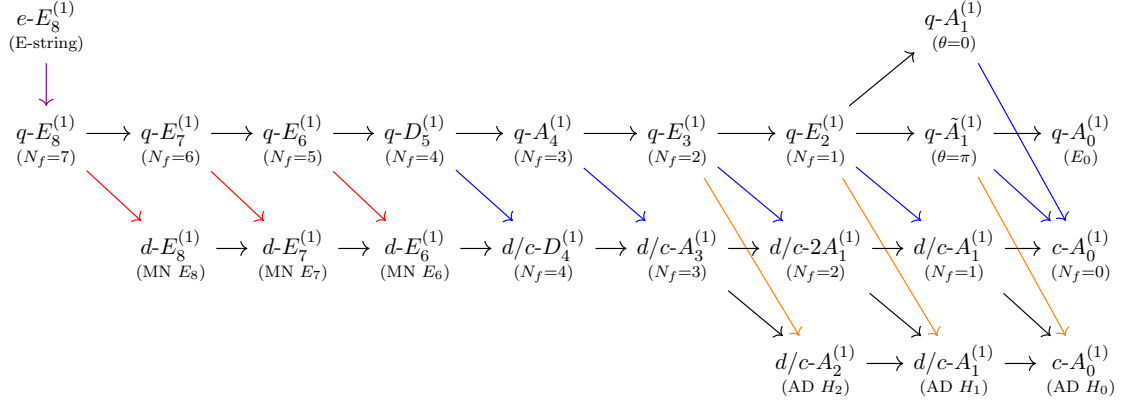
\begin{figure}
    \centering
    \begin{tikzcd}[row sep=1.5em, column sep=1.1em, scale cd=0.78]
        \underset{(\text{E-string})}{e\text{-}E_8^{(1)}} \arrow[d, violet] & & & & & & & \underset{(\theta=0)}{q\text{-}A_1^{(1)}} \arrow[rdd, blue] \\
        \underset{(N_f=7)}{q\text{-}E_8^{(1)}} \arrow[r] \arrow[rd, red] & \underset{(N_f=6)}{q\text{-}E_7^{(1)}} \arrow[r] \arrow[rd, red] & \underset{(N_f=5)}{q\text{-}E_6^{(1)}} \arrow[r] \arrow[rd, red] & \underset{(N_f=4)}{q\text{-}D_5^{(1)}} \arrow[r] \arrow[rd, blue] & \underset{(N_f=3)}{q\text{-}A_4^{(1)}} \arrow[r] \arrow[rd, blue] & \underset{(N_f=2)}{q\text{-}E_3^{(1)}} \arrow[r] \arrow[rd, blue] \arrow[rdd, orange] & \underset{(N_f=1)}{q\text{-}E_2^{(1)}} \arrow[r] \arrow[ru] \arrow[rd, blue] \arrow[rdd, orange] & \underset{(\theta=\pi)}{q\text{-}\tilde{A}_1^{(1)}} \arrow[r] \arrow[rd, blue] \arrow[rdd, orange] & \underset{(E_0)}{q\text{-}A_0^{(1)}} \\
        & \underset{(\text{MN }E_8)}{d\text{-}E_8^{(1)}} \arrow[r] & \underset{(\text{MN }E_7)}{d\text{-}E_7^{(1)}} \arrow[r] & \underset{(\text{MN }E_6)}{d\text{-}E_6^{(1)}} \arrow[r] & \underset{(N_f=4)}{d/c\text{-}D_4^{(1)}} \arrow[r] & \underset{(N_f=3)}{d/c\text{-}A_3^{(1)}} \arrow[r] \arrow[rd] & \underset{(N_f=2)}{d/c\text{-}2A_1^{(1)}} \arrow[r] \arrow[rd] & \underset{(N_f=1)}{d/c\text{-}A_1^{(1)}} \arrow[r] \arrow[rd] & \underset{(N_f=0)}{c\text{-}A_0^{(1)}} \\
        & & & & & & \underset{(\text{AD }H_2)}{d/c\text{-}A_2^{(1)}} \arrow[r] & \underset{(\text{AD }H_1)}{d/c\text{-}A_1^{(1)}} \arrow[r] & \underset{(\text{AD }H_0)}{c\text{-}A_0^{(1)}}
    \end{tikzcd}
    \caption{Classification of Painlev\'e equations by symmetry type and the corresponding supersymmetric field theories. The labels $ e $, $ q $, $ d $ and $ c $ denote elliptic, multiplicative $ q $-difference, additive difference, and continuous Painlev\'e equations, respectively. We write $ E_3=A_2+A_1 $ and $ E_2=A_1+A_1 $ for simplicity. The arrows represent degenerations of Painlev\'e equations, while the colored arrows indicate dimensional reductions on the field theory side.} \label{fig:PainleveClassify}
\end{figure}

The classification of the Painlev\'e equations suggests a natural correspondence between the Painlev\'e equations and supersymmetric field theories. First, the elliptic Painlev\'e equation can be associated with the 6d $ \mathcal{N}=(1,0) $ SCFT with $ E_8 $ global symmetry, known as the E-string theory \cite{Ganor:1996mu, Seiberg:1996vs}. Its circle compactification gives a 5d Kaluza-Klein (KK) theory whose low-energy description is the 5d $ \mathcal{N}=1 $ $ \mathrm{SU}(2) $ gauge theory coupled to $ N_f=8 $ fundamental hypermultiplets. This theory can be geometrically engineered by compactifying M-theory on the local elliptic Calabi-Yau threefold, given by the total space of the canonical bundle over the nine-point blowup of $ \mathbb{P}^2 $. Second, the $ q $-Painlev\'e equations correspond to the family of 5d $ \mathcal{N}=1 $ SCFTs known as $ E_n $ theories \cite{Seiberg:1996bd, Morrison:1996xf}. For $ 2 \leq n \leq 8 $, these theories admit low-energy descriptions as $ \mathrm{SU}(2) $ gauge theories coupled to $ N_f=n-1 $ fundamental hypermultiplets. The mass deformation from $ N_f=8 $ to $ N_f=7 $ corresponds to the degeneration from the elliptic Painlev\'e equation to the $ q $-Painlev\'e equation $ q\text{-}E_8^{(1)} $, denoted by the violet arrow in Figure~\ref{fig:PainleveClassify}. For $ n=1 $, there are two distinct SCFT fixed points denoted by $ E_1 $ and $ \tilde{E}_1 $, whose IR descriptions are pure $ \mathrm{SU}(2) $ gauge theories with discrete theta angles $ \theta=0 $ and $ \theta=\pi $, respectively. This family of 5d SCFTs has $ E_n $ global symmetry, where $ E_5=\mathrm{SO}(10) $, $ E_4=\mathrm{SU}(5) $, $ E_3=\mathrm{SU}(3) \times \mathrm{SU}(2) $, $ E_2=\mathrm{SU}(2)\times \mathrm{U}(1) $, and $ E_1=\mathrm{SU}(2) $, while the global symmetry of the $ \tilde{E}_1 $ theory is $ \mathrm{U}(1) $. There is a natural RG flow from the $ N_f=7 $ theory to the pure $ \mathrm{SU}(2)_{\theta=0,\pi} $ theories induced by mass deformations that decouple the fundamental hypermultiplets. This flow corresponds to the degeneration from the $ q\text{-}E_8^{(1)} $ Painlev\'e equation to the $ q\text{-}A_1^{(1)} $ and $ q\text{-}\tilde{A}_1^{(1)} $ Painlev\'e equations in Figure~\ref{fig:PainleveClassify}. Furthermore, the 5d $ E_0 $ SCFT, which is also called the local $ \mathbb{P}^2 $ theory and has trivial global symmetry, can be obtained by a mass deformation of the $ \tilde{E}_1 $ SCFT. This flow corresponds to the degeneration of the geometry from $ q\text{-}\tilde{A}_1^{(1)} $ to $ q\text{-}A_0^{(1)} $.

Third, the additive and differential Painlev\'e equations in the third and fourth lines of Figure~\ref{fig:PainleveClassify} correspond to the 4d $ \mathcal{N}=2 $ rank-1 theories. The circle compactification of the 5d $ \mathrm{SU}(2) $ gauge theory with $ N_f \leq 4 $ fundamental hypermultiplets reduces to the 4d $ \mathrm{SU}(2) $ gauge theory with the same number of flavor hypermultiplets; this reduction corresponds to the degeneration denoted by the blue arrows in Figure~\ref{fig:PainleveClassify}. For the 5d $ \mathrm{SU}(2) $ gauge theory with $ N_f=5, 6, 7 $, circle compactification yields the 4d $ \mathcal{N}=2 $ Minahan-Nemeschansky (MN) SCFTs with $ E_6 $, $ E_7 $, and $ E_8 $ global symmetries, respectively \cite{Minahan:1996fg, Minahan:1996cj}. On the side of discrete Painlev\'e equations, there are corresponding degenerations from $ q $-Painlev\'e equations to additive difference Painlev\'e equations, denoted by the red arrows in Figure~\ref{fig:PainleveClassify}. It is therefore natural to conjecture that these additive difference Painlev\'e equations are related to the MN theories, although these difference Painlev\'e equations are not directly related to the usual differential Painlev\'e equations \cite{Boalch:2007}. Because of the non-Lagrangian nature of the MN theories, their connection to Painlev\'e equations is less developed than the connections for the other cases. We do not address this problem in this paper. Finally, the special RG flows from 5d theories to 4d AD theories, indicated by the orange arrows in Figure~\ref{fig:PainleveClassify} have been studied recently in \cite{Bonelli:2025juz}.

In this paper, we study the concrete connections between elliptic/$ q $-Painlev\'e equations and 5d gauge theories through their tau-functions and instanton partition functions. The tau-functions of the $ q $-Painlev\'e equations connected to the differential Painlev\'e equations by the blue arrows in Figure~\ref{fig:PainleveClassify} can be naturally understood as the $ q $-analogues of the tau-functions of the differential Painlev\'e equations. These tau-functions are conjectured to be Fourier transforms of conformal blocks of the $ q $-deformed Virasoro algebra \cite{Bershtein:2016aef, Jimbo:2017ael}, which are identified with dual Nekrasov instanton partition functions of 5d $ \mathrm{SU}(2) $ gauge theories via the five-dimensional extension of the AGT correspondence \cite{Awata:2009ur}. For example, the tau-function of the $ q\text{-}A_1^{(1)} $ Painlev\'e equation is the $ q $-analogue of \eqref{eq:PIII-tau} and is given by \cite{Bershtein:2016aef, Matsuhira:2018qtx, Bonelli:2020dcp}
\begin{align}
    \tau(t | s, a) = \sum_{n\in \mathbb{Z}} s^n Z(t, aq^n, q) \, , \quad
\end{align}
where $ Z $ is the full Nekrasov partition function of the 5d pure $ \mathrm{SU}(2)_0 $ theory defined in section~\ref{sec:5d}, $ t=e^{-1/g_{\mathrm{YM}}^2} $ is the exponentiated inverse gauge coupling squared, $ a=e^{-\phi} $ encodes the Coulomb branch parameter $ \phi $, and $ q $ is the $ \Omega $-deformation parameter in the self-dual $ \Omega $-background. In analogy with \eqref{eq:PIII3-bilinear}, the two tau-functions $ \tau_1(t) = \tau(t|s,a) $ and $ \tau_2(t) = \sqrt{s}\tau(t|s,a\sqrt{q}) $ satisfy the bilinear relations:
\begin{align}\label{eq:qIII3-bilinear}
    \overline{\tau_1} \underline{\tau_1} = \tau_1^2 + \sqrt{t} \tau_2^2 \, , \quad
    \overline{\tau_2} \underline{\tau_2} = \tau_2^2 + \sqrt{t} \tau_1^2 \, ,
\end{align}
where we use the notation $ \overline{\tau_i} = \tau_i(qt) $ and $ \underline{\tau_i} = \tau_i(q^{-1}t) $ for the forward and backward discrete time evolutions. These equations provide a bilinearization of the $ q\text{-}A_1^{(1)} $ equation, also known as the $ q $-Painlev\'e $ \text{III}_3 $ equation, given by
\begin{align}
    f\overline{f} = g^2 \, , \quad
    g\underline{g} = \frac{f(f+t)}{f+1} \, ,
\end{align}
where $ f $ and $ g $ are expressed in terms of the tau-functions as
\begin{align}
    f = \sqrt{t} \frac{\tau_2^2}{\tau_1^2} \, , \quad
    g = q^{1/4} \sqrt{t} \frac{\tau_2 \overline{\tau_2}}{\tau_1 \overline{\tau_1}} \, .
\end{align}
Such connections between $ q $-Painlev\'e equations and 5d instanton partition functions have been studied for the cases $ N_f \leq 4 $ \cite{Bershtein:2016aef, Jimbo:2017ael, Matsuhira:2018qtx,Gavrylenko:2025nuo}, which lie in the coalescence cascade below the $ q\text{-}D_5^{(1)} $ equation in Figure~\ref{fig:PainleveClassify}. We generalize these connections to the cases $ N_f>4 $ and to the elliptic case, and we bootstrap the corresponding bilinear relations from 5d gauge-theory data.

\section{Bilinear relations of tau-functions}\label{sec:5dbootstrap}

In this section, we describe bilinear relations of elliptic and $ q $-Painlev\'e equations from the perspective of supersymmetric field theories. We define tau-functions as dual Nekrasov partition functions of 5d field theories and discuss the bootstrap method for their bilinear relations.

\subsection{5d gauge theories on \texorpdfstring{$ \Omega $}{Ω}-background}\label{sec:5d}

Consider a 5d $ \mathcal{N}=1 $ gauge theory with gauge group $ G $. The theory contains a vector multiplet for the gauge group $ G $ and charged matter hypermultiplets. The vector multiplet contains a vector field $ A_\mu $ and a real scalar field $ \phi $. Each hypermultiplet contains complex scalar fields that form a doublet under the $ \mathrm{SU}(2)_R $ R-symmetry of the superconformal algebra and transform in a representation of the gauge group $ G $. On the Coulomb branch of the moduli space, the scalar field $ \phi $ in the vector multiplet acquires vacuum expectation values $ \phi_i $ in the Cartan subalgebra of $ G $, where $ i=1,\cdots,\mathfrak{r}_G $ with $ \mathfrak{r}_G=\rank(G) $. As a result, the gauge group $ G $ is broken to its maximal torus $ \mathrm{U}(1)^{\mathfrak{r}_G} $. The Coulomb branch is parametrized by these scalar expectation values, and the low-energy effective Abelian gauge theory is characterized by the cubic prepotential $ \mathcal{F}(\phi) $ given by \cite{Witten:1996qb,Seiberg:1996bd, Intriligator:1997pq}
\begin{align}\label{eq:IMS}
    \mathcal{F} = \left(\frac{m_0}{2} K_{ij} \phi_i \phi_j + \frac{\kappa}{6} d_{ijk} \phi_i \phi_j \phi_k\right) + \frac{1}{12} \left( \sum_{\alpha \in \mathbf{R}} |\alpha \cdot \phi|^3 - \sum_f \sum_{w \in \mathbf{w}_f} |w\cdot \phi + m_f|^3  \right) \, .
\end{align}
Here, $ m_0 = 1/g^2 $ is the inverse of the gauge coupling squared, $ K_{ij} = \Tr(T_i T_j) $ is the Killing form, with $ T_i $ denoting generators of $ G $ in the fundamental representation, $ \kappa $ is the classical Chern-Simons level, and $ d_{ijk} = \frac{1}{2} \Tr T_i \{T_j, T_k\} $ is non-zero only for $ G = \mathrm{SU}(N) $ with $ N \geq 3 $. The remaining two terms are one-loop contributions from vector multiplets and hypermultiplets, where $ \mathbf{R} $ denotes the root system of $ G $ and $ \mathbf{w}_f $ is the set of weights of the $ f $-th hypermultiplet with mass $ m_f $. The Coulomb branch is divided into distinct sub-chambers, which are distinguished by the signs of the arguments of the absolute-value terms in \eqref{eq:IMS}, and the prepotential takes different expressions in different sub-chambers.

The derivatives of the prepotential compute supersymmetry protected physical quantities. The first derivative $ T_i = \partial_i \mathcal{F}(\phi) $ is the tension $ T_i $ of the magnetic monopole strings, the second derivative $ (\tau_{\mathrm{eff}})_{ij} = \partial_i \partial_j \mathcal{F}(\phi) $ yields the effective coupling $ \tau_{\mathrm{eff}} $ which is the metric on the Coulomb branch, and the third derivatives $ \partial_i \partial_j \partial_k \mathcal{F}(\phi) $ compute the Chern-Simons couplings. These physical observables can be used to classify possible 5d gauge theories which have UV-completions \cite{Intriligator:1997pq, Jefferson:2017ahm}. The main idea of the classification program is that UV-complete 5d gauge theories have non-trivial physical Coulomb branch, which has a positive semi-definite metric, and on which the tensions of all monopole strings $ T_i $ are non-negative. 

We study tau-functions of discrete Painlev\'e equations using the partition functions of 5d $ \mathrm{SU}(2) $ gauge theories. The partition function of a 5d $ \mathcal{N}=1 $ theory on the $ \Omega $-deformed $ \mathbb{C}^2 \times S^1 $ is the Witten index defined as \cite{Nekrasov:2002qd}
\begin{align}\label{eq:Z}
    Z(\phi, m; \epsilon_1, \epsilon_2) = \Tr\left[ (-1)^F e^{-\beta\{Q,Q^\dagger\}} e^{-\epsilon_1(J_1+J_R)} e^{-\epsilon_2(J_2+J_R)} e^{-\phi \cdot \Pi} e^{-m \cdot H} \right] \, ,
\end{align}
where $ F $ is the fermion number, $ J_1 $, $ J_2 $ are the Cartan generators of the $ \mathrm{SO}(4) $ Lorentz rotation, $ J_R $ is the Cartan generator of the $ \mathrm{SU}(2)_R $ R-symmetry, and $ \Pi $ and $ H $ are the gauge and flavor charges, respectively. The supercharge $ Q $ and its conjugate $ Q^\dagger $ commute with $ J_1+J_R $ and $ J_2+J_R $. We denote by $ \beta $ the radius of $ S^1 $, by $ \epsilon_{1,2} $ the $ \Omega $-deformation parameters, and by $ \phi $ and $ m $ the chemical potentials associated with the gauge and flavor symmetries, respectively. This Witten index counts BPS states annihilated by $ Q $ and $ Q^\dagger $, and is therefore independent of $ \beta $.

The partition function of a 5d gauge theory on the $ \Omega $-background factorizes into the product of a regularization factor $ e^{\mathcal{E}} $, the perturbative partition function $ Z_{\mathrm{pert}} $ and the instanton partition function $ Z_{\mathrm{inst}} $ as
\begin{align}\label{eq:Z-factor}
    Z(\phi, m_0, m_f; \epsilon_1, \epsilon_2) = e^{\mathcal{E}(\phi, m_0, m_f; \epsilon_1, \epsilon_2)} \cdot Z_{\mathrm{pert}}(\phi,m_f;\epsilon_1,\epsilon_2) \cdot Z_{\mathrm{inst}}(\phi,m_0,m_f; \epsilon_1,\epsilon_2) \, ,
\end{align}
where $ m_0 $ denotes the inverse of the gauge coupling squared and $ m_f $ represents the chemical potentials for the other flavor symmetries. The perturbative partition function $ Z_{\mathrm{pert}} $ consists of the 1-loop contributions of vector multiplets and hypermultiplets, and $ \mathcal{E} $ is called the effective prepotential, which is a polynomial in the Coulomb parameter $ \phi $ and arises from the classical contributions together with the regularization of the one-loop terms \cite{Kim:2020hhh}. These factors depend on the boundary conditions at infinity of $ \mathbb{C}^2 $ that preserve the supercharges $ Q $ and $ Q^\dagger $. We choose boundary conditions such that vector multiplets associated with the positive roots of the gauge group, as well as hypermultiplets with positive masses, survive. Under these boundary conditions, the perturbative partition function is given by
\begin{align}\label{eq:pertZ}
    Z_{\mathrm{pert}} = \exp\!\left[ \sum_{n=1}^\infty \! \left( \sum_{\alpha \in \mathbf{R}^+} \frac{-\cosh(n\epsilon_+) e^{-n\alpha \cdot \phi}}{2n\sinh(\frac{n\epsilon_1}{2}) \sinh(\frac{n \epsilon_2}{2})} + \sum_f \sum_{w\in \mathbf{w}_f} \frac{e^{-n|w\cdot \phi+m_f|}}{4n\sinh(\frac{n\epsilon_1}{2}) \sinh(\frac{n \epsilon_2}{2})} \right)\! \right] \!,
\end{align}
where $ \mathbf{R}^+ $ is the set of positive roots of the gauge group $ G $ and $ \epsilon_+ = \frac{\epsilon_1+\epsilon_2}{2} $. The effective prepotential under the same boundary conditions is given by
\begin{align}\label{eq:effprep}
    \mathcal{E} = \frac{1}{\epsilon_1 \epsilon_2} \left( \mathcal{F} + \frac{\epsilon_1^2+\epsilon_2^2}{48} C_i^G \phi_i + \frac{\epsilon_+^2}{2} C_i^R \phi_i \right) \, .
\end{align}
Here, $ \mathcal{F} $ is the cubic prepotential \eqref{eq:IMS}, and $ C_i^G $ and $ C_i^R $ are mixed gauge-gravitational and gauge-$ \mathrm{SU}(2)_R $ Chern-Simons coefficients induced by the 1-loop contributions of the charged fermions \cite{Bonetti:2013ela, Grimm:2015zea, BenettiGenolini:2019zth}
\begin{align}
    C_i^G = -\frac{\partial}{\partial \phi_i} \left( \sum_{\alpha \in \mathbf{R}} |\alpha \cdot \phi| - \sum_f \sum_{w\in \mathbf{w}_f} |w \cdot \phi + m_f| \right) \, , \quad
    C_i^R = \frac{1}{2} \frac{\partial}{\partial \phi_i} \sum_{\alpha \in \mathbf{R}} |\alpha \cdot \phi| \, .
\end{align}
Finally, the instanton partition function $ Z_{\mathrm{inst}} $ is given by
\begin{align}
    Z_{\mathrm{inst}}(\phi,m_0,m_f;\epsilon_1,\epsilon_2) = 1 + \sum_{n=1}^\infty \mathfrak{q}^n Z_n(\phi,m_f; \epsilon_1,\epsilon_2) \, ,
\end{align}
where $ \mathfrak{q} = e^{-m_0} $ is the instanton counting parameter and $ Z_n $ is the $ n $-instanton partition function defined by the path integral over the instanton moduli space. When the ADHM construction of the instanton moduli space is available, the instanton partition function of the 5d gauge theory can be computed by evaluating the path integral of the ADHM quantum mechanics using supersymmetric localization \cite{Nekrasov:2002qd, Nekrasov:2003rj}.

We remark that the partition function \eqref{eq:Z} can also be interpreted as the refined closed topological string partition function. The refined free energy $ F = \log Z $ admits the genus expansion given by
\begin{align}\label{eq:topexpandrefine}
    F(\phi,m;\epsilon_1,\epsilon_2) = \sum_{n,g=0}^\infty (\epsilon_1 \epsilon_2)^{g-1} (\epsilon_1+\epsilon_2)^{2n} F_{(n,g)}(\phi,m) \, .
\end{align}
The effective prepotential \eqref{eq:effprep} is the polynomial part of the $ F_{(n,g)} $ with $ n+g \leq 1 $. In addition to the effective prepotential, the refined free energy encodes information about BPS states of 5d field theories through
\begin{align}\label{eq:gv}
    F(\phi,m;\epsilon_1,\epsilon_2) = \mathcal{E} + \sum_{j_l,j_r,d} \sum_{n=1}^\infty (-1)^{2(j_l+j_r)} N_{j_l,j_r}^{d} \frac{ \chi_{j_l}(n\epsilon_-) \chi_{j_r}(n\epsilon_+)}{4n\sinh(\frac{n \epsilon_1}{2}) \sinh(\frac{n \epsilon_2}{2})} e^{-n d \cdot t} \, ,
\end{align}
where $ N_{j_l,j_r}^{d} $ is the degeneracy of a single-particle BPS state with spin $ (j_l,j_r) = (\frac{J_1-J_2}{2}, \frac{J_1+J_2}{2}) $ under the $ \mathrm{SO}(4) $ Lorentz rotations and charge $ d $ under the gauge and flavor symmetries of the theory. Here, $ \epsilon_\pm = \frac{\epsilon_1\pm \epsilon_2}{2} $ are the $ \Omega $-deformation parameters, $ t=t(\phi,m) $ collectively denotes the chemical potentials, and $ \chi_j $ is the spin-$ j $ character of $ \mathrm{SU}(2) $. For example, the W-boson and perturbative hypermultiplet contributions in the perturbative partition function \eqref{eq:pertZ} correspond to $ (j_l,j_r)=(0,\frac{1}{2}) $ and $ (0,0) $, respectively. In the geometric construction of 5d field theories realized by M-theory compactified on a local Calabi-Yau threefold, this partition function counts BPS states arising from M2-branes wrapping compact 2-cycles in the Calabi-Yau threefold, where $ t $ represents the K\"ahler parameters of the Calabi-Yau geometry. The ordinary topological string partition function $ Z^{\mathrm{top}} = \exp(F^{\mathrm{top}}) $ can be obtained in the unrefined limit $ \epsilon_1 = -\epsilon_2 = \epsilon $ of the $ \Omega $-deformation parameters:
\begin{align}\label{eq:topexpandunrefine}
    F^{\mathrm{top}}(\phi,m;\epsilon) = F(\phi,m;\epsilon,-\epsilon) = \sum_{g=0}^\infty \epsilon^{2g-2} F_g(\phi,m) \, ,
\end{align}
where $ \epsilon $ is identified with the topological string coupling constant and $ F_g $ is called the genus-$ g $ free energy.

\paragraph{Perturbative data of $ \mathrm{SU}(2) $ gauge theories}

In this paper, we mainly consider 5d $ \mathrm{SU}(2) $ gauge theories coupled to $ N_f $ fundamental hypermultiplets. We denote this theory by the $ \mathrm{SU}(2)+N_f\mathbf{F} $ theory. The refined effective prepotential can be written as
\begin{align}\label{eq:F-su2}
    \mathcal{E} &= \frac{1}{\epsilon_1 \epsilon_2} \left( \frac{8-N_f}{6}\phi^3 + m_0\phi^2  - \frac{1}{2} \sum_{i=1}^{N_f} m_i^2 \phi  - \frac{\epsilon_1^2 + \epsilon_2^2}{24} (2-N_f) \phi + \epsilon_+^2 \phi \right) \, ,
\end{align}
where $ \phi $ is the Coulomb branch parameter in the Dynkin basis of the gauge group $ \mathrm{SU}(2) $, and we choose the chamber $ \phi>0 $ and $ \phi \pm m_i \geq 0 $ for the mass parameters $ m_1, \cdots, m_{N_f} $ of the fundamental hypermultiplets. In this chamber, the 1-loop perturbative partition function is given by
\begin{align} \label{eq:su2-pert}
    Z_{\mathrm{pert}} = \exp\left[ \sum_{n=1}^\infty \left( \frac{-\cosh(n\epsilon_+) e^{-2n\phi}}{2n\sinh(\frac{n\epsilon_1}{2}) \sinh(\frac{n \epsilon_2}{2})} + \sum_{i=1}^{N_f} \frac{e^{-n(\phi+m_i)} + e^{-n(\phi-m_i)}}{4n\sinh(\frac{n\epsilon_1}{2}) \sinh(\frac{n \epsilon_2}{2})} \right) \right] .
\end{align}
When $ N_f>8 $, the coefficient of the cubic term in $ \phi $ in the prepotential \eqref{eq:F-su2} is negative, and consequently the Coulomb branch metric becomes negative \cite{Seiberg:1996bd}. This signals the UV inconsistency of the theory. For this reason, we only consider $ N_f \leq 8 $ theories. From the geometric perspective, these theories can be engineered by M-theory compactified on a local Calabi-Yau threefold, namely the total space of the canonical bundle over the $ N_f $-point blowup of $ \mathbb{P}^1 \times \mathbb{P}^1 $, while the $ \mathrm{SU}(2)_\pi $ theory corresponds to the canonical bundle over the one-point blowup of the complex projective plane $ \mathbb{P}^2 $. The 5d $ \mathrm{SU}(2) $ gauge theories with $ N_f \leq 7 $ flow to SCFTs in the UV, while the $ \mathrm{SU}(2)+8\mathbf{F} $ theory is UV-completed by the 6d $ \mathcal{N}=(1,0) $ E-string theory compactified on a circle.

\paragraph{SU(2) gauge theory with $ N_f \leq 4 $}

The instanton partition functions for theories with $ N_f \leq 4 $ can be computed from the path integral of the instanton ADHM quantum mechanics, whose gauge group is $ \mathrm{U}(n) $, via supersymmetric localization \cite{Hwang:2014uwa}:
\begin{align}\label{eq:su2-adhm}
    Z_n = \sum_{|Y_1|+|Y_2|=n} \prod_{i=1}^2 \prod_{s\in Y_i} \frac{(-1)^{n+\kappa-N_f/2} e^{-\kappa \varphi(s)} \prod_{l=1}^{N_f} 2\sinh(\frac{\varphi(s)+m_l}{2})}{\prod_{j=1}^2 2\sinh(\frac{E_{ij}}{2}) \cdot 2\sinh(\frac{E_{ij}-2\epsilon_+}{2})} \, .
\end{align}
Here, $ (Y_1,Y_2) $ is a colored Young diagram, $ s=(s_1,s_2) \in Y_i $ is a box in the Young diagram $ Y_i $ with vertical and horizontal positions $ s_1 $ and $ s_2 $, respectively, measured from the upper-left corner of $ Y_i $. $ E_{ij} $ and $ \varphi(s) $ are defined as
\begin{align}
    \begin{aligned}
        E_{ij} &= (-1)^i \phi - (-1)^j \phi - \epsilon_1 h_i(s) + \epsilon_2(v_j(s)+1) \, , \\
        \varphi(s) &= (-1)^i \phi - \epsilon_+ - (s_1-1)\epsilon_1 - (s_2-1)\epsilon_2 \, ,
    \end{aligned}
\end{align}
where $ h_i(s) $ is the arm length of $ s\in Y_i $ and $ v_j(s) $ is the leg length of $ s \in Y_j $. When $ N_f>0 $, $ \kappa=0 $ for even $ N_f $ and $ \kappa=\pm 1/2 $ for odd $ N_f $. Changing the sign of $ \kappa $ can be absorbed into a redefinition of the mass parameters $ m_l $. In the case $ N_f=0 $, the two possible choices $ \kappa=0 $ and $ \kappa=1 $ lead to distinct instanton partition functions. They correspond to the choices of the $ \mathbb{Z}_2 $-valued discrete theta angle $ \theta=0 $ and $ \theta=\pi $, respectively, which are associated with the non-trivial homotopy group $ \pi_4(\mathrm{SU}(2)) = \mathbb{Z}_2 $ of the gauge group.

\paragraph{SU(2) gauge theory with $ N_f > 4 $}

When there are more than four fundamental hypermultiplets in the 5d $ \mathrm{SU}(2) $ gauge theory, the formula \eqref{eq:su2-adhm} does not give the correct instanton partition function due to the higher-order poles in the localization integral of the ADHM quantum mechanics. Nevertheless, it is possible to compute the instanton partition function by regarding the $ \mathrm{SU}(2) $ gauge theory as an $ \mathrm{Sp}(1) $ gauge theory and constructing the ADHM quantum mechanics, with an $ \mathrm{O}(n) $ gauge group. Unlike \eqref{eq:su2-adhm}, the instanton partition function is not expressed as a sum over Young diagrams.\footnote{The unrefined partition function can be represented as a sum over Young diagrams \cite{Nawata:2021dlk, Chen:2023smd}.} However, one can compute the instanton partition function order by order in the instanton expansion. At the 1-instanton order, the result is
\begin{align}\label{eq:sp1-adhm-1inst}
    Z_1 = \frac{\prod_{l=1}^{N_f} 2\sinh(\frac{m_l}{2})}{32\sinh(\frac{\epsilon_1}{2}) \sinh(\frac{\epsilon_2}{2}) \sinh(\frac{\epsilon_+\pm\phi}{2})} + \frac{\prod_{l=1}^{N_f} 2\cosh(\frac{m_l}{2})}{32\sinh(\frac{\epsilon_1}{2}) \sinh(\frac{\epsilon_2}{2}) \cosh(\frac{\epsilon_+\pm\phi}{2})} \, ,
\end{align}
where we use the shorthand notation $ \sinh(x\pm y) = \sinh(x+y) \sinh(x-y) $. The case $ N_f=0 $ reproduces the instanton partition function of the $ \mathrm{SU}(2)_0 $ theory, while the 1-instanton partition function of the $ \mathrm{SU}(2)_\pi $ theory can be obtained by changing the sign of the first term. When $ N_f \leq 4 $, this formula reproduces the 1-instanton partition function in \eqref{eq:su2-adhm}, up to changes in the signs of the mass parameters $ m_l $ and constant terms that do not depend on the dynamical parameter $ \phi $ of the theory. The computational details for the higher-instanton partition functions are summarized in Appendix~\ref{app:adhm}. In this paper, we choose the signs of the mass parameters in the instanton partition functions as \eqref{eq:sp1-adhm-1inst}. As a result, the 1-instanton partition function for $ N_f \geq 1 $ has the following leading behavior in the $ e^{-\phi} $ expansion:
\begin{align}\label{eq:1inst-leading}
    Z_1 = \frac{\chi_{\mathbf{R}}(m_l) e^{-\phi}}{4\sinh(\frac{\epsilon_1}{2}) \sinh(\frac{\epsilon_2}{2})} + \mathcal{O}(e^{-2\phi}) \, ,
\end{align}
where
\begin{align}
    \chi_{\mathbf{R}}(m_l) = e^{(m_1+\cdots+m_{N_f})/2} \sum_{n=\text{odd}} \sum_{P_{n}} \prod_{l \in P_n} e^{-m_l} \, , \quad
    |P_n| = n \, ,
\end{align}
for the set $ P_n \subset \{1,2,\cdots,N_f\} $ of order $ n $. We remark that $ \chi_{\mathbf{R}} $ is the character of the $ \mathrm{SO}(2N_f) $ flavor symmetry algebra in the $ 2^{N_f-1} $-dimensional representation $ \mathbf{R} $. When $ N_f \leq 7 $, this leading term is the contribution of hypermultiplet states that form part of a representation of the enhanced $ E_{N_f+1} $ flavor symmetry.

\paragraph{SU(2) gauge theory with $ N_f = 8 $}
The 5d $ \mathrm{SU}(2)+8\mathbf{F} $ theory is UV dual to the 6d $ \mathcal{N}=(1,0) $ E-string theory compactified on a circle. This 6d theory has $ E_8 $ flavor symmetry, and its effective prepotential is given by \cite{Kim:2020hhh}
\begin{align}\label{eq:E-Estr}
    \mathcal{E}_{\text{E-string}} = \frac{1}{\epsilon_1 \epsilon_2} \left(  \frac{\mu_0}{2} \varphi^2 - \frac{1}{2} \sum_{l=1}^8 \mu_l^2 \varphi + \frac{\epsilon_1^2+\epsilon_2^2}{4} \varphi + \epsilon_+^2 \varphi \right) \, ,
\end{align}
where $ \varphi $ is the 6d tensor branch parameter, $ \mu_0 = 1/R_{\mathrm{6d}} $ is the inverse radius of the 6d circle and $ \mu_{l=1,\cdots,8} $ are the holonomies for the $ E_8 $ flavor symmetry. The full partition function of E-string theory also contains non-perturbative contributions coming from the instanton strings \cite{Kim:2014dza}. The parametrizations $ (\phi,m_l) $ in 5d and $ (\varphi, \mu_l) $ in 6d are related by
\begin{align}\label{eq:map-6d5d}
    \varphi = \phi+m_0-m_8 \, , \quad
    \mu_0 = 2m_0 \, , \quad
    \mu_l = m_l \, , \ (1 \leq l \leq 7) \, , \quad
    \mu_8 = m_8 - 2m_0 \, .
\end{align}
Consequently, the effective prepotential \eqref{eq:F-su2} with $ N_f=8 $ and \eqref{eq:E-Estr} are equivalent, up to an overall constant that does not depend on the dynamical parameter $ \phi $ or $ \varphi $ of the theory.

\subsection{Bootstrapping bilinear relations}

In this subsection, we describe the bootstrap method for the bilinear relations of tau-functions using 5d gauge theory data. To define the tau-function, we consider the blowup geometry $ \hat{\mathbb{C}}^2 $ obtained by blowing up the origin of $ \mathbb{C}^2 $ to an exceptional divisor $ \mathbb{P}^1 $. The projective coordinates of $ \hat{\mathbb{C}}^2 $ can be defined by $ (z_0, z_1, z_2) \sim (\lambda^{-1} z_0, \lambda z_1, \lambda z_2) $ for $ \lambda \in \mathbb{C}^* $, where the exceptional $ \mathbb{P}^1 $ corresponds to the $ z_0=0 $ locus. The two generators $ J_1 $ and $ J_2 $ of Lorentz rotations act on the projective coordinates by
\begin{align}
    (z_0, z_1, z_2) \mapsto (z_0, e^{\epsilon_1} z_1, e^{\epsilon_2} z_2)  \, ,
\end{align}
where $ \epsilon_1 $ and $ \epsilon_2 $ are Cartans of the $ \mathrm{SO}(4) $ rotation. The north pole $ (0,1,0) $ and the south pole $ (0,0,1) $ of the $ \mathbb{P}^1 $ are fixed points under the rotation. The $ \mathbb{C}^* $-invariant local coordinates $ (U,V) $ near the fixed points are $ (z_0 z_1, z_2/z_1) $ and $ (z_1/z_2, z_0z_2) $, respectively. The generators of the Lorentz rotation $ J_1 $ and $ J_2 $ act as
\begin{align}\label{eq:local-coord}
    (U,V) \mapsto \left\{\begin{array}{ll}
            (e^{\epsilon_1} U, e^{\epsilon_2-\epsilon_1}V ) & \quad \text{(north pole)} \\
            (e^{\epsilon_1-\epsilon_2} U, e^{\epsilon_2}V ) & \quad \text{(south pole),}
    \end{array}\right.
\end{align}
near the fixed points. This is illustrated in Figure~\ref{fig:blowup}.

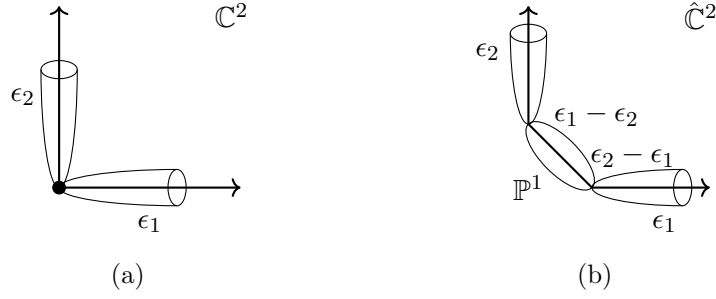
\begin{figure}
    \centering
    \begin{subfigure}[b]{0.4\linewidth}
        \centering
        \begin{tikzpicture}[scale=1.2]
            \draw [<->, thick] (0,2) -- (0,0) -- (2,0);
            \filldraw (0,0) circle (0.07);
            \draw (1.9, 1.9) node {$ \mathbb{C}^2 $};
            \draw (1.3, 0) [partial ellipse = 90:270:1.3 and 0.2];
            \draw (1.3, 0) ellipse (0.1 and 0.2);
            \draw (0, 1.3) [partial ellipse = 0:-180:0.2 and 1.3];
            \draw (0, 1.3) ellipse (0.2 and 0.1);
            \draw (1, -0.4) node {$ \epsilon_1 $};
            \draw (-0.4, 1) node {$ \epsilon_2 $};
        \end{tikzpicture}
        \caption{}
    \end{subfigure}
    \begin{subfigure}[b]{0.4\linewidth}
        \centering
        \begin{tikzpicture}[scale=1.2]
            \draw [<->, thick] (0,2) -- (0,0.7) -- (0.7,0) -- (2,0);
            \draw (1.9, 1.9) node {$ \hat{\mathbb{C}}^2 $};
            \draw (0,0) node {$ \mathbb{P}^1 $};
            \draw (1.7, 0) [partial ellipse = 90:270:1 and 0.2];
            \draw (1.7, 0) ellipse (0.1 and 0.2);
            \draw (0, 1.7) [partial ellipse = 0:-180:0.2 and 1];
            \draw (0, 1.7) ellipse (0.2 and 0.1);
            \draw [rotate around={-45:(0.35,0.35)}] (0.35, 0.35) ellipse (0.495 and 0.2);
            \draw (1.5, -0.4) node {$ \epsilon_1 $};
            \draw (-0.45, 1.5) node {$ \epsilon_2 $};
            \draw (1.15, 0.35) node {$ \epsilon_2-\epsilon_1 $};
            \draw (0.75, 0.8) node {$ \epsilon_1-\epsilon_2 $};
        \end{tikzpicture}
        \caption{}
    \end{subfigure}
    \caption{The blowup geometry $ \hat{\mathbb{C}}^2 $ is obtained by blowing up the origin of $ \mathbb{C}^2 $ and replacing it by $ \mathbb{P}^1 $.} \label{fig:blowup}
\end{figure}

The full partition function $ \hat{\mathcal{Z}} $ on $ \hat{\mathbb{C}}^2 \times S^1 $ can be obtained by performing supersymmetric localization on the Coulomb branch. This involves summing over dynamical magnetic fluxes $ n=(n_i)_{i=1}^{\mathfrak{r}_G} $ in the coweight lattice of the gauge algebra, where $ \mathfrak{r}_G $ is the rank of the gauge group. The background magnetic fluxes $ B=(b_i)_{i=1}^{\mathfrak{r}_F} $ for the maximal torus $ \mathrm{U}(1)^{\mathfrak{r}_F} $ of the flavor symmetry group can also be turned on, where $ \mathfrak{r}_F $ is the rank of the flavor symmetry. These magnetic fluxes are supported on the compact 2-cycle $ \mathbb{P}^1 $. For each flux sector, the partition function is localized at the two fixed points on the $ \mathbb{P}^1 $. As a consequence, $ \hat{\mathcal{Z}} $ can be written as
\begin{align}\label{eq:blowup}
    \hat{\mathcal{Z}} = \sum_{n\in \mathbb{Z}^{\mathfrak{r}_G}+r} (-1)^{|n|} \hat{Z}(\phi_j+n_j \epsilon_1, m_k+b_k \epsilon_1; \epsilon_1, \epsilon_2-\epsilon_1) \hat{Z}(\phi_j + n_j\epsilon_2, m_k+b_k\epsilon_2; \epsilon_1-\epsilon_2, \epsilon_2) \, ,
\end{align}
where $ |n|=\sum n_j $, and we assume that $ \phi_j $ is the Coulomb branch parameter in the Dynkin basis, so that $ r $ is the fractional part of the dynamical magnetic flux $ n $. The two $ \hat{Z} $ factors on the RHS correspond to the partition functions localized at the north and south poles, respectively. They are defined in the same way as \eqref{eq:Z}, but with $ (-1)^F $ replaced by $ (-1)^{2J_R} $ due to the effect of the magnetic flux on the angular momentum of charged states \cite{Huang:2017mis, Kim:2020hhh}. This replacement is equivalent to the shift $ \epsilon_1 \to \epsilon_1+2\pi i $ for the $ \Omega $-deformation parameter in the partition function $ Z $ except for the regularization factor $ e^{\mathcal{E}} $. This formula is called the \emph{blowup equation} \cite{Nakajima:2003pg, Nakajima:2005fg, Gottsche:2006bm}.

The magnetic fluxes $ (n,B) $ on $ \mathbb{P}^1 $ satisfy a quantization condition. For a BPS state with spin $ (j_l, j_r) $ under the $ \mathrm{SO}(4) $ Lorentz rotation and charges $ (\Pi, H) $ under the gauge and flavor symmetries, the magnetic fluxes $ n $ and $ B $ satisfy \cite{Huang:2017mis, Kim:2020hhh, Closset:2022vjj}
\begin{align}\label{eq:quantization}
    (-1)^{2j_l + 2j_r + 1} = (-1)^{n \cdot \Pi + B \cdot H} \, .
\end{align}
The vector and hypermultiplets have $ (j_l,j_r)=(0,\frac{1}{2}) $ and $ (0,0) $, respectively, and the charges of W-bosons and perturbative hypermultiplets constrain the quantization of the magnetic fluxes. For special sets of magnetic fluxes satisfying the quantization condition \eqref{eq:quantization} studied in \cite{Huang:2017mis, Kim:2019uqw, Kim:2020hhh}, the LHS of the blowup equation $ \hat{\mathcal{Z}} $ becomes the partition function $ \hat{Z} $, which is equivalent to the original partition function defined on $ \mathbb{C}^2 \times S^1 $. For more general sets of magnetic fluxes, $ \hat{\mathcal{Z}} $ can be written in terms of the partition functions of $ \mathbb{C}^2 \times S^1 $ in the presence of half-BPS Wilson loop operators inserted at the origin \cite{Kim:2021gyj, Wang:2023zcb}.

We now consider the Nekrasov-Shatashvili (NS) limit $ (\epsilon_1, \epsilon_2) \to (\epsilon, 0) $ of $ \hat{\mathbb{C}}^2_{\epsilon_1,\epsilon_2} $. From Figure~\ref{fig:blowup}(b), the equivariant weights of the coordinates near the north pole of the exceptional divisor are $ (\epsilon,-\epsilon) $. The partition function localized at this patch becomes the partition function in the self-dual background $ \mathbb{C}^2_{\epsilon,-\epsilon} $ with shifted chemical potentials. On the other hand, the local patch near the south pole of the $ \mathbb{P}^1 $ has the equivariant weights $ (\epsilon, 0) $, and the behavior of the normalized localized partition function in the NS limit is
\begin{align}
    \frac{\hat{Z}(\phi_j+n_j \epsilon_2, m_k+b_k \epsilon_2; \epsilon_1-\epsilon_2, \epsilon_2)}{\hat{Z}(\phi,m;\epsilon_1,\epsilon_2)} = \exp\left(n_j \frac{\partial \mathcal{W}}{\partial \phi_j} + b_k \frac{\partial \mathcal{W}}{\partial m_k} - \frac{\partial \mathcal{W}}{\partial \epsilon_1} + \mathcal{O}(\epsilon_2) \right) \, ,
\end{align}
where $ \mathcal{W} $ is the effective twisted superpotential defined as
\begin{align}
    \mathcal{W}(\phi,m,\epsilon_1) = \lim_{\epsilon_2 \to 0} \epsilon_2 \hat{F}(\phi,m;\epsilon_1,\epsilon_2) \, ,
\end{align}
and $\hat{F}=
\log \hat Z$.
Thus, the NS limit of the blowup formula \eqref{eq:blowup} can be expressed as
\begin{equation}\label{eq:ns-blowup}
    \lim_{\epsilon_2 \to 0} \frac{\hat{\mathcal{Z}}}{\hat{Z}} = \sum_n \exp\left(n_j \left( \frac{\partial \mathcal{W}}{\partial \phi_j} + \pi i \right) + b_k \frac{\partial \mathcal{W}}{\partial m_k} - \frac{\partial \mathcal{W}}{\partial \epsilon_1} \right) \hat{Z}(\phi_j+n_j \epsilon, m_k+b_k \epsilon; \epsilon, -\epsilon) \, .
\end{equation}
This has the same structure as the Nekrasov-Okounkov dual partition function \eqref{eq:dualpf}. Motivated by this, we define the tau-function associated with the 5d gauge theory as \cite{Bershtein:2018srt, Bonelli:2024wha, Bonelli:2025juz}
\begin{align}\label{eq:tau}
    \tau\genfrac[]{0pt}{1}{r}{B}(s,\phi,m,\epsilon) = \sum_{n\in \mathbb{Z}^{\mathfrak{r}_G}+r} \bigg[ \prod_{j=1}^{\mathfrak{r}_G} s_j^{n_j} \bigg] \hat{Z}(\phi_j+n_j \epsilon, m_k+b_k \epsilon; \epsilon) \, ,
\end{align}
where $ \hat{Z}(\phi,m;\epsilon) $ is the partition function of the 5d theory in the self-dual $ \Omega $-background $ (\epsilon_1, \epsilon_2)=(\epsilon,-\epsilon) $, and $ s $ is the fugacity for the dynamical magnetic flux $ n $ on the exceptional divisor $ \mathbb{P}^1 $. The tau-function is thus the grand-canonical partition function on $ \hat{\mathbb{C}}^2 \times S^1 $ in the NS limit, and it is connected to the blowup formula by specializing the fugacity $ s $ as in \eqref{eq:ns-blowup}.

We now describe the bootstrapping algorithm for the bilinear relations of the tau-functions. Motivated by the bilinear relation \eqref{eq:qIII3-bilinear} of the $ q $-Painlev\'e $ \mathrm{III}_3 $ system, we write an ansatz for a bilinear relation consisting of three terms:
\begin{align}\label{eq:ansatz}
    \Lambda_1(m,\epsilon) \tauftn{r_1}{B_1} \tauftn{r_1'}{B_1'} + \Lambda_2(m,\epsilon) \tauftn{r_2}{B_2} \tauftn{r_2'}{B_2'} + \Lambda_3(m,\epsilon) \tauftn{r_3}{B_3} \tauftn{r_3'}{B_3'} = 0 \, ,
\end{align}
where we assume that $ \Lambda_i(m,\epsilon) $ do not depend on the dynamical parameter $ \phi $ of the 5d field theory, and the magnetic fluxes $ (r_i, B_i) $ and $ (r_i',B_i') $ satisfy the appropriate quantization condition \eqref{eq:quantization}. Using the definition \eqref{eq:tau} of the tau-function, this can be rewritten as
\begin{equation}
    \sum_{i=1}^3 \Lambda_i \sum_{n,n' \in \mathbb{Z}^{\mathfrak{r}_G}} s^{n+n'+r_i+r_i'} \hat{Z}\big(\phi+(n+r_i)\epsilon, m+B_i\epsilon,\epsilon\big) \hat{Z}\big(\phi+(n'+r_i')\epsilon, m+B_i'\epsilon,\epsilon\big) = 0 ,
\end{equation}
where $ s=(s_1,\cdots,s_{\mathfrak{r}_G}) $ collectively represents the fugacities for the dynamical magnetic fluxes on $ \mathbb{P}^1 $. Since the coefficient of each power of $ s $ must vanish separately for a consistent bilinear relation, we require that the fractional parts of the dynamical magnetic fluxes satisfy
\begin{align}\label{eq:total-flux-cond-r}
    r_1 + r_1' \equiv r_2+r_2' \equiv r_3+r_3' \pmod{\mathbb{Z}^{\mathfrak{r}_G}} \, .
\end{align}
Furthermore, as we will argue in Section~\ref{sec:topological-modularity}, an additional analysis of modularity imposes the following constraints on the background fluxes
\begin{align}\label{eq:total-flux-cond-B}
    B_1 + B_1' = B_2 + B_2' = B_3 + B_3' \, .
\end{align}
We remark that the tau-functions are labeled by the magnetic fluxes $ (r,B) $, and we expect that changing these fluxes corresponds to a B\"acklund transformation on the integrable system side.

Our proposal is that the coefficients $ \Lambda_i $ of the bilinear relation \eqref{eq:ansatz} can be determined from the classical and perturbative information of the 5d gauge theory, as well as its symmetry data or the spinning spectrum of BPS states carrying the minimal gauge charges. We now describe an algorithm for finding the coefficients $ \Lambda_i $. First, fix a set of magnetic fluxes $ (r,B) $ that satisfy the quantization condition \eqref{eq:quantization}. Since imposing only this quantization condition yields infinitely many possible magnetic fluxes, we impose an upper bound for the background magnetic fluxes, such as $ |b_i| < b_{\mathrm{max}} $ or $ \lVert B \rVert < B_{\mathrm{bound}} $. Let us denote the set of such magnetic fluxes $ (r,B) $ by $ \mathcal{S} $. For each element of $ \mathcal{S} $, we associate a tau-function \eqref{eq:tau}. Second, set a \emph{total flux} $ \Phi = (r_{\mathrm{tot}}; B_{\mathrm{tot}}) $ and find a set $ \mathcal{B} $ of pairs of magnetic fluxes $ \{(r,B),(r',B')\} $ defined as
\begin{align}
    \mathcal{B} = \big\{
        \{(r,B),(r',B')\} \mid
        (r,B), (r',B') \in \mathcal{S} \, , \
        r+r' \equiv r_{\mathrm{tot}} \bmod{\mathbb{Z}^{\mathfrak{r}_G}} \, , \
        B+B' = B_{\mathrm{tot}}
    \big\} \, .
\end{align}
Then, choosing three different elements of $ \mathcal{B} $ determines the ansatz \eqref{eq:ansatz} of the bilinear relations.

To find the coefficients $ \Lambda_i $ of the ansatz, we introduce the following assumption on the partition function $ \hat{Z} $. We assume that the partition function written in the form \eqref{eq:gv} has a well-defined power series expansion in $ e^{-d \cdot t} $. This is because the mass of a BPS particle in the 5d field theory is given by $ d \cdot t $ at a generic point of the Coulomb branch, and this quantity is \emph{positive}. More precisely, the K\"ahler parameters $ t $ can be expressed as $ t = C\cdot \phi + A \cdot m $, where $ C $ and $ A $ are constant $ (\mathfrak{r}_G+\mathfrak{r}_F)\times \mathfrak{r}_G $ and $ (\mathfrak{r}_G+\mathfrak{r}_F) \times \mathfrak{r}_F $ matrices, respectively. The components of $ t $ give a convenient set of coordinates on the extended Coulomb branch, and the mass of any BPS particle can be expressed as a non-negative linear combination $ d \cdot t $ of these components. The UV completeness of the 5d field theory requires the existence of a non-trivial physical Coulomb branch, namely a chamber in which all K\"ahler parameters are non-negative, $ t_a \geq 0 $ for all $ a=1,\cdots,\mathfrak{r}_G+\mathfrak{r}_F $ in the $ m\to 0 $ limit \cite{Jefferson:2017ahm, Jefferson:2018irk}.\footnote{The physical Coulomb branch also requires the tensions of monopole strings to be non-negative. This condition reads $ \partial \mathcal{F} / \partial \phi \geq 0 $.} Inside this chamber, the partition function admits an expansion in positive powers of $ e^{-d\cdot t} $. We can therefore take the limit $ t\to \infty $ by sending the dynamical Coulomb moduli $ C \cdot \phi $ to infinity, while keeping the mass parameters $ m $ fixed. In this limit, only the prepotential part $ \mathcal{E} $ in \eqref{eq:gv} contributes to the bilinear relations, since the other terms are suppressed by positive powers of $ e^{-t} $. On the other hand, the coefficients $ \Lambda_i $ are not affected by this limit, since they do not depend on the dynamical parameters $ \phi $ of the theory. We require that the ansatz \eqref{eq:ansatz} be consistent in this limit, which implies that the leading order in the $ e^{-t} $ expansion of the ansatz vanishes at each power of $ s $ separately. This gives strong constraints on the coefficients $ \Lambda_i $. For 5d gauge theories, further constraints on the ans\"atze can be imposed by considering the $ m_0 \to \infty $ limit. In this limit, both the prepotential $ \mathcal{E} $ and the perturbative partition function $ Z_{\mathrm{pert}} $ contribute to the bilinear relations. Again, one needs to impose that each order of $ s $ in the ansatz cancels separately for the chosen coefficients $ \Lambda_i $.

However, in some cases, the leading order of the bilinear relations in the $ t \to \infty $ and $ m_0 \to \infty $ limits is insufficient to determine whether the ansatz is fully consistent, even though the coefficients $ \Lambda_i $ can be determined. A bilinear relation that is consistent at leading order may become inconsistent after the full partition function is included in the tau-functions. Two pieces of information about the 5d field theory are useful for detecting such inconsistent bilinear relations. The first is the subleading order of the partition function and tau-functions in the $ e^{-t} $ expansion, which encodes the BPS states carrying the non-zero minimal gauge charges. These data can be determined for general 5d gauge theories using blowup equations \cite{Kim:2020hhh}. The second is the symmetry of the 5d theory. The full partition function $ \hat{Z} $ is invariant under the Weyl reflection $ w\in W $, where $ W $ is the Weyl group of the global symmetry of the given 5d field theory. Here, one needs to be careful about how a Weyl reflection acts on the parameters of the partition function, since the Weyl reflection can act not only on the mass parameters of the theory, but also on the Coulomb branch parameters $ \phi $. Nevertheless, there is a suitable reparametrization $ (\phi, m) \mapsto (\varphi, \mu) $ such that $ \varphi $ is the \emph{invariant Coulomb branch parameter}, that is, invariant under the Weyl reflections \cite{Hayashi:2019jvx}. In this parametrization, the Weyl reflection of the full partition function $ \hat{Z} = \hat{Z}(\varphi,\mu;\epsilon) $ is given by $ w(\hat{Z}(\varphi,\mu;\epsilon)) = \hat{Z}(\varphi, w(\mu);\epsilon) = \hat{Z}(\varphi,\mu;\epsilon) $. If we denote by $ r $ and $ B $ the fractional part of the magnetic flux and the background magnetic fluxes associated with $ \varphi $ and $ \mu $, then the tau-function defined from the flux $ (r,B) $ transforms under the Weyl reflection $ w \in W $ as
\begin{align}
    w(\tauftn{r}{B}) = \tauftn{r}{w(B)} \, .
\end{align}
Thus, a Weyl reflection maps a bilinear relation to another bilinear relation. Consequently, an ansatz for a bilinear relation can be ruled out if any member of its Weyl orbit is inconsistent. If a gauge theory has symmetry enhancement due to non-perturbative instanton effects, the Weyl reflections of the enhanced symmetry contain non-trivial information beyond the perturbative data of the gauge theory.

This paper focuses on 5d $ \mathrm{SU}(2) + N_f\mathbf{F} $ theories, which are connected to discrete Painlev\'e systems. From \eqref{eq:F-su2} and \eqref{eq:su2-pert}, the unrefined effective prepotential and perturbative partition function that we use are
\begin{align}
    \mathcal{E} &= -\frac{1}{\epsilon^2} \left( \frac{8-N_f}{6}\phi^3 + m_0\phi^2  - \frac{1}{2} \sum_{l=1}^{N_f} m_l^2 \phi  - \frac{\epsilon^2}{12} (2-N_f) \phi \right) \, , \label{eq:su2F-unref-pre} \\
    \hat{Z}_{\mathrm{pert}} &= \prod_{i,j=0}^\infty \frac{\prod_{l=1}^{N_f} (1 + e^{-(\phi+m_l)} q^{i+j+1}) (1+ e^{-(\phi-m_l)} q^{i+j+1})}{(1- e^{-2\phi} q^{i+j+1})^2} \, , \label{eq:su2F-unref-pert}
\end{align}
where $ q=e^{-\epsilon} $. The plus sign in the numerator of $ \hat{Z}_{\mathrm{pert}} $ is the effect of the replacement $ (-1)^F \to (-1)^{2J_R} $ in the hypermultiplet contributions. The perturbative partition function can be expanded in positive powers of $ e^{-\phi} $ as $ \hat{Z}_{\mathrm{pert}} = 1 + \mathcal{O}(e^{-\phi}) $, which is consistent with our expectation. We further assume that the $ k $-instanton partition function $ Z_k $ has the same property:
\begin{align}\label{eq:inst-assume}
    \hat{Z}_k(\phi,m_f; \epsilon) \sim \mathcal{O}(e^{-\phi}) \, ,
\end{align}
for all $ k $. Note that the instanton partition functions from ADHM computations sometimes contain $ \mathcal{O}(1) $ terms in the $ e^{-\phi} $ expansion. However, such terms can be interpreted as extra factors that decouple from the 5d field theory, and we subtract them so that the partition functions are free of such contributions \cite{Hwang:2014uwa, Kim:2020hhh}. The quantization of magnetic fluxes can be determined from the perturbative partition function and the leading order of the $ e^{-\phi} $ expansion of the partition function \eqref{eq:1inst-leading} as
\begin{align}\label{eq:Nf-quantization}
    2r \in \mathbb{Z} \, , \quad
    r + b_i \in \mathbb{Z} + \frac{1}{2} \, , \quad
    r + b_0 + b_i - \frac{1}{2} \sum_{l=1}^{N_f} b_l \in \mathbb{Z} + \frac{1}{2} \, , \quad (1 \leq i \leq N_f) \, .
\end{align}
We bootstrap bilinear relations of elliptic and $ q $-Painlev\'e equations from the prepotential \eqref{eq:su2F-unref-pre}, the 1-loop perturbative partition function \eqref{eq:su2F-unref-pert} and the property \eqref{eq:inst-assume} of the partition function. Requiring the consistency conditions for the ansatz \eqref{eq:ansatz} in the limits $ \phi \to \infty $ and $ m_0 \to \infty $ is enough to determine the coefficients $ \Lambda_i $ in the bilinear relations. However, some of the determined ans\"atze turn out to be inconsistent if we include the instanton corrections. Such incorrect ans\"atze can be ruled out by using either the leading order of the 1-instanton partition function \eqref{eq:1inst-leading} or the enhanced global symmetry of the $ \mathrm{SU}(2)+N_f \mathbf{F} $ theory.

\subsection{Constraint from modularity}\label{sec:topological-modularity}

From the viewpoint of M-theory, a UV-consistent 5d $\mathcal{N}=1$ theory can be geometrically engineered by compactifying M-theory on a non-compact Calabi-Yau threefold $X$ \cite{Witten:1996qb,Intriligator:1997pq}, and the partition function is identified with the refined topological string partition function as \eqref{eq:topexpandrefine}. In this subsection, we argue the magnetic flux condition \eqref{eq:total-flux-cond-B} from the topological B-model formulation.

The topological string partition function is a wave function obtained by quantizing the symplectic space of B-model periods, and a change of symplectic frame acts as a change of polarization \cite{Bershadsky:1993ta,Bershadsky:1993cx,Witten:1993ed,Gunaydin:2006bz}. In a real polarization, the genus-$g$ amplitudes with $g>1$ are therefore weight-zero quasi-modular objects rather than ordinary modular functions \cite{Aganagic:2006wq,Huang:2006hq,Grimm:2007tm}. This modular property provides a strict constraint on the blowup equations \cite{Huang:2017mis}. By performing an analysis similar to that in \cite{Huang:2017mis}, we find that it also constrains the bilinear relations of tau-functions used in our bootstrap procedure. In the following, we consider the leading order in the $\epsilon$ expansion in \eqref{eq:topexpandunrefine}.

Let $\phi$ and $m$ be the collections of Coulomb parameters $\phi=(\phi_1,\cdots,\phi_{\mathfrak{r}_G})$ and mass parameters $m=(m_0,m_1,\cdots,m_{\mathfrak{r}_F})$, respectively, and define their dual parameters by
\begin{align}
    \phi_{D,i}=\frac{\partial \hat F_0}{\partial\phi_i},
    \qquad
    m_{D,f}=\frac{\partial  \hat F_0}{\partial m_f}.
\end{align}
Here, $ \hat F_0 $ is the genus-zero free energy defined in \eqref{eq:topexpandunrefine} for the hatted partition function $\hat{Z}$. It is useful to organize the second order derivatives of $F_0$ as
\begin{align}\label{eq:tauvw}
    \tau_{ij}
    &=
    \frac{\partial^2\hat F_0}{\partial\phi_i\partial\phi_j},
    &
    v_{if}
    &=
    \frac{\partial^2\hat F_0}{\partial\phi_i\partial m_f},
    &
    \omega_{fg}
    &=
    \frac{\partial^2\hat F_0}{\partial m_f\partial m_g}.
\end{align}
Here, $\tau_{ij}$ is the effective coupling matrix, while $v_{if}$ and $\omega_{fg}$ describe the mixing of the dynamical and flavor parameters. We will drop the indices and use bold symbols $\boldsymbol{\tau}$, $\boldsymbol v$ and $\boldsymbol\omega$ when no confusion arises. In the B-model, $\phi$ and $\phi_D$ form the dynamical periods and $m$ are the constant periods. Consider a monodromy transformation acting on the periods as\footnote{In the standard normalization, conventional factors of $2\pi\mathrm{i}$ should accompany the periods and the associated modular and elliptic variables. We suppress these factors throughout this subsection, with the understanding that they are absorbed consistently into the definitions of these quantities.}
\begin{align}\label{eq:modulartransformation}
    \begin{pmatrix}
        \phi_D'\\
        \phi'
    \end{pmatrix}
    &=
    M
    \begin{pmatrix}
        \phi_D\\
        \phi
    \end{pmatrix}, \quad
    M =
    \begin{pmatrix}
        \mathsf{A} & \mathsf{B}\\
        \mathsf{C} & \mathsf{D}
    \end{pmatrix}
    \in\Gamma\subset\mathrm{Sp}(2\mathfrak{r}_G,\mathbb{Z}), \quad
    \text{and}\quad m'=m.
\end{align}
The symplectic condition $M^T\Omega M=\Omega$, with
\begin{align}
    \Omega
    =
    \begin{pmatrix}
        0 & \mathbf{1}_{\mathfrak{r}_G}\\
        -\mathbf{1}_{\mathfrak{r}_G} & 0
    \end{pmatrix},
\end{align}
implies
\begin{align}\label{eq:Sprelations}
    \mathsf{A}^T\mathsf{C}
    &=
    \mathsf{C}^T\mathsf{A},
    &
    \mathsf{B}^T\mathsf{D}
    &=
    \mathsf{D}^T\mathsf{B},
    &
    \mathsf{A}^T\mathsf{D}
    -
    \mathsf{C}^T\mathsf{B}
    &=
    \mathbf{1}_{\mathfrak{r}_G}.
\end{align}

The modular transformations of the second-order derivatives of $F_0$ can be derived from the monodromy transformations \cite{deWit:1996gjy}. From the differentials
\begin{align}\label{eq:dphi}
    d\phi_D'
    &=
    (\mathsf{A}\boldsymbol\tau+\mathsf{B})d\phi+\mathsf{A}\boldsymbol v\,dm,
    &
    d\phi'
    &=
    (\mathsf{C}\boldsymbol\tau+\mathsf{D})d\phi+\mathsf{C}\boldsymbol v\,dm.
\end{align}
we obtain, at fixed $m'$, 
\begin{align}
    \boldsymbol\tau'=\frac{\partial \phi_D'}{\partial \phi'}=(\mathsf{A}\boldsymbol \tau+\mathsf{B})(\mathsf{C}\boldsymbol\tau+\mathsf{D})^{-1},
\end{align}
which is precisely the modular transformation of the effective coupling $\boldsymbol \tau $. Next, we derive the modular transformation for $\boldsymbol v$ and $\boldsymbol \omega$. From the symplectic condition, we have
\begin{align}\label{eq:symplecticrelations}
   \begin{pmatrix}
        \phi_D'^T\ 
        \phi'^T
    \end{pmatrix}\Omega \begin{pmatrix}
        d\phi_D'\\
       d \phi'
    \end{pmatrix}=\begin{pmatrix}
        \phi_D^T\ 
        \phi^T
    \end{pmatrix}\Omega \begin{pmatrix}
        d\phi_D\\
       d \phi
    \end{pmatrix}.
\end{align}
Notice that by requiring that there is no integration constant involving $m$, the integration of \eqref{eq:modulartransformation} gives \cite[eq. (9)]{deWit:1996gjy}
\begin{align}\label{eq:Fharrelations}
    \hat F'_0=\hat F_0-\frac{1}{2}\phi^T \phi_{D}+\frac{1}{2}\phi'{}^T\phi_{D}' \, ,
\end{align}
where $\hat F'_0$ and $\hat F_0$ represent the genus zero free energies $\hat F_0'(\phi',m')$ and $\hat F_0(\phi,m)$ after and before the modular transformation respectively. Utilizing \eqref{eq:symplecticrelations}, the differential of \eqref{eq:Fharrelations} can be expressed as
\begin{align}
    d\hat F'=d\hat F+\phi_D'{}^Td\phi'-\phi_D^Td\phi=\phi_D'{}^Td\phi'+m_D^T dm=\phi_D'{}^Td\phi'+m_D^T dm' \, ,
\end{align}
where $m_D=\partial \hat{F_0}/\partial m$. Then, at fixed $\phi'$, we have
\begin{align}
    \frac{\partial \hat{F}_0'}{\partial m'}=m_D'=m_D \, .
\end{align}
This implies the equality for $dm_D'=dm_D$ which gives
\begin{align}
    \boldsymbol v'{}^Td\phi'+\boldsymbol\omega'dm'=\boldsymbol v^Td\phi+\boldsymbol\omega dm \, .
\end{align}
At fixed $m'$, we obtain
\begin{align}
    \boldsymbol v'=(\mathsf{C}\boldsymbol\tau+\mathsf{D})^{-T}\boldsymbol v \, .
\end{align}
Moreover, \eqref{eq:dphi} at fixed $ \phi' $ implies
\begin{align}
    \boldsymbol\omega' &=\boldsymbol\omega- \boldsymbol v^T(\mathsf{C}\boldsymbol\tau+\mathsf{D})^{-1}\mathsf{C}\boldsymbol v.
\end{align}
After obtaining these transformation rules, we apply them to the tau-functions to see whether a covariant form of the bilinear relations provides additional constraints on the ansatz. Utilizing the notation \eqref{eq:tauvw}, the genus expansion of the tau-function can be expressed as
\begin{align}\label{eq:Zinvstau}
    \hat Z^{-1}\tauftn{r}{B}=e^{\frac{B^Tm_D}{\epsilon}+\frac{1}{2}B^T\boldsymbol\omega B}
    \left(\Theta_r(\boldsymbol\tau,z)+\mathcal{O}(\epsilon)\right),
\end{align}
where $\hat{Z}$ is the full self-dual Nekrasov partition function and $\Theta_r(\boldsymbol\tau,z)$ is the Riemann theta function
\begin{align}
    \Theta_r(\boldsymbol\tau,z)
    =
    \sum_{n\in\mathbb{Z}^{\mathfrak{r}_G}+r}
    \exp\left(
        \frac{1}{2}n^T\boldsymbol\tau n+n^Tz
    \right), \quad z=\log s+\frac{\phi_D}{\epsilon}+ \boldsymbol vB.
\end{align}
Under a monodromy transformation, $\boldsymbol\tau$ transforms as the period
matrix, while $\phi_D'=\mathsf A\phi_D+\mathsf B\phi,\boldsymbol v'=(\mathsf C\boldsymbol\tau+\mathsf D)^{-T}\boldsymbol v$. We choose the transformation of the fugacity $s$ such that the complete theta argument $z$ transforms as $z'=(\mathsf{C}\boldsymbol\tau+\mathsf{D})^{-T}z$. Thus, we can identify the monodromy transformation as a modular transformation for the Riemann theta function appearing in \eqref{eq:Zinvstau}. At leading order in $\epsilon$ and up to a phase factor, the tau-function \eqref{eq:Zinvstau} then transforms as\footnote{Here we impose another condition that the characteristic of the Riemann theta function appearing in \eqref{eq:Zinvstau} is invariant under modular transformation. Together with the flux quantization condition on $r$, this requirement imposes constraints on the transformation matrix $M$, restricting it to a subgroup $\Gamma$ as written in \eqref{eq:modulartransformation}.}
\begin{align}\label{eq:taumonodromy}
    \hat Z^{-1}\tauftn{r}{B} \
    \xrightarrow{\mathrm{monodromy}} \
    &\det(\mathsf{C}\boldsymbol\tau+\mathsf{D})^{1/2}
    \exp\Bigg[
            \frac{1}{2}z^T(\mathsf{C}\boldsymbol\tau+\mathsf{D})^{-1}\mathsf{C}z
            -\frac{1}{2}B^T\boldsymbol v^T(\mathsf{C}\boldsymbol\tau+\mathsf{D})^{-1}
            \mathsf{C}\boldsymbol vB
    \Bigg] \nonumber \\
    &\qquad \cdot \hat Z^{-1}\tauftn{r}{B} \left(1+\mathcal{O}(\epsilon)\right).
\end{align}
Here the determinant and the first term in the exponential are obtained from the modular transformation for a standard Riemann theta function. Using $z=\log s+\phi_D/\epsilon+\boldsymbol vB$ and the symmetry of $(\mathsf{C}\boldsymbol\tau+\mathsf{D})^{-1}\mathsf{C}$, the exponent in \eqref{eq:taumonodromy} becomes
\begin{align}
    &
    \frac{1}{2}z^T(\mathsf{C}\boldsymbol\tau+\mathsf{D})^{-1}\mathsf{C}z
    -\frac{1}{2}B^T\boldsymbol v^T(\mathsf{C}\boldsymbol\tau+\mathsf{D})^{-1}
    \mathsf{C}\boldsymbol vB \\
    &=
    \frac{1}{2}
    \left(\log s+\frac{\phi_D}{\epsilon}\right)^T
    (\mathsf{C}\boldsymbol\tau+\mathsf{D})^{-1}\mathsf{C}
    \left(\log s+\frac{\phi_D}{\epsilon}\right)
    +
    \left(\log s+\frac{\phi_D}{\epsilon}\right)^T
    (\mathsf{C}\boldsymbol\tau+\mathsf{D})^{-1}\mathsf{C}\boldsymbol vB . \nonumber
\end{align}
In particular, the terms quadratic in $B$ cancel. It follows that the product of two tau-functions transforms as
\begin{align}
\hat Z^{-2}\tauftn{r}{B}\tauftn{r'}{B'}
\quad\xrightarrow{\mathrm{monodromy}}\quad
&
\det(\mathsf{C}\boldsymbol \tau+\mathsf{D})\times
\hat Z^{-2}\tauftn{r}{B}\tauftn{r'}{B'}
\left(1+\mathcal{O}(\epsilon)\right)
\nonumber\\
&\times
\exp\Bigg[
\left(\log s+\frac{\phi_D}{\epsilon}\right)^T
(\mathsf{C}\boldsymbol \tau+\mathsf{D})^{-1}\mathsf{C}
\left(\log s+\frac{\phi_D}{\epsilon}\right)
\nonumber\\
&\hspace{12mm}
+
\left(\log s+\frac{\phi_D}{\epsilon}\right)^T
(\mathsf{C}\boldsymbol \tau+\mathsf{D})^{-1}\mathsf{C}\boldsymbol v(B+B')
\Bigg].
\end{align}
Therefore, for the bilinear relation to be consistent, the exponential factor should be the same for all terms, which requires
\begin{align}
B_1+B_1'=B_2+B_2'=B_3+B_3'=B_{\mathrm{tot}}.
\end{align} 
This is precisely the total flux condition \eqref{eq:total-flux-cond-B} imposed on the bilinear ansatz \eqref{eq:ansatz}.

\subsection{Instructive examples} \label{sec:easy-example}

In this subsection, we present illustrative examples of bootstrapping bilinear relations from 5d gauge theory in detail.

\subsubsection{Bilinear relations from pure \texorpdfstring{$ \mathrm{SU}(2) $}{SU(2)} gauge theories} \label{sec:Nf0-bilinear}

As the simplest example, let us consider the 5d pure $ \mathrm{SU}(2) $ gauge theories with discrete theta angles $ \theta=0 $ and $ \pi $. These theories are perturbatively indistinguishable, and the unrefined effective prepotential and 1-loop partition function are therefore given by
\begin{align}
    \mathcal{E}(\phi,m_0,\epsilon) &= -\frac{8\phi^3 + 6m_0 \phi^2 - \epsilon^2\phi}{6\epsilon^2} \, , \quad
    Z_{\mathrm{pert}}(\phi, \epsilon) = \prod_{i,j=0}^\infty \frac{1}{(1-a^2 q^{i+j+1})^2} \, , \label{eq:pure-su2-unref}
\end{align}
where $ a=e^{-\phi} $ and $ q=e^{-\epsilon} $. However, they are non-perturbatively different, and this is reflected in distinct choices of magnetic flux quantization on $ \mathbb{P}^1 $ \cite{Kim:2020hhh}. Depending on the discrete theta angle, the possible quantizations are 
\begin{align}
    \begin{aligned}\label{eq:su2-quantization}
        \theta = 0 \quad &: \quad r=0 , \ B \in \mathbb{Z} \quad \text{ or } \quad r=\frac{1}{2}, \ B \in \mathbb{Z} \, , \\
        \theta = \pi \quad &: \quad r = 0 , \ B \in \mathbb{Z}+\frac{1}{2} \quad \text{ or } \quad r = \frac{1}{2} , \ B \in \mathbb{Z} \, ,
    \end{aligned}
\end{align}
where $ B $ is the background magnetic flux for the $ \mathrm{U}(1) $ instanton symmetry. For a given magnetic flux $ (r, B) $ satisfying the quantization condition \eqref{eq:su2-quantization}, we define the tau-function as
\begin{align}
    \tauftn{r}{B}(s, \phi, m_0, \epsilon)
    = \sum_{n \in \mathbb{Z}+r} s^n \hat{Z}(\phi+n\epsilon, m_0+B \epsilon; \epsilon) \, .
\end{align}

We now illustrate the explicit procedure for the bootstrap program in the $ \theta=0 $ case. To be concrete, let us consider the ansatz \eqref{eq:ansatz} with total flux $ \Phi=(0;0) $. Among the tuples of magnetic fluxes, we first focus on the tau-functions defined from the magnetic fluxes
\begin{align}\label{eq:su2-correct-flux-eg}
    (r_1,B_1) = (0, 1) \, , \quad
    (r_2,B_2) = (0, 0) \, , \quad
    (r_3,B_3) = (1/2, 0) \, ,
\end{align}
where the other three tau-functions $ \tauftn{r_i'}{B_i'} $ are automatically determined from the total flux. Using the effective prepotential \eqref{eq:pure-su2-unref}, the three terms in the bilinear relation are
\begin{align}
    \begin{aligned}\label{eq:su2-correct-prep}
        \tauftn{r_1}{B_1} \tauftn{r_1'}{B_1'} &= \sum_{N \in \mathbb{Z}} f(N) s^N \sum_{n\in \mathbb{Z}} \mathfrak{q}^{2n(2n-N)} a^{4n(1+2n-2N)} q^{2n(1+2n-2N)} (1+\mathcal{O}(\mathfrak{q},a)) \, , \\
        \tauftn{r_2}{B_2} \tauftn{r_2'}{B_2'} &= \sum_{N \in \mathbb{Z}} f(N) s^N \sum_{n\in \mathbb{Z}} \mathfrak{q}^{2n(2n-N)} a^{8n(n-N)+2N} q^{N(4n^2-4nN+N)} (1+\mathcal{O}(\mathfrak{q},a)) \, , \\
        \tauftn{r_3}{B_3} \tauftn{r_3'}{B_3'} &= \sum_{N \in \mathbb{Z}} f(N) s^N \sum_{n\in \mathbb{Z}} \mathfrak{q}^{\frac{1}{2}(1+2n)(1+2n-2N)} a^{8n(1+n-N)-2N+2} \\
        & \hspace{15em} \cdot q^{N(4n(1+n-N)-N+1)} (1+\mathcal{O}(\mathfrak{q},a)) \, ,
    \end{aligned}
\end{align}
where $ f(N) $ is the common factor for the three terms, which does not depend on the summation index $ n $, and the $ \mathcal{O}(\mathfrak{q},a) $ terms are the contributions from the 1-loop and instanton partition functions. For the bilinear relation to be satisfied, the linear combinations of \eqref{eq:su2-correct-prep} at each power of $ s $ must vanish separately. Under the assumptions on the behavior of the instanton partition functions \eqref{eq:inst-assume}, the $ s^0 $ and $ s^{1} $ parts of the ansatz in the $ \phi \to \infty $ limit are given by
\begin{align}
    \begin{aligned}
        s^0 \ &: \ \left(1 + \mathcal{O}(a)\right) \Lambda_1 + \left(1 + \mathcal{O}(a)\right) \Lambda_2 + \left(\mathcal{O}(a^2)\right) \Lambda_3 = 0 \, , \\
        s^1 \ &: \ \left(1 + \mathcal{O}(a)\right) \Lambda_1 + \left(\mathcal{O}(a^2)\right) \Lambda_2 + \left(\frac{1}{\sqrt{\mathfrak{q}}} + \mathcal{O}(a) \right) \Lambda_3 = 0 \, .
    \end{aligned}
\end{align}
This determines the coefficients $ \Lambda_i $ in the bilinear relation as
\begin{align}
    \Lambda_1 = -1 \, , \quad
    \Lambda_2 = 1 \, , \quad
    \Lambda_3 = \sqrt{\mathfrak{q}} \, .
\end{align}
They are compatible with the other powers of $ s $ in the $ \phi\to \infty $ limit. Using the 1-loop partition function \eqref{eq:pure-su2-unref}, one can also verify that these coefficients are compatible with the limit $ m_0 \to \infty $.

Instead of \eqref{eq:su2-correct-flux-eg}, we next construct the ansatz from the magnetic fluxes
\begin{align}\label{eq:su2-wrongflux}
    (r_1,B_1) = (0, 2) \, , \quad
    (r_2,B_2) = (0, 1) \, , \quad
    (r_3,B_3) = (0, 0) \, ,
\end{align}
with total flux $ (0,0) $. Each term in the bilinear relation is then
\begin{align}\label{eq:su2-wrong-bilinear}
        \tauftn{r_1}{B_1} \tauftn{r_1'}{B_1'} &= \sum_{N \in \mathbb{Z}} f(N) s^N \sum_{n\in \mathbb{Z}} \mathfrak{q}^{2n(n-N)} a^{8n(n-N+1)} q^{4nN(n-N+1)} (1+\mathcal{O}(\mathfrak{q},a)) \, , \\
        \tauftn{r_2}{B_2} \tauftn{r_2'}{B_2'} &= \sum_{N \in \mathbb{Z}} f(N) s^N \sum_{n\in \mathbb{Z}} \mathfrak{q}^{2n(n-N)} a^{4n(2n-2N+1)+2N} q^{N(2n(2n-2N+1)+N)} (1+\mathcal{O}(\mathfrak{q},a)) \, , \nonumber \\
        \tauftn{r_3}{B_3} \tauftn{r_3'}{B_3'} &= \sum_{N \in \mathbb{Z}} f(N) s^N \sum_{n\in \mathbb{Z}} \mathfrak{q}^{2n(n-N)} a^{8n(n-N)+4N} q^{2N(2n(n-N)+N)} (1+\mathcal{O}(\mathfrak{q},a)) \, , \nonumber
\end{align}
where $ f(N) $ is the common factor of the three terms independent of the summation index $ n $. The $ s^0 $ and $ s^1 $ parts of the ansatz in the large Coulomb parameter $ \phi \to \infty $ limit are
\begin{align}
    \begin{aligned}\label{eq:su2-wrongeq}
        s^0 \ &: \ \left(1+\mathfrak{q}^2 + \mathcal{O}(a) \right) \Lambda_1 + \left( 1 + \mathcal{O}(a) \right) \Lambda_2 + \left( 1 + \mathcal{O}(a) \right) \Lambda_3 = 0 \\
        s^1 \ &: \ \left(1+\mathcal{O}(a) \right) \Lambda_1 + \left( \mathcal{O}(a^2) \right) \Lambda_2 + \left( \mathcal{O}(a^4) \right) \Lambda_3 = 0 \, .
    \end{aligned}
\end{align}
This implies $ \Lambda_1=0 $ and $ \Lambda_2+\Lambda_3=0 $. We also need to verify whether this choice is compatible with the perturbative partition function \eqref{eq:pure-su2-unref}. At order $ s^1 $, the $ n=0 $ and $ n=1 $ sectors of \eqref{eq:su2-wrong-bilinear} are the leading contributions in the $ m_0\to \infty $ limit. Since the contributions from the perturbative partition functions are the same for all three terms \eqref{eq:pure-su2-unref} in this case, the $ s^1 $ part of the ansatz in the $ m_0 \to \infty $ limit is given by
\begin{align}
    \left( 1+q^4 a^8 + \mathcal{O}(\mathfrak{q}) \right) \Lambda_1 + \left( qa^2 + q^3a^6 + \mathcal{O}(\mathfrak{q}) \right) \Lambda_2 + \left (2q^2 a^4 + \mathcal{O}(\mathfrak{q}) \right) \Lambda_3 = 0 \, .
\end{align}
This is incompatible with \eqref{eq:su2-wrongeq}, so we conclude that the magnetic fluxes \eqref{eq:su2-wrongflux} do not yield a consistent bilinear relation.

In this way, we scan the possible tuples of fluxes in the set
\begin{align}
    \begin{aligned}
        \mathcal{S} &= \left\{ (r,B) \, \bigg| \, r=0 \text{ or } \frac{1}{2}, \ B  \in \mathbb{Z} , \ |B| \leq 5 \right\} \, , \\
        \mathcal{B} &= \left\{ \{(r,B),(r',B')\} \mid (r,B), (r',B') \in \mathcal{S} , \ r+r'\in \mathbb{Z} , \ B+B' = 0 \, \right\} \, ,
    \end{aligned}
\end{align}
and find only four bilinear relations of the form \eqref{eq:ansatz}:
\begin{gather}\label{eq:su(2)0-bilinear}
    \begin{gathered}
    \overline{\tau_1} \underline{\tau_1} = \tau_1^2 + \sqrt{\mathfrak{q}} \tau_2^2 \, , \quad
    \overline{\tau_2} \underline{\tau_2} = \tau_2^2 + \sqrt{\mathfrak{q}} \tau_1^2 \, , \\
    \overline{\tau_1} \underline{\tau_1} = (1-\mathfrak{q}) \tau_1^2 + \sqrt{\mathfrak{q}} \overline{\tau_2}{\underline{\tau_2}} \, , \quad
    \overline{\tau_2} \underline{\tau_2} = (1-\mathfrak{q}) \tau_2^2 + \sqrt{\mathfrak{q}} \overline{\tau_1}{\underline{\tau_1}} \, .
    \end{gathered}
\end{gather}
Here, the two tau-functions $ \tau_1 $ and $ \tau_2 $ are defined as
\begin{align}
    \tau_1 = \tauftn{0}{0} \, , \quad
    \tau_2 = \tauftn{1/2}{0} \, ,
\end{align}
and we denote the forward and backward discrete time evolutions as
\begin{align}
    \overline{\tau_i} = \tau_i(s,\phi,m_0+\epsilon,\epsilon) \, , \quad
    \underline{\tau_i} = \tau_i(s,\phi,m_0-\epsilon,\epsilon) \, .
\end{align}
Although we derive these relations from the classical and perturbative data of the 5d $ \mathrm{SU}(2) $ gauge theory, they are compatible with the instanton partition function \eqref{eq:su2-adhm}. We have checked the validity of the bilinear relations up to 3-instanton order. Note that the bilinear relations in the second line are linear combinations of those in the first line, so it suffices to consider the first two bilinear relations. On the other hand, we do not obtain any bilinear relation from the total fluxes $ \Phi= (0;1) $, $ (1/2; 0) $ and $ (1/2;1) $. As discussed at the end of section~\ref{sec:discreteP-5d}, the two bilinear relations in the first line of \eqref{eq:su(2)0-bilinear} are equivalent to the $ q $-Painlev\'e equation given by
\begin{align}
    f\overline{f} = g^2 \, , \quad
    g\underline{g} = \frac{f(f+\mathfrak{q})}{f+1} \qquad \Leftrightarrow \qquad
    \overline{f}\underline{f} = \left( \frac{f+\mathfrak{q}}{f+1} \right)^2 \, ,
\end{align}
where $ f $ and $ g $ are given by
\begin{align}
    f = \sqrt{\mathfrak{q}} \frac{\tau_2^2}{\tau_1^2} \, , \quad
    g = q^{1/4} \sqrt{\mathfrak{q}} \frac{\tau_2 \overline{\tau_2}}{\tau_1 \overline{\tau_1}} \, .
\end{align}
This is the $ q\text{-}A_1^{(1)} $ equation \eqref{eq:qA1}.

The bilinear relations of the tau-functions from the $ \mathrm{SU}(2)_\pi $ theory can be treated in a similar way. By scanning the possible tuples of magnetic fluxes for the total fluxes $ \Phi = (0;0) $ and $ (0;1) $ up to the bound $ |B_i| \leq 5 $, we find two bilinear relations
\begin{align}
    \begin{aligned}\label{eq:su2-pi-bilinear}
        \tauftn{0}{-1/2} \tauftn{0}{3/2} &= \tauftn{0}{1/2}^2 + q^{1/4} \sqrt{\mathfrak{q}} \tauftn{1/2}{0} \tauftn{1/2}{1} \, , \\
        \tauftn{1/2}{1} \tauftn{1/2}{-1} &= \tauftn{1/2}{0}^2 + \sqrt{\mathfrak{q}} \tauftn{0}{-1/2} \tauftn{0}{1/2} \, .
    \end{aligned}
\end{align}
We do not obtain any additional bilinear relations from the total fluxes $ \Phi = (1/2,\pm 1/2) $. These bilinear relations describe another $ q $-Painlev\'e system $ q\text{-}\tilde{A}_1^{(1)} $. The discrete time evolution for this $ q $-Painlev\'e equation is defined as \cite{Bershtein:2018srt}
\begin{align}
    \overline{\tau}(s,\phi,m_0,\epsilon) = \tau(s,\phi+\tfrac{\epsilon}{2}, m_0+\tfrac{\epsilon}{2},\epsilon) \, , \quad
    \underline{\tau}(s,\phi,m_0,\epsilon) = \tau(s,\phi-\tfrac{\epsilon}{2}, m_0-\tfrac{\epsilon}{2},\epsilon) \, .
\end{align}
If we define $ \tau_1 = \underline{\tauftn{0}{1/2}} $, then the two bilinear relations in \eqref{eq:su2-pi-bilinear} can be written as a single bilinear relation
\begin{align}
    \overline{\overline{\tau_1}} \underline{\underline{\tau_1}} = \tau_1^2 + \sqrt{\mathfrak{q}} \overline{\tau_1} \underline{\tau_1} \, .
\end{align}
The variables $ (f,g) $ defined by
\begin{align}
    f = \sqrt{\mathfrak{q}} \frac{\overline{\tau_1} \underline{\tau_1}}{\tau_1^2} \, , \quad
    g = q^{1/4} \mathfrak{q} \frac{\overline{\overline{\tau_1}} \underline{\tau_1}}{\overline{\tau_1} \tau_1} \, ,
\end{align}
satisfy
\begin{align}
    f\overline{f} = g \, , \quad
    g\underline{g} = \mathfrak{q}^2 (1+f) \, ,
    \qquad \Leftrightarrow \qquad
    \overline{f}f^2\underline{f} = \mathfrak{q}^2 (1+f) \, ,
\end{align}
which is the $ q\text{-}\tilde{A}_1^{(1)} $ equation \eqref{eq:qA1t}, with redefinitions $ f\to 1/f $, $ g\to 1/g $.

\subsubsection{Bilinear relations from \texorpdfstring{$ \mathrm{SU}(2)+1\mathbf{F} $}{SU(2)+1F}}

As the next simplest example, let us consider the 5d $ \mathrm{SU}(2) $ gauge theory coupled to one fundamental hypermultiplet, which has $ E_2 = \mathrm{SU}(2) \times \mathrm{U}(1) $ global symmetry. We relate this theory to the $ q $-Painlev\'e equation $ q\text{-}E_2^{(1)} $. The unrefined prepotential and 1-loop contributions of the vector and hypermultiplet are given by the $ N_f=1 $ specialization of \eqref{eq:su2F-unref-pre} and \eqref{eq:su2F-unref-pert}. The small $ a=e^{-\phi} $ expansion of the unrefined partition function is given by
\begin{align}\label{eq:Nf1-invarZ}
    \hat{Z} = e^{\mathcal{E}} \cdot \left( 1 + \frac{q}{(1-q)^2} \left( M_1 + \frac{1}{M_1} + \mathfrak{q} \sqrt{M_1} \right) a + \mathcal{O}(a^2) \right) \, ,
\end{align}
where $ M_1=e^{-m_1} $ and the first order in $ a $ represents the three hypermultiplet contributions. The possible magnetic flux quantizations from \eqref{eq:Nf-quantization} are given by
\begin{align}
    \begin{aligned}\label{eq:Nf1-quantization}
        & r = 0 \, , \quad
        B \in \left( \mathbb{Z} + \frac{1}{4}, \, 2\mathbb{Z} + \frac{1}{2} \right) \quad \text{or} \quad
        B \in \left( \mathbb{Z} - \frac{1}{4}, \, 2\mathbb{Z} - \frac{1}{2} \right) \, , \\
        & r = \frac{1}{2} \, , \quad
        B \in \left( \mathbb{Z}, \, 2\mathbb{Z} \right) \quad \text{or} \quad
        B \in \left( \mathbb{Z} + \frac{1}{2}, \, 2\mathbb{Z} + 1 \right) \, .
    \end{aligned}
\end{align}
We then fix a set $ \mathcal{S} $ of candidate magnetic fluxes $ (r,B) $ satisfying \eqref{eq:Nf1-quantization} with a chosen bound $ |b_0|, |b_1| < b_{\mathrm{max}} = 3 $.

The bilinear relations of tau-functions can be bootstrapped as in section~\ref{sec:Nf0-bilinear}. For a given trial total flux $ \Phi $, we find the set $ \mathcal{B} $ of all possible tuples of magnetic fluxes $ \{(r_i,B_i), (r_i', B_i')\} $. By considering the prepotential and 1-loop contributions in the bilinear relations, we can completely fix the coefficients $ \Lambda_i $ of the ansatz. If we set a total flux $ \Phi=(0;0,0) $, we obtain bilinear relations of tau-functions
\begin{gather}
    \begin{gathered}\label{eq:Nf1-tau}
        \tau_1 = \tauftn{0}{-1/4,-1/2} \, , \quad
        \tau_2 = \tauftn{0}{1/4, 1/2} \, , \quad
        \tau_3 = \tauftn{1/2}{0,0} \, , 
    \end{gathered}
\end{gather}
given by
\begin{align}\label{eq:Nf1-bilinear1}
    \overline{\tau_1} \underline{\tau_2} = \tau_1 \tau_2 + \sqrt{\mathfrak{q}{M_1}} \tau_3^2 \, , \quad
    \overline{\tau_3} \underline{\tau_3} = \tau_3^2 + \sqrt{\mathfrak{q}} \tau_1 \tau_2 \, ,
\end{align}
as well as their linear combinations, where $ \overline{\tau_i} = \tau_i(m_0+\epsilon) $ and $ \underline{\tau_i} = \tau_i(m_0-\epsilon) $ represent the discrete time evolution. We have checked these bilinear relations using the instanton partition function of the $ \mathrm{SU}(2)+1\mathbf{F} $ theory up to 3-instanton order. Additionally, the tau-functions 
\begin{align}
    \tau_4 = \tauftn{1/2}{-1/2, -1} \, , \quad
    \tau_5 = \tauftn{1/2}{1/2, 1} \, ,
\end{align}
lead to the following candidate bilinear relations:
\begin{align}
    \begin{aligned}\label{eq:Nf1-wrong-bilinear}
        \tau_4 \tau_5 = \overline{\tau_4} \underline{\tau_5} + \sqrt{\mathfrak{q}} \left( M_1 - \frac{1}{M_1} \right) \overline{\tau_1} \underline{\tau_2} \, , \quad
        \tau_4 \tau_5 = \overline{\tau_4} \underline{\tau_5} + \sqrt{\mathfrak{q}} \left( M_1 - \frac{1}{M_1} \right) \tau_1 \tau_2 \, , \\
        \sqrt{\mathfrak{q}} (1-M_1^2) \underline{\tau_1} \overline{\tau_2} - \mathfrak{q}\sqrt{M_1} (1-\sqrt{M_1}) \tau_4 \tau_5 - M_1 (1-\mathfrak{q}M_1^{3/2}) \overline{\tau_4} \underline{\tau_5} = 0 \, .
    \end{aligned}
\end{align}
Each of these bilinear relations holds in the $ \phi \to \infty $ and $ m_0 \to \infty $ limits. However, these bilinear relations are mutually inconsistent, as the first two relations imply $ \tau_1\tau_2=\overline{\tau_1}\underline{\tau_2} $. Together with \eqref{eq:Nf1-bilinear1}, this would imply that $ \tau_3 $ vanishes identically for all parameters, which is false already at the level of the prepotential and perturbative parts. Indeed, the three bilinear relations in \eqref{eq:Nf1-wrong-bilinear} turn out to be inconsistent once the $ a^1 $ information in \eqref{eq:Nf1-invarZ} is included.

We can find more bilinear relations by changing the total flux. In the case of total flux $ \Phi=(0;1,0) $, we find one additional bilinear relation:
\begin{align}\label{eq:Nf1-bilinear2}
    \tau_1 \overline{\tau_2} = \overline{\tau_1} \tau_2 + q^{1/4} \sqrt{\mathfrak{q}M_1} \overline{\tau_3} \tau_3 \, .
\end{align}
Although it is possible to construct more bilinear relations by changing the total flux, the three bilinear relations \eqref{eq:Nf1-bilinear1} and \eqref{eq:Nf1-bilinear2} are already sufficient to describe the $ q $-Painlev\'e equation $ q\text{-}E_2^{(1)} $. This can be checked by defining
\begin{align}
    f = \sqrt{\mathfrak{q}} \frac{\tau_3^2}{\tau_1 \tau_2} \, , \quad
    g = q^{1/4} \sqrt{\mathfrak{q}} \frac{\overline{\tau_3} \tau_3}{\overline{\tau_1} \tau_2} \, .
\end{align}
It is straightforward to check that $ (f,g) $ satisfies
\begin{align}
    f\overline{f} = \frac{g^2}{1+\sqrt{M_1} g} \, , \quad
    g\underline{g} = \frac{f(f+\mathfrak{q})}{1+\sqrt{M_1}f} \, .
\end{align}
These equations form the $ q\text{-}E_2^{(1)} $ system \eqref{eq:qE2}.

\section{Applications}\label{sec:applications}

In this section, we bootstrap bilinear relations of discrete Painlev\'e systems associated with 5d $ \mathrm{SU}(2) $ gauge theories coupled to $ N_f $ fundamental hypermultiplets. These theories have enhanced global symmetry $E_{N_f+1}$ for $N_f\le 7$ and affine symmetry $E_8^{(1)}$ for $N_f=8$. Each bilinear relation that we find is tested using the instanton partition function discussed in section~\ref{sec:5d} and Appendix~\ref{app:adhm}. Alternatively, the bilinear relations are also verified using the BPS data with manifest enhanced global symmetry obtained in \cite{Huang:2013yta}.

\subsection{E-string and elliptic Painlev\'e equation}\label{sec:elliptic}

We first consider the elliptic Painlev\'e equation $ e\text{-}E_8^{(1)} $ and the 5d $ \mathrm{SU}(2) $ gauge theory coupled to 8 fundamental hypermultiplets. The elliptic Painlev\'e equation is the most general member of the hierarchy of discrete Painlev\'e equations. This 5d gauge theory is UV-completed by the 6d $ \mathcal{N}=(1,0) $ SCFT with $ E_8 $ symmetry, known as the E-string theory. We begin with the unrefined effective prepotential \eqref{eq:E-Estr} of the E-string theory:
\begin{align}\label{eq:E-Estr2}
    \mathcal{E}_{\text{E-string}} = -\frac{1}{\epsilon^2} \left( \frac{\mu_0}{2} \varphi^2 - \frac{1}{2} \sum_{l=1}^8 \mu_l^2 \varphi + \frac{\epsilon^2}{2} \varphi \right) \, .
\end{align}
The set $ \mathcal{S}_{\text{6d}} $ of candidate magnetic fluxes is given by
\begin{align}\label{eq:Estr-flux}
    \begin{aligned}
        \mathcal{S}_{\text{6d}} = \bigg\{ (r, B) \, &\bigg| \, r = \frac{1}{2} \, , B=(b_0, \alpha) \, , b_0 \in \mathbb{Z} , \, |b_0| \leq b_{\mathrm{max}}, \\
        &\qquad \qquad \, \langle \alpha, v \rangle \in \mathbb{Z} \text{ for } v \in \Delta(E_8) , \, \lVert \alpha \rVert \leq \alpha_{\mathrm{max}} \bigg\} ,
    \end{aligned}
\end{align}
where $ r $ is the fractional part of the dynamical magnetic flux associated with $ \varphi $, $ B=(b_0, \alpha) $ denotes the background magnetic fluxes associated with $ (\mu_0, \mu_{l=1,\cdots,8}) $, and $ \Delta(E_8) $ is the root system of the $ E_8 $ algebra, expressed in an orthonormal basis:
\begin{align}
    \Delta(E_8) = \{ 0, \pm e_i \pm e_j \mid i\neq j \} \cup \left\{ \frac{1}{2} \sum_{i=1}^8 (\pm e_i ) \,\bigg|\, \text{even number of }(-) \text{ signs} \right\} \, .
\end{align}
Here, $ e_i \in \mathbb{R}^8 $ denotes an element of the orthonormal basis with $ \langle e_i, e_j \rangle = \delta_{ij} $. The magnetic flux quantization in \eqref{eq:Estr-flux} can be derived either from the leading order of the 6d partition function on the tensor branch \cite{Huang:2017mis, Kim:2020hhh, Kim:2014dza}, or from \eqref{eq:Nf-quantization} together with the mapping \eqref{eq:map-6d5d}. For $ \alpha_{\mathrm{max}}=\sqrt{2} $, the possible values of $ \alpha $ also form the $ E_8 $ root system. Let us further restrict to $ b_{\mathrm{max}}=0 $. Then, if we construct an ansatz by choosing magnetic fluxes $ (r_i,B_i) $ and $ (r_i',B_i') $ in $ \mathcal{S}_{\text{6d}} $ satisfying the total flux condition \eqref{eq:total-flux-cond-B}, together with $ \alpha_i, \alpha_i' \neq 0 $, $ \mu_0 $ becomes the modular parameter of the bilinear relation in the $ \varphi\to \infty $ limit:
\begin{align}\label{eq:Estr-tau-pert}
    \tauftnd{r_i}{B_i} \tauftnd{r_i'}{B_i'} = \sum_{N\in \mathbb{Z}} f(N) s^N \left( \theta_{I(N)}\left(2\mu_0, z_i \right) + \mathcal{O}(e^{-\varphi}) \right) \, ,
\end{align}
where $ \tau_{\mathrm{6d}} $ is the tau-function defined in terms of the partition function of the 6d E-string theory, $ f(N) $ is a common factor for all $ i=1,2,3 $, $ \theta_{I(N)} $ is the Jacobi theta function defined in \eqref{eq:theta}, and
\begin{align}
    I(N) = \left\{ \begin{array}{ll}
            2 & \ (N=\text{even}) \\
            3 & \ (N=\text{odd}) \, ,
    \end{array}\right.
    \qquad
    z_i = \langle \alpha_i - \alpha_i', \mu \rangle \, ,
\end{align}
with the inner product $ \langle \cdot, \cdot \rangle $ between 8-dimensional vectors. We treat cases in which $ \alpha_i $ or $ \alpha_i' $ is zero separately. Since the bilinear relation must hold for each power of $ s $, the coefficients $ \Lambda_i $ are determined from \eqref{eq:Estr-tau-pert} as
\begin{align}
    \Lambda_i = \theta_1\Big(\mu_0, \frac{z_j+z_k}{2}\Big) \theta_1\Big(\mu_0, \frac{z_j-z_k}{2} \Big) \, , 
\end{align}
where $ (i,j,k)=(1,2,3) $ and its cyclic permutations. Consequently, the $ \varphi \to \infty $ limit of the bilinear relation reduces to the following identity for Jacobi theta functions:
\begin{align}
    \theta_1\Big(\mu_0, \frac{z_j+z_k}{2}\Big) \theta_1\Big(\mu_0, \frac{z_j-z_k}{2}\Big) \theta_I(2\mu_0, z_i) + \text{cyclic in }(i,j,k) = 0 \, ,
\end{align}
for $ I=2, 3 $.

Although the leading order of the ansatz for the bilinear relations in the $ \varphi\to \infty $ limit yields identities of Jacobi theta functions for every choice of magnetic fluxes in \eqref{eq:Estr-flux} with $ b_{\mathrm{max}}=0 $, $ \alpha_{\mathrm{max}}=\sqrt{2} $ and $ \alpha \neq 0 $, some candidate relations fail at higher orders in $ e^{-\varphi} $. We eliminate these bilinear relations using the $ E_8 $ Weyl group symmetry and 5d perturbative data. Let us denote by $ \Phi_{\mathrm{6d}} = (0; 0, \lambda) = (0; 0, \alpha_i + \alpha_i') $ the total magnetic flux of the bilinear relation constructed from the 6d E-string partition function. Since the full partition function $ \hat{Z} $ is invariant under the $ E_8 $ Weyl group action, the tau-function transforms under the Weyl group $ W(E_8) $ as
\begin{align}
    w\Big(\tauftnd{0}{0,\alpha}\Big) = \tauftnd{0}{0,w(\alpha)} \, , \quad w \in W(E_8) \, .
\end{align}
Consequently, $ w \in W(E_8) $ maps a bilinear relation to another bilinear relation. For this reason, it is enough to consider the possible values of $ \lambda $ up to Weyl transformations. Scanning all pairs with $ \alpha_i, \alpha_i' \in \Delta(E_8) $ yields five inequivalent choices
\begin{align}
    \lambda = 0, \, e_7+e_8 , \, 2e_7+2e_8 , \, e_6+e_7+2e_8 , \, 2e_8 \, .
\end{align}
However, for $ \lambda=2e_7+2e_8 $ and $ \lambda=e_6+e_7+2e_8 $, only one flux pair $ (\alpha_i, \alpha_i') $ satisfies $ \alpha_i+\alpha_i' = \lambda $. These cases therefore admit no three-term bilinear relations.

The remaining three cases admit non-trivial bilinear relations. There are 120, 28, and 7 possible pairs $ (\alpha_i, \alpha_i') $ of magnetic fluxes for $ \lambda=0 $, $ e_7+e_8 $, and $ 2e_8 $, respectively. We further constrain the magnetic fluxes using the 5d perturbative partition function \eqref{eq:su2F-unref-pert}. Under the reparametrization \eqref{eq:map-6d5d}, the unrefined prepotential \eqref{eq:E-Estr2} maps to \eqref{eq:su2F-unref-pre} with $ N_f=8 $, up to terms independent of the dynamical parameter $ \phi $. The magnetic fluxes that we are considering are also mapped under \eqref{eq:Estr-flux} as
\begin{align}\label{eq:8F-flux}
    \mathcal{S}_{\text{5d}} = \left\{ (0,B) \, \bigg| \, 
    B = (0, \alpha) , \ \langle \alpha, e_8 \rangle \in \mathbb{Z}+\frac{1}{2} \right\} \cup
    \left\{ \Big( \frac{1}{2},B \Big) \, \bigg| \, B=(0, \alpha) , \ \langle \alpha, e_8 \rangle \in \mathbb{Z}  \right\} .
\end{align}
Here, the first component of $ (r,B) \in \mathcal{S}_{\text{5d}} $ denotes the fractional part of the dynamical magnetic flux for $ \phi $, and $ \alpha $ is the background magnetic flux associated with the eight mass parameters $ m_{1,\cdots,8} $. Using the prepotential \eqref{eq:su2F-unref-pre} and the perturbative partition function \eqref{eq:su2F-unref-pert} of $ \mathrm{SU}(2)+8\mathbf{F} $, we test whether the ans\"atze for the bilinear relations with the total fluxes $ \lambda=0, e_7+e_8 $ and $ 2e_8 $ are consistent in the $ m_0\to \infty $ limit.

For $ \lambda=0 $ and $ \lambda=e_7+e_8 $, the $ m_0 \to \infty $ limit rules out some flux choices but does not exclude all of them. To show that all ans\"atze within these classes are inconsistent, we use the following lemma: two bilinear relations constructed from the magnetic fluxes $ \{(\alpha_i, \alpha_i')\} $ and $ \{(\beta_i, \beta_i')\} $ are connected to each other by a Weyl transformation if the Gram matrices $ \langle x_i, x_j \rangle $ and $\langle y_i, y_j \rangle $ are identical, where $ x_i=\alpha_i-\alpha_i' $ and $ y_i=\beta_i-\beta_i' $. In other words, there exists a Weyl group element $ w \in W(E_8) $ such that $ w(\alpha_i)=\beta_i $ and $ w(\alpha_i') = \beta_i' $ if $ \langle x_i, x_j \rangle = \langle y_i, y_j \rangle $ for all $ i $ and $ j $. The proof of this lemma is given in Appendix~\ref{app:lemma}.

This lemma implies that if the 5d perturbative data exclude one magnetic flux set, they also exclude every magnetic flux set with the same Gram matrix $ \langle x_i, x_j \rangle $, even when some of these sets appear consistent at the leading order as $ m_0 \to \infty $. In the case of $ \lambda = e_7+e_8 $, there are two inequivalent Gram matrices up to permutations of the magnetic fluxes:
\begin{align}
        \langle x_i, x_j \rangle =
        \begin{pmatrix}
            6 & 2 & 2 \\  2 & 6 & 2 \\ 2 & 2 & 6
        \end{pmatrix} ,
        \begin{pmatrix}
            6 & -2 & 2 \\  -2 & 6 & 2 \\ 2 & 2 & 6
        \end{pmatrix} .
\end{align}
For each Gram matrix, we find a set of magnetic fluxes that leads to an inconsistency in the ansatz in the $ m_0 \to \infty $ limit. Similarly, there are five inequivalent Gram matrices
\begin{align}
    \langle x_i, x_j \rangle = 
        \begin{pmatrix}
            8 & 4 & 4 \\ 4 & 8 & 4 \\ 4 & 4 & 8
        \end{pmatrix} ,
        \begin{pmatrix}
            8 & 0 & 4 \\ 0 & 8 & 4 \\ 4 & 4 & 8
        \end{pmatrix} ,
        \begin{pmatrix}
            8 & 0 & 0 \\ 0 & 8 & 4 \\ 0 & 4 & 8
        \end{pmatrix} ,
        \begin{pmatrix}
            8 & 0 & 0 \\ 0 & 8 & 0 \\ 0 & 0 & 8
        \end{pmatrix} ,
        \begin{pmatrix}
            8 & -4 & 4 \\ -4 & 8 & 4 \\ 4 & 4 & 8
        \end{pmatrix} ,
\end{align}
up to permutations of the magnetic fluxes, in the case of $ \lambda=0 $. We find a set of magnetic fluxes for each Gram matrix whose ansatz is inconsistent in the $ m_0 \to \infty $ limit. Thus, there are no consistent bilinear relations for the total fluxes $ \lambda=e_7+e_8 $ and $ \lambda=0 $.

On the other hand, all ans\"atze associated with $ \lambda=2e_8 $ are consistent with the $ m_0 \to \infty $ limit. There are 7 pairs of magnetic fluxes $ \mathcal{B} = \{(r_i,B_i), (r_i',B_i')\} $ in this class:
\begin{align}
    r_i = \frac{1}{2} \, , \ B_i = (0, e_i+e_8) \, , \quad
    r_i' = \frac{1}{2} \, , \ B_i' = (0, -e_i+e_8) \, , \quad
    (i = 1,2,\cdots,7) \, .
\end{align}
There is only one inequivalent Gram matrix, $ \langle x_i, x_j \rangle = 4\mathbb{I}_3 $, where $ \mathbb{I}_3 $ is the $ 3\times 3 $ identity matrix. Accordingly, no inconsistency arises for any ansatz in this class, and the bilinear relations take the form
\begin{align}\label{eq:elliptic-bilinear}
    \theta_1(\mu_0, \mu_j\pm \mu_k) \tauftnd{1/2}{0, e_i+e_8} \tauftnd{1/2}{0, -e_i+e_8} + (i,j,k)\text{-cyclic} = 0 \, ,
\end{align}
where we use the shorthand notation $ \theta_1(\mu_0, \mu_j\pm \mu_k) = \theta_1(\mu_0, \mu_j+\mu_k) \theta_1(\mu_0, \mu_j-\mu_k) $, and $ i,j,k=1,\cdots, 7 $ are distinct integers. We have checked this class of bilinear relations up to 3-instanton order in the 5d description, and also up to order $ e^{-3\varphi} $ in the $ e^{-\varphi} $ expansion using the elliptic genus of the E-string theory.

Finally, we consider the case $ \alpha=0 $, corresponding to one of the indices $ i,j,k $ being $ 8 $ in \eqref{eq:elliptic-bilinear}. In this case, requiring the total flux $ \lambda=2e_8 $, we must include the magnetic flux $ B=(0,2e_8) $ among the allowed magnetic fluxes, which satisfies the quantization condition \eqref{eq:Estr-flux}. The prepotential contribution to the ansatz yields the bilinear relations
\begin{align}\label{eq:elliptic-bilinear-affine}
    \theta_1(\mu_0, \tilde{\mu}_j \pm \tilde{\mu}_k) \tauftnd{1/2}{0, e_i+e_8} \tauftnd{1/2}{0, -e_i+e_8} + (i,j,k)\text{-cyclic} = 0 \, ,
\end{align}
where $ i,j,k=1,\cdots,8 $ are distinct integers, and $ \tilde{\mu}_i $ is defined by $ \tilde{\mu}_i=\mu_i + \delta_{i8} \epsilon $. These bilinear relations remain consistent after the inclusion of the 5d perturbative partition function. We have also checked these bilinear relations up to 3-instanton order in the 5d description, and up to order $ e^{-3\varphi} $ in the $ e^{-\varphi} $ expansion using the elliptic genus of the E-string theory. In fact, the bilinear relations containing the magnetic fluxes $ (\alpha,\alpha')=(0,2e_8) $ can be mapped to the other bilinear relations without this flux by an \emph{affine} Weyl reflection of the $ E_8^{(1)} $ algebra. We present this in Appendix~\ref{app:elliptic}.

We now express the bilinear relations in the 5d $ \mathrm{SU}(2)+8\mathbf{F} $ parametrization. We split the bilinear relations into two cases, based on the quantization of the first component of the total flux, one with $ \Phi = (0; 0,2e_8) $ and the other with $ \Phi = (1/2; 0, \rho+e_8) $, where $ \rho = -e_1 + \sum_{i=1}^8 e_i/2 $. For each case, the bilinear relations can be represented by
\begin{align}
    \begin{aligned}\label{eq:elliptic-bilinear-5d}
        &\Lambda_{ijk}^{(1)} \tauftn{1/2}{0, e_i+e_8} \tauftn{1/2}{0, -e_i+e_8} + (i,j,k)\text{-cyclic} = 0 \, , \\
        &\Lambda_{ijk}^{(2)} \tauftn{1/2}{0,\sigma_i e_i+e_8} \tauftn{0}{0,\rho- \sigma_i e_i} + (i,j,k)\text{-cyclic} = 0 \, ,
    \end{aligned}
\end{align}
where $ \tau $ denotes the tau-function defined in terms of the 5d $ \mathrm{SU}(2)+8\mathbf{F} $ partition function, $ i,j,k=1,\cdots,8 $ are distinct integers, and $ \sigma_i = (-1)^{\delta_{i1}+\delta_{i8}} $. The coefficients of the bilinear relations are determined as
\begin{align}
    \Lambda_{ijk}^{(1)} = \theta_1(2m_0, \tilde{m}_j \pm \tilde{m}_k) \, , \quad
    \Lambda_{ijk}^{(2)} = C_i^{(2)} \theta_1(2m_0, \tilde{m}_k-\tilde{m}_j) \theta_1(2m_0, \mathfrak{m}-\tilde{m}_j-\tilde{m}_k) \, ,
\end{align}
where we use the notation
\begin{gather}
    \begin{gathered}
        \tilde{m}_1 = -m_1 \, , \quad
        \tilde{m}_i = m_i \, , \ (2 \leq i \leq 7) \, , \quad
        \tilde{m}_8 = m_8 + \epsilon \, , \\
        C_i^{(2)} = e^{-\tilde{m}_i/2} \, , \quad
        \mathfrak{m} = -m_0 - m_1 - m_8 + \frac{1}{2}\sum_{l=1}^8 m_l \, .
    \end{gathered}
\end{gather}
The bilinear relations \eqref{eq:elliptic-bilinear-affine}, or \eqref{eq:elliptic-bilinear-5d} are equivalent to the elliptic Painlev\'e equation \cite{Kajiwara:2003, Kajiwara:2004, Masuda:2011, Kajiwara:2015}. For the $ q $-Painlev\'e equations discussed in section~\ref{sec:easy-example}, which correspond to pure $ \mathrm{SU}(2) $ gauge theories or to the theory with one fundamental hypermultiplet, the discrete time evolution is realized by shifting the instanton fugacity $ \mathfrak{q} $. In contrast, the time evolution of the elliptic Painlev\'e equation is given by a translation element of the affine Weyl group. The equivalence of the bilinear relations and the elliptic Painlev\'e equation, including the precise definition of the discrete time evolution, is discussed in Appendix~\ref{app:elliptic}.

\subsection{\texorpdfstring{$ \mathrm{SU}(2) + 7\mathbf{F} $}{SU(2) + 7F}} \label{sec:7F}

In this subsection, we consider the $ q $-Painlev\'e equation $ q\text{-}E_8^{(1)} $ and the 5d $ \mathrm{SU}(2) $ gauge theory coupled to 7 fundamental hypermultiplets, which has $ E_8 $ global symmetry. The unrefined effective prepotential and perturbative partition function are given in \eqref{eq:su2F-unref-pre} and \eqref{eq:su2F-unref-pert} with $ N_f=7 $. This theory can be obtained by decoupling one fundamental hypermultiplet from the $ \mathrm{SU}(2)+8\mathbf{F} $ theory. At the level of the partition function, this decoupling is achieved by taking the $ m_0 - m_8/2 \to \infty $ limit. As in the E-string theory, the reparametrization given by
\begin{align}
    \phi = \varphi + \mu_8 \, , \quad
    m_i = \mu_i \, , \quad (1 \leq i \leq 7) \, , \quad
    m_0 = -\frac{1}{2}\mu_8 \, ,
\end{align}
yields the $ E_8 $-invariant form of the prepotential:
\begin{align}
    \mathcal{E}_{\mathrm{inv}} = -\frac{1}{\epsilon^2} \left( \frac{1}{6}\varphi^3 - \frac{1}{2} \sum_{l=1}^8 \mu_l^2 \varphi + \frac{5\epsilon^2}{12}\varphi \right) \, .
\end{align}
The $ q $-Painlev\'e equation $ q\text{-}E_8^{(1)} $ and its bilinear relations can be deduced from the elliptic Painlev\'e equation and its bilinear relations. Since the $ \mathrm{SU}(2)+7\mathbf{F} $ theory has $ E_8 $ symmetry, we can use a method similar to that used for the E-string theory in section~\ref{sec:elliptic}. Using \eqref{eq:8F-flux}, let us consider the following set of magnetic fluxes:
\begin{align}
    \mathcal{S} = \left\{ (r,B) \, \bigg| \, r \equiv \frac{1}{2} + \langle \alpha, e_8 \rangle \bmod{\mathbb{Z}} , \ B = \Bigl(-\frac{1}{2} \langle \alpha, e_8 \rangle, \langle \alpha, e_i \rangle_{i=1}^7\Bigr) , \ (r;0,\alpha) \in \mathcal{S}_{\mathrm{5d}}^{8\mathbf{F}} \right\} ,
\end{align}
where $ \mathcal{S}_{\mathrm{5d}}^{8\mathbf{F}} $ is the set of magnetic fluxes for the $ \mathrm{SU}(2)+8\mathbf{F} $ theory. Here, $ r $ and $ B $ denote, respectively, the fractional part of the dynamical magnetic flux associated with $ \phi $, and the background magnetic flux associated with the mass parameters $ (m_0,m_1,\cdots,m_7) $.

As in the $ \mathrm{SU}(2)+8\mathbf{F} $ theory, it is enough to consider the bilinear relations up to $ E_8 $ Weyl transformations. We then find that the bilinear relations arising from the partition function of the $ \mathrm{SU}(2)+7\mathbf{F} $ theory are trigonometric versions of the bilinear relations of the elliptic Painlev\'e equation. For the choice of total flux $ \Phi = (0; -1,0,\cdots,0) $, the prepotential and perturbative partition function of the $ \mathrm{SU}(2)+7\mathbf{F} $ theory yield the bilinear relations
\begin{align}
    \begin{aligned}
        &\Lambda_{ijk}^{(1)} \tauftn{1/2}{-1/2,e_i} \tauftn{1/2}{-1/2,-e_i} + (i,j,k)\text{-cyclic} = 0 \, , \\
        &\Lambda_{0jk}^{(1)}\tauftn{1/2}{0,0} \tauftn{1/2}{-1,0} + \Lambda_{jk0}^{(1)} \tauftn{1/2}{-1/2,e_j} \tauftn{1/2}{-1/2,-e_j} + \Lambda_{k0j}^{(1)} \tauftn{1/2}{-1/2,e_k} \tauftn{1/2}{-1/2,-e_k} = 0  \, ,
    \end{aligned}
\end{align}
where $ i,j,k=1,\cdots,7 $ are distinct integers. The coefficients $ \Lambda_{ijk}^{(1)} $ are determined as
\begin{align}\label{eq:7F-bilinear1}
    \Lambda_{ijk}^{(1)} = \sinh\Big(\frac{\tilde{m}_j + \tilde{m}_k}{2}\Big) \sinh\Big(\frac{\tilde{m}_j - \tilde{m}_k}{2}\Big) \, ,
\end{align}
including the case in which one of the indices is $ 0 $, where
\begin{align}\label{eq:7F-tildem}
    \tilde{m}_0 = 2m_0 - \epsilon \, , \quad
    \tilde{m}_1 = -m_1 \, , \quad
    \tilde{m}_i = m_i \, , \quad (i>1) \, .
\end{align}
This is the trigonometric version of the first elliptic bilinear relation in \eqref{eq:elliptic-bilinear}. We have checked these bilinear relations up to 3-instanton order using the instanton partition function of the $ \mathrm{SU}(2)+7\mathbf{F} $ theory.

Another type of bilinear relations, corresponding to the second elliptic bilinear relations in \eqref{eq:elliptic-bilinear} can be obtained by considering the total flux $ \Phi = (\frac{1}{2}; -\frac{3}{4}, \rho) $, where $ \rho = -e_1 + \sum_{i=1}^7 e_i/2 $. Using the prepotential and perturbative partition function of the $ \mathrm{SU}(2)+7\mathbf{F} $ theory we find the following bilinear relations:
\begin{align}\label{eq:7F-bilinear2}
    &\Lambda_{ijk}^{(2)} \tauftn{1/2}{-1/2, e_i} \tauftn{0}{-1/4, \rho-e_i} + \text{cyclic in } (i,j,k) = 0 \, , \\
    &\Lambda_{1jk}^{(2)} \tauftn{1/2}{-1/2, -e_1} \tauftn{0}{-1/4, \rho+e_1} + \Lambda_{jk1}^{(2)} \tauftn{1/2}{1/2,e_j} \tauftn{0}{-1/4,\rho-e_j} + \Lambda_{k1j}^{(2)} \tauftn{1/2}{1/2,e_k} \tauftn{0}{-1/4,\rho-e_k} = 0 \, , \nonumber \\
    &\Lambda_{0jk}^{(2)} \tauftn{1/2}{0,0} \tauftn{0}{-3/4,\rho} + \Lambda_{jk0}^{(2)} \tauftn{1/2}{1/2,e_j} \tauftn{0}{-1/4,\rho-e_j} + \Lambda_{k0j}^{(2)} \tauftn{1/2}{1/2,e_k} \tauftn{0}{-1/4,\rho-e_k} = 0 \, , \nonumber
\end{align}
where $ i,j,k = 2,3,\cdots,7 $ are distinct integers. The coefficients $ \Lambda_{ijk}^{(2)} $ are given by
\begin{align}\label{eq:7F-bilinear2-Lambda}
    \Lambda_{ijk}^{(2)} &= C_i \sinh\Big(\frac{\tilde{m}_k-\tilde{m}_j}{2}\Big) \sinh\Big(\frac{\mathfrak{m} - \tilde{m}_j - \tilde{m}_k}{2}\Big) \, ,
\end{align}
for all distinct integers $ i,j,k=0,1,\cdots,7 $, where $ \tilde{m}_i $ are defined in \eqref{eq:7F-tildem} and
\begin{equation}\label{eq:7F-C}
    C_0 = \frac{q^{9/8}\prod_{l=2}^7 \sqrt{M_l}}{\mathfrak{q} \sqrt{M_1}} , \quad C_1 = \frac{1}{\sqrt{M_1}} , \quad C_i = \sqrt{M_i} ,  \quad
    \mathfrak{m} = m_0 - m_1 + \frac{1}{2} \sum_{l=1}^7 m_l ,
\end{equation}
for $ \mathfrak{q}=e^{-m_0} $ and $ M_i=e^{-m_i} $. We have checked these bilinear relations using the instanton partition function of the $ \mathrm{SU}(2)+7\mathbf{F} $ theory, up to 3-instanton order. The bilinear relations \eqref{eq:7F-bilinear1} and \eqref{eq:7F-bilinear2} are equivalent to the $ q $-Painlev\'e equation $ q\text{-}E_8^{(1)} $ given in \eqref{eq:qE8} \cite{Masuda:2011}. As in the elliptic Painlev\'e equation, the discrete time evolution is not simply a shift of the instanton fugacity $ \mathfrak{q} $, but is instead a translation element in the affine Weyl group $ W(E_8^{(1)}) $. We present the equivalence of the bilinear relations and the $ q\text{-}E_8^{(1)} $ Painlev\'e equation in Appendix~\ref{app:qE8}.

\subsection{\texorpdfstring{$ \mathrm{SU}(2) + 6\mathbf{F} $}{SU(2) + 6F}}\label{sec:6F}

We next consider the $ q $-Painlev\'e equation $ q\text{-}E_7^{(1)} $ and the 5d $ \mathrm{SU}(2)+6\mathbf{F} $ theory, which has enhanced $ E_7 $ global symmetry. The unrefined effective prepotential and perturbative partition function are given by \eqref{eq:su2F-unref-pre} and \eqref{eq:su2F-unref-pert} with $ N_f=6 $. Although the candidate magnetic fluxes can again be obtained from those of the $ \mathrm{SU}(2)+8\mathbf{F} $ theory, here we consider the set
\begin{align}
    \mathcal{S} = \left\{ (r,B) \, \big| \, (r,B) \text{ satisfying \eqref{eq:Nf-quantization}}, \ |b_i| \leq 1 \text{ for } 0 \leq i \leq 6 \right\} \, .
\end{align}
Here, $ r $ is the fractional part of the dynamical magnetic flux associated with $ \phi $ and $ B=(b_i)_{i=0}^6 $ denotes the background magnetic fluxes associated with $ (m_i)_{i=0}^6 $.

By considering the total fluxes $ \Phi = (0; 1,0) $ and $ \Phi = (1/2; 1,\rho) $ as in $ \mathrm{SU}(2)+7\mathbf{F} $ theory, where $ \rho = -e_1 + \sum_{i=1}^6 e_i/2 $, we obtain the following consistent bilinear relations:
\begin{align}
    \begin{aligned}\label{eq:6F-bilinear1}
        & \Lambda_{ijk}^{(1)} \tauftn{1/2}{1/2,e_i} \tauftn{1/2}{1/2,-e_i} + (i,j,k)\text{-cyclic} = 0 \, , \\
        & \Lambda_{7jk}^{(1)} \tauftn{1/2}{0,0} \tauftn{1/2}{1,0} - \tauftn{1/2}{1/2,e_j} \tauftn{1/2}{1/2,-e_j} + \tauftn{1/2}{1/2,e_k} \tauftn{1/2}{1/2,-e_k} = 0 \, , \\
        &\Lambda_{ijk}^{(2)} \tauftn{1/2}{1/2, \sigma_i e_i} \tauftn{0}{1/2, \rho- \sigma_i e_i} + (i,j,k)\text{-cyclic} = 0 \, , \\
        &\Lambda_{7jk}^{(2)} \tauftn{1/2}{0,0} \tauftn{0}{1,\rho} + \Lambda_{jk7}^{(2)} \tauftn{1/2}{1/2, \sigma_j e_j} \tauftn{0}{1/2,\rho- \sigma_j e_j} + \Lambda_{k7j}^{(2)} \tauftn{1/2}{1/2, \sigma_k e_k} \tauftn{0}{1/2,\rho- \sigma_k e_k} = 0 \, , \\
        &\Lambda_{8jk}^{(2)} \tauftn{1/2}{1,0} \tauftn{0}{0,\rho} + \Lambda_{jk8}^{(2)} \tauftn{1/2}{1/2, \sigma_j e_j} \tauftn{0}{1/2,\rho- \sigma_j e_j} + \Lambda_{k8j}^{(2)} \tauftn{1/2}{1/2, \sigma_k e_k} \tauftn{0}{1/2,\rho- \sigma_k e_k} = 0 \, .
    \end{aligned}
\end{align}
Here, $ 1 \leq i, j, k \leq 6 $ are distinct integers, $ \sigma_i= (-1)^{\delta_{i1}} $, and the coefficients of the bilinear relations are given by
\begin{alignat}{2}
    &\Lambda_{ijk}^{(1)} = \Lambda_{7jk}^{(1)} = 4\sinh\Big(\frac{m_j + m_k}{2}\Big) \sinh\Big(\frac{m_j - m_k}{2}\Big) \, , \\
    &\Lambda_{ijk}^{(2)} = \sinh\Big(\frac{\sigma_j m_j - \sigma_k m_k}{2}\Big) \, ,
    &&\Lambda_{njk}^{(2)} = C_n^{(2)} \sinh\Big(\frac{\sigma_j m_j - \sigma_k m_k}{2}\Big) \nonumber \\
    &\Lambda_{jkn}^{(2)} = \sinh\Big(\frac{-\mathfrak{m} + m_1 + \sigma_k m_k + \gamma_n}{2}\Big) \, ,
    &&\Lambda_{knj}^{(2)} = \sinh\Big(\frac{\mathfrak{m} - m_1 - \sigma_j m_j - \gamma_n}{2}\Big) \, , \nonumber
\end{alignat}
for $ 1 \leq i,j,k \leq 6 $ and $ n=7,8 $, where we define
\begin{gather}
    \begin{gathered}\label{eq:6F-tildem}
        C_7^{(2)} = \frac{M_1^{1/4}}{q^{5/8} \sqrt{\mathfrak{q}} \prod_{l=2}^6 M_l^{1/4}} \, , \quad
        C_8^{(2)} = \frac{q^{1/8} \prod_{l=2}^6 M_l^{1/4}}{\sqrt{\mathfrak{q}} M_1^{1/4}} \, , \\
        \gamma_7 = -m_0-\epsilon \, , \quad
        \gamma_8 = m_0 \, , \quad
        \mathfrak{m} = \frac{1}{2}\sum_{l=1}^6 m_l \, .
    \end{gathered}
\end{gather}
The coefficients of the bilinear relations are determined from the $ \phi \to \infty $ and $ m_0 \to \infty $ limits of the ans\"atze. We further use the subleading order of the 1-instanton partition function \eqref{eq:1inst-leading} to rule out inconsistent bilinear relations. We have checked the bilinear relations \eqref{eq:6F-bilinear1} using the instanton partition function of the $ \mathrm{SU}(2)+6\mathbf{F} $ theory, up to 3-instanton order.

One can find many other bilinear relations by changing the total flux $ \Phi $. We present another type of bilinear relation, which comes from the total flux $ \Phi=(0;1,e_5+e_6) $. From the gauge theory data of the $ \mathrm{SU}(2)+6\mathbf{F} $ theory, we find
\begin{align}\label{eq:6F-bilinear3}
    &\Lambda_{ijkl}^{(3)} \tauftn{0}{1/2,\lambda-e_i-e_j} \tauftn{0}{1/2,\lambda-e_k-e_l} + (i,j,k)\text{-cyclic} = 0 \, , \\
    &\tauftn{0}{1/2,\lambda-e_1-e_4} \tauftn{0}{1/2,\lambda-e_2-e_3} - \tauftn{0}{1/2,\lambda-e_1-e_3} \tauftn{0}{1/2,\lambda-e_2-e_4} + \Lambda_{53}^{(3)} \tauftn{1/2}{1/2,e_5} \tauftn{1/2}{1/2,e_6} = 0 \, , \nonumber \\
    &\Lambda_{71}^{(3)} \tauftn{0}{1/2,\lambda-e_1-e_4} \tauftn{0}{1/2,\lambda-e_2-e_3} \!+\! \Lambda_{72}^{(3)} \tauftn{0}{1/2,\lambda-e_1-e_3} \tauftn{0}{1/2,\lambda-e_2-e_4} \!+\! \Lambda_{73}^{(3)} \tauftn{1/2}{0,e_5+e_6} \tauftn{1/2}{1,0} \!=\! 0 , \nonumber  \\
    &\Lambda_{81}^{(3)} \tauftn{0}{1/2,\lambda-e_1-e_4} \tauftn{0}{1/2,\lambda-e_2-e_3} \!+\! \Lambda_{82}^{(3)} \tauftn{0}{1/2,\lambda-e_1-e_3} \tauftn{0}{1/2,\lambda-e_2-e_4} \!+\! \Lambda_{83}^{(3)} \tauftn{1/2}{0,0} \tauftn{1/2}{1,e_5+e_6} \!=\! 0 , \nonumber
\end{align}
where $ 1 \leq i,j,k,l \leq 4 $ are distinct integers, $ \lambda=\sum_{i=1}^6 e_i/2 $, and the coefficients are determined by
\begin{align}
    \Lambda_{ijkl}^{(3)} &= \sinh\Big(\frac{m_i-m_j}{2}\Big) \sinh\Big(\frac{m_k-m_l}{2}\Big) , \
    \Lambda_{n3}^{(3)} = 4C_n^{(3)} \sinh\Big(\frac{m_1-m_2}{2}\Big) \sinh\Big(\frac{m_3-m_4}{2}\Big) , \nonumber \\
    \Lambda_{n1}^{(3)} &= 4\sinh\Big(\frac{\mathfrak{m}-m_1-m_3 - \gamma_n}{2}\Big) \sinh\Big(\frac{\mathfrak{m}-m_2-m_4 - \gamma_n}{2}\Big) \, , \nonumber \\
    \Lambda_{n2}^{(3)} &= -4\sinh\Big(\frac{\mathfrak{m}-m_1-m_4 - \gamma_n}{2}\Big) \sinh\Big(\frac{\mathfrak{m}-m_2-m_3 - \gamma_n}{2}\Big) \, .
\end{align}
Here, $ n=5,7,8 $, $ \gamma_n $ and $ \mathfrak{m} $ are defined in \eqref{eq:6F-tildem}, and
\begin{gather}
    C_5^{(3)} = q^{1/4} \sqrt{\mathfrak{q}} \, , \quad
    C_7^{(3)} = q^{1/4} \sqrt{\frac{M_5 M_6}{\mathfrak{q}}} \, , \quad
    C_8^{(3)} = \frac{1}{q^{3/4} \sqrt{\mathfrak{q}M_5M_6}} \, .
\end{gather}
The indices $ (1,2,3,4) $ can be permuted in the last three equations of \eqref{eq:6F-bilinear3}. The bilinear relations \eqref{eq:6F-bilinear1} and \eqref{eq:6F-bilinear3} are sufficient to describe the $ q\text{-}E_7^{(1)} $ system \eqref{eq:qE7}, where the discrete time evolution is given by a translation element of the affine $ E_7^{(1)} $ Weyl group. The equivalence is discussed in Appendix~\ref{app:E7}, by comparing our result with \cite{Masuda:2009}.

\subsection{\texorpdfstring{$ \mathrm{SU}(2) + 5\mathbf{F} $}{SU(2) + 5F}}

We next consider the bilinear relations coming from the 5d $ \mathrm{SU}(2)+5\mathbf{F} $ theory, which has $ E_6 $ global symmetry. The unrefined effective prepotential and perturbative partition function are given by \eqref{eq:su2F-unref-pre} and \eqref{eq:su2F-unref-pert} with $ N_f=5 $. To bootstrap the bilinear relations of tau-functions, we consider the set of candidate magnetic fluxes given by
\begin{align}
    \mathcal{S} = \left\{ (r,B) \, \big| \, (r,B) \text{ satisfying \eqref{eq:Nf-quantization}}, \ |b_0| \leq 3/2 , \ |b_i| \leq 1 , \ (1 \leq i \leq 5) \right\} \, ,
\end{align}
where $ r $ is the fractional part of the dynamical magnetic flux associated with $ \phi $ and $ B=(b_i)_{i=0}^5 $ denotes the background magnetic fluxes associated with the mass parameters $ (m_i)_{i=0}^5 $.

Due to the complexity of the bilinear relations, we present only a minimal set of bilinear relations that we expect to be equivalent to the $ q\text{-}E_6^{(1)} $ Painlev\'e equation. Systems of bilinear relations for a given total flux $ \Phi $ can be found in a similar manner to the previous subsections. Let us define nine tau-functions as
\begin{alignat}{3}
    \begin{aligned}\label{eq:5F-tau}
        &\tau_1^{(1)} = \tauftn{0}{-3/4, \rho-e_2} \, , \quad
        &&\tau_2^{(1)} = \tauftn{1/2}{-1/2,-e_1} \, , \quad
        &&\tau_3^{(1)} = \tauftn{0}{-1/4, \rho-e_2-e_3} \, , \\
        &\tau_1^{(2)} = \tauftn{0}{-3/4, \rho-e_3} \, , \quad
        &&\tau_2^{(2)} = \tauftn{1/2}{-1/2, e_4} \, , \quad
        &&\tau_3^{(2)} = \tauftn{1/2}{-1, 0} \, , \\
        &\tau_1^{(3)} = \tauftn{1/2}{0, 0} \, , \quad
        &&\tau_2^{(3)} = \tauftn{1/2}{-1/2, e_5} \, , \quad
        &&\tau_3^{(3)} = \tauftn{0}{-1/4, \rho} \, ,
    \end{aligned}
\end{alignat}
where $ \rho = -e_1 + \sum_{i=1}^5 e_i/2 $. These tau-functions then satisfy the nine bilinear relations
\begin{align}\label{eq:5F-bilinear}
    \Lambda_{i1}^{(j)} \tau_i^{(j)} \ell_j(\tau_{i+1}^{(j)}) + \Lambda_{i2}^{(j)} \ell_j(\tau_{i}^{(j)}) \tau_{i+1}^{(j)} + \Lambda_{i3}^{(j)} \tau_{i+2}^{(j+1)} \tau_{i+2}^{(j+2)} = 0 \, ,
\end{align}
where $ i,j=1,2,3 $ are understood modulo 3, and $ \ell_j $ acts on the tau-functions as
\begin{gather}
    \begin{gathered}
        \ell_1\left(\tauftn{r}{b_0, b_i e_i}\right) = \tauftn{r}{b_0, b_ie_i +e_1+e_2} \, , \\
        \ell_2\left(\tauftn{r}{b_0, b_i e_i}\right) = \tauftn{r+1/2}{b_0+3/4, b_i e_i+\rho-e_2-e_4} \, , \quad
        \ell_3\left(\tauftn{r}{b_0, b_i e_i}\right) = \tauftn{r+1/2}{b_0-3/4, b_i e_i-\rho-e_1+e_4} \, ,
    \end{gathered}
\end{gather}
These maps satisfy $ \ell_1 \ell_2 \ell_3 = 1 $. The coefficients $ \Lambda_{ik}^{(j)} $ are determined as
\begin{gather}
    \begin{gathered}
        \Lambda_{i1}^{(1)} = 1 \, , \qquad
        \Lambda_{12}^{(1)} = \Lambda_{22}^{(1)} = -\sqrt{M_1M_2} \, , \qquad
        \Lambda_{32}^{(1)} = -1 \, , \\
        \Lambda_{13}^{(1)} = \frac{-1 + M_1M_2}{q^{1/8}\sqrt{M_2}} \, , \quad
        \Lambda_{23}^{(1)} = q^{1/8} \frac{-1+M_1M_2}{\sqrt{M_1}} \, , \quad
        \Lambda_{33}^{(1)} = \frac{\sqrt{\mathfrak{q}} (-1+M_1M_2)}{q^{1/4} \sqrt{M_1M_2M_3}} \, ,
    \end{gathered}
\end{gather}
for $ j=1 $;
\begin{gather}
    \begin{gathered}
        \Lambda_{11}^{(2)} = \Lambda_{31}^{(2)} = 1 \, , \quad
        \Lambda_{21}^{(2)} = \frac{1}{q^{1/8} \sqrt{M_4}} \, , \quad
        \Lambda_{12}^{(2)} = -q^{-1/8} \sqrt{\frac{\mathfrak{q}M_3}{M_2}} \, , \\
        \Lambda_{22}^{(2)} = -\frac{1}{q^{1/4} M_4} \, , \quad
        \Lambda_{32}^{(2)} = -q^{1/4} \sqrt{\frac{\mathfrak{q}M_5}{M_1}} \, , \quad
        \Lambda_{i3}^{(2)} = -1 + \mathfrak{q} \sqrt{\frac{M_3M_5}{M_1M_2M_4}} \, ,
    \end{gathered}
\end{gather}
for $ j=2 $; and
\begin{gather}
    \begin{gathered}
        \Lambda_{11}^{(3)} = q^{1/4} M_5 \, , \quad
        \Lambda_{21}^{(3)} = q^{-3/8} \sqrt{\mathfrak{q}M_2M_3} \, , \quad
        \Lambda_{31}^{(3)} = q^{-3/4} \sqrt{\frac{\mathfrak{q}M_1}{M_4}} \, , \\
        \Lambda_{12}^{(3)} = -q^{1/8} \sqrt{M_5} \, , \quad
        \Lambda_{22}^{(3)} = \Lambda_{32}^{(3)} = -1 \, , \quad
        \Lambda_{i3}^{(3)} = 1 - \frac{\mathfrak{q}}{q} \sqrt{\frac{M_1M_2M_3M_5}{M_4}} \, ,
    \end{gathered}
\end{gather}
for $ j=3 $. We have checked that these bilinear relations remain valid in the presence of instanton corrections using the instanton partition function of the $ \mathrm{SU}(2)+5\mathbf{F} $ theory, up to 3-instanton order. These nine bilinear relations \eqref{eq:5F-bilinear} have exactly the same structure as the bilinear relations studied in \cite{Tsuda:2008}. Thus, we expect that the tau-functions defined from the 5d $ \mathrm{SU}(2)+5\mathbf{F} $ theory describe the $ q\text{-}E_6^{(1)} $ Painlev\'e system. The detailed comparision is presented in Appendix~\ref{app:E6}.

\subsection{\texorpdfstring{$ \mathrm{SU}(2) + 4\mathbf{F} $}{SU(2) + 4F}}\label{sec:4F}

In this subsection, we consider the bilinear relations of $ q\text{-}D_5^{(1)} $, which is also called the $ q $-Painlev\'e VI equation \cite{Jimbo:1995}, from the 5d $ \mathrm{SU}(2)+4\mathbf{F} $ theory, which has enhanced $ D_5 $ global symmetry. The bilinear relations of this $ q $-Painlev\'e equation were studied in \cite{Tsuda:2006, Jimbo:2017ael}, and we can reproduce them using the bootstrap method based on the 5d field theory data. The unrefined effective prepotential and perturbative partition function are given by \eqref{eq:su2F-unref-pre} and \eqref{eq:su2F-unref-pert} with $ N_f=4 $. We consider the set $ \mathcal{S} $ of magnetic fluxes given by
\begin{align}
    \mathcal{S} = \left\{ (r,B) \, \big| \, (r,B) \text{ satisfying \eqref{eq:Nf-quantization}}, \ |b_0| \leq 3/2 , \  |b_i| \leq 1 , \ (1 \leq i \leq 4) \right\} \, ,
\end{align}
where $ r $ is the fractional part of the dynamical magnetic flux associated with $ \phi $ and $ B = (b_0,b_1,\cdots,b_4) $ denotes the background magnetic fluxes associated with $ (m_0,m_1,\cdots,m_4) $.

Let us consider the total fluxes $ \Phi = (0; 0, e_3-e_4) $ and $ \Phi = (0; 1, e_3-e_4) $, where $ e_i $ are unit vectors of $ \mathbb{R}^4 $. Then, by considering the prepotential, the 1-loop perturbative partition function, and the leading order of the 1-instanton partition function which encodes the $ D_5 $ symmetry, we find that the eight tau-functions defined by
\begin{alignat}{4}
    \begin{aligned}
        &\tau_1 = \tauftn{0}{-1/2, \rho-e_1-e_2} \, , \
        &&\tau_2 = \tauftn{0}{1/2, \rho} \, , \
        &&\tau_3 = \tauftn{0}{-1/2, \rho-e_1} \, , \
        &&\tau_4 = \tauftn{0}{1/2, \rho-e_2} \, , \\
        &\tau_5 = \tauftn{1/2}{0,0} \, , \
        &&\tau_6 = \tauftn{1/2}{0,e_3-e_4} \, , \
        &&\tau_7 = \tauftn{1/2}{1/2,e_3} \, , \
        &&\tau_8 = \tauftn{1/2}{-1/2,-e_4} \, ,
    \end{aligned}
\end{alignat}
where $ \rho = -e_4 + \sum_{i=1}^4 e_i/2 $, satisfy
\begin{alignat}{2}
    \begin{aligned}\label{eq:4F-bilinear}
        &\tau_1 \tau_2 = \tau_3 \tau_4 + \sqrt{\mathfrak{q}M_1M_2} \tau_5 \tau_6 \, , \quad
        &&\overline{\tau_1} \underline{\tau_2} = \tau_3 \tau_4 + \sqrt{\frac{\mathfrak{q}}{M_1M_2}} \tau_5 \tau_6 \, , \\
        &\tau_7 \tau_8 = \tau_5 \tau_6 + \sqrt{\mathfrak{q}M_3M_4} \tau_3 \tau_4 \, , \quad
        &&\underline{\tau_7} \overline{\tau_8} = \tau_5 \tau_6 + \sqrt{\frac{\mathfrak{q}}{M_3M_4}} \tau_3 \tau_4 \, , \\
        & \tau_3 \overline{\tau_4} = \overline{\tau_1} \tau_2 + q^{1/4} \sqrt{\frac{\mathfrak{q} M_1}{M_2}} \tau_7 \overline{\tau_8} \, , \quad
        && \overline{\tau_3} \tau_4 = \overline{\tau_1} \tau_2 + q^{1/4} \sqrt{\frac{\mathfrak{q}M_2}{M_1}} \tau_7 \overline{\tau_8} \, , \\
        & \tau_5 \overline{\tau_6} = \tau_7 \overline{\tau_8} + q^{3/4} \sqrt{\frac{\mathfrak{q}M_3}{M_4}} \overline{\tau_1} \tau_2 \, , \quad
        && \overline{\tau_5} \tau_6 = \tau_7 \overline{\tau_8} + q^{-1/4} \sqrt{\frac{\mathfrak{q}M_4}{M_3}} \overline{\tau_1}\tau_2 \, .
    \end{aligned}
\end{alignat}
Here, the overbar and underbar denote the forward and backward discrete time evolutions, which are defined as
\begin{align}
    \overline{\tau_i} = \tau_i(m_0+\epsilon) \, , \quad
    \underline{\tau_i} = \tau_i(m_0-\epsilon) \, .
\end{align}
We remark that the discrete time evolution can be regarded as a simple shift of the instanton fugacity. We have checked that these bilinear relations remain valid after including the instanton corrections, using the instanton partition function of the $ \mathrm{SU}(2)+4\mathbf{F} $ theory up to 3-instanton order. 

The eight bilinear relations in \eqref{eq:4F-bilinear} are not the only bilinear relations that we can find. If one starts from other choices of total flux, one obtains additional bilinear relations for the tau-functions defined from a different set of magnetic fluxes. However, the bilinear relations \eqref{eq:4F-bilinear} are sufficient to describe the $ q\text{-}D_5^{(1)} $ system. By defining $ f $ and $ g $ in terms of tau-functions as
\begin{align}
    f = \sqrt{\mathfrak{q}} \frac{\tau_5 \tau_6}{\tau_3 \tau_4} \, , \quad
    g = q^{1/4} \sqrt{\mathfrak{q}} \frac{\tau_7 \overline{\tau_8}}{\overline{\tau_1} \tau_2} \, ,
\end{align}
we find that they satisfy the non-linear $ q $-difference equations
\begin{align}
    f\overline{f} = \frac{(g+q\mathfrak{q}\sqrt{M_3/M_4}) (g+\mathfrak{q}\sqrt{M_4/M_3})}{(g + \sqrt{M_2/M_1}) (g+\sqrt{M_1/M_2})} \, , \quad
    g\underline{g} = \frac{(f+\frac{\mathfrak{q}}{\sqrt{M_3M_4}})(f + \mathfrak{q}\sqrt{M_3M_4})}{(f+\frac{1}{\sqrt{M_1M_2}}) (f+\sqrt{M_1M_2})} \, .
\end{align}
This is the $ q $-Painlev\'e $ q\text{-}D_5^{(1)} $ system given in \eqref{eq:qD5}.

\subsection{\texorpdfstring{$ \mathrm{SU}(2) + 3\mathbf{F} $}{SU(2) + 3F}}

The bilinear relations of the $ q\text{-}A_4^{(1)} $ equation and the connection between its $ \tau $-functions and the 5d instanton partition function of the $ \mathrm{SU}(2)+3\mathbf{F} $ theory were studied in \cite{Matsuhira:2018qtx, Bonelli:2022dse}. This $ q $-Painlev\'e equation is also called the $ q $-Painlev\'e V equation, which corresponds to the $ q $-uplift of the differential Painlev\'e V equation. In this subsection, we bootstrap the bilinear relations of $ q\text{-}A_4^{(1)} $ from the gauge theory data of the 5d $ \mathrm{SU}(2)+3\mathbf{F} $ theory, which has $ E_4 = \mathrm{SU}(5) $ global symmetry. The unrefined effective prepotential and perturbative partition function are given by \eqref{eq:su2F-unref-pre} and \eqref{eq:su2F-unref-pert} with $ N_f=3 $. We consider the set $ \mathcal{S} $ of magnetic fluxes given by
\begin{align}
    \mathcal{S} = \left\{ (r,B) \, \big| \, (r,B) \text{ satisfying \eqref{eq:Nf-quantization}}, \ |b_i| \leq 2 , \ (0 \leq i \leq 3) \right\} \, ,
\end{align}
where $ r $ is the fractional part of the dynamical magnetic flux associated with $ \phi $, and $ B=(b_0,b_1,b_2,b_3) $ denotes the background magnetic fluxes associated with the mass parameters $ (m_0,\cdots,m_3) $.

For a fixed trial total flux $ \Phi $, we scan possible pairs in the candidate set of magnetic fluxes. Using the unrefined prepotential and perturbative partition function, we fix the coefficients $ \Lambda_i $ in the ansatz for the bilinear relations. Some of the ans\"atze are further ruled out by using the leading order information of the 1-instanton partition function. By considering two total fluxes $ \Phi = (0; 0, e_2-e_3) $ and $ \Phi = (0; 1, e_2-e_3) $, where $ e_i $ are unit vectors of $ \mathbb{R}^3 $, we find that the six tau-functions defined as
\begin{alignat}{3}
    \begin{aligned}
        &\tau_1 = \tauftn{0}{-1/4,\rho-e_1} \, , \
        &&\tau_2 = \tauftn{0}{1/4, \rho} \, , \
        &&\tau_3 = \tauftn{1/2}{-1/2,-e_3} \, , \\
        &\tau_4 = \tauftn{1/2}{1/2,e_2} \, , \
        &&\tau_5 = \tauftn{1/2}{0,0} \, , \
        &&\tau_6 = \tauftn{1/2}{0,e_2-e_3} \, ,
    \end{aligned}
\end{alignat}
with $ \rho = (e_1+e_2-e_3)/2 $, satisfy
\begin{alignat}{2}
    \begin{aligned}\label{eq:3F-bilinear}
        & \overline{\tau_1} \underline{\tau_2} = \tau_1 \tau_2 + \sqrt{\frac{\mathfrak{q}}{M_1}} \tau_5 \tau_6 \, , \quad
        && \tau_3 \tau_4 = \tau_5 \tau_6 + \sqrt{\mathfrak{q}M_2M_3} \tau_1 \tau_2 \, , \\
        & \overline{\tau_3} \underline{\tau_4} = \tau_5 \tau_6 + \sqrt{\frac{\mathfrak{q}}{M_2M_3}} \tau_1 \tau_2 \, , \quad
        && \tau_1 \overline{\tau_2} = \overline{\tau_1} \tau_2 + q^{1/4} \sqrt{\mathfrak{q}M_1} \overline{\tau_3} \tau_4 \, , \\
        & \tau_5 \overline{\tau_6} = \overline{\tau_3} \tau_4 + q^{3/4} \sqrt{\frac{\mathfrak{q}M_2}{M_3}} \overline{\tau_1} \tau_2 \, , \quad
        && \overline{\tau_5} \tau_6 = \overline{\tau_3} \tau_4 + q^{-1/4} \sqrt{\frac{\mathfrak{q}M_3}{M_2}} \overline{\tau_1} \tau_2 \, .
    \end{aligned}
\end{alignat}
Here, $ \overline{\tau_i} = \tau_i(m_0+\epsilon) $ and $ \underline{\tau_i}=\tau_i(m_0-\epsilon) $ denote the discrete time evolution. We have checked that these bilinear relations remain valid after including instanton corrections, using the instanton partition function of the $ \mathrm{SU}(2)+3\mathbf{F} $ theory up to 3-instanton order.

It is possible to find additional sets of bilinear relations by considering other choices of total flux $ \Phi $. Nevertheless, the six bilinear relations \eqref{eq:3F-bilinear} are sufficient to describe the $ q $-Painlev\'e system. Let us define $ (f,g) $ by
\begin{align}
    f = \sqrt{\mathfrak{q}} \frac{\tau_5 \tau_6}{\tau_1\tau_2} \, , \quad
    g = q^{1/4} \sqrt{\mathfrak{q}} \frac{\overline{\tau_3} \tau_4}{\overline{\tau_1}\tau_2} \, .
\end{align}
From the bilinear relations \eqref{eq:3F-bilinear}, we find that $ f $ and $ g $ satisfy
\begin{align}
    f\overline{f} = \frac{(g+q\mathfrak{q}\sqrt{M_2/M_3}) (g+\mathfrak{q}\sqrt{M_3/M_2})}{\sqrt{M_1} g+1} \, , \quad
    g\underline{g} = \sqrt{M_1} \frac{(f+\mathfrak{q}\sqrt{M_2M_3}) (f+\frac{\mathfrak{q}}{\sqrt{M_2M_3}})}{f+\sqrt{M_1}} .
\end{align}
This is the $ q $-Painlev\'e system $ q\text{-}A_4^{(1)} $ as given in \eqref{eq:qA4}.

\subsection{\texorpdfstring{$ \mathrm{SU}(2) + 2\mathbf{F} $}{SU(2) + 2F}}\label{sec:2F}

Finally, we consider the connection between the $ q\text{-}E_3^{(1)} $ Painlev\'e equation and the 5d $ \mathrm{SU}(2)+2\mathbf{F} $ theory, which has $ E_3 = \mathrm{SU}(3) \times \mathrm{SU}(2) $ global symmetry. This Painlev\'e equation is also called the $ q $-Painlev\'e $ \text{III}_1 $ equation, which corresponds to the $ q $-deformation of the differential Painlev\'e $ \text{III}_1 $ equation. The unrefined effective prepotential and perturbative partition function are given by \eqref{eq:su2F-unref-pre} and \eqref{eq:su2F-unref-pert} with $ N_f=2 $. Using the prepotential, the 1-loop perturbative partition function and the leading order of the partition function which encodes $ E_3 $ symmetry, we can bootstrap the tau-functions of the $ q\text{-}E_3^{(1)} $ equation. We consider the set $ \mathcal{S} $ of magnetic fluxes given by
\begin{align}
    \mathcal{S} = \left\{ (r,B) \, \big| \, (r,B) \text{ satisfying \eqref{eq:Nf-quantization}}, \ |b_i| \leq 2 , \ (0 \leq i \leq 2) \right\} \, ,
\end{align}
where $ r $ denotes the fractional part of the dynamical magnetic flux associated with $ \phi $, and $ B=(b_0,b_1,b_2) $ denotes the background magnetic fluxes associated with the mass parameters $ (m_0,m_1,m_2) $.

For the set of candidate magnetic fluxes, we construct the tau-functions and find their bilinear relations. By considering two total fluxes $ \Phi = (0; \pm 1/2, 0,1) $, we find that the four tau-functions
\begin{align}
    \tau_1 = \tauftn{0}{0,-1/2,1/2} \, , \quad
    \tau_2 = \tauftn{0}{-1/2,1/2,1/2} \, , \quad
    \tau_3 = \tauftn{1/2}{0,0,0} \, , \quad
    \tau_4 = \tauftn{1/2}{-1/2,0,1} \, ,
\end{align}
satisfy
\begin{alignat}{2}
    \begin{aligned}
        &\overline{\tau_1}\tau_2 = \tau_1 \overline{\tau_2} + q^{1/8} \sqrt{\frac{\mathfrak{q}}{M_1}} \tau_3 \overline{\tau_4} \, , \quad
        &&\overline{\tau_3} \tau_4 = \tau_3 \overline{\tau_4} + q^{-1/8} \sqrt{\frac{\mathfrak{q}}{M_2}} \tau_1 \overline{\tau_2} \, , \\
        &\underline{\tau_1} \overline{\tau_2} = \tau_1 \tau_2 + q^{-1/8} \sqrt{\mathfrak{q}M_1} \tau_3 \tau_4 \, , \quad
        &&\underline{\tau_3} \overline{\tau_4} = \tau_3 \tau_4 + q^{1/8} \sqrt{\mathfrak{q}M_2} \tau_1 \tau_2 \, ,
    \end{aligned}
\end{alignat}
where the discrete time evolution is denoted as $ \overline{\tau_i} = \tau_i(m_0+\epsilon) $ and $ \underline{\tau_i} = \tau_i(m_0-\epsilon) $. We have checked the validity of these bilinear relations using the instanton partition function of the $ \mathrm{SU}(2)+2\mathbf{F} $ theory, up to 3-instanton order. Furthermore, it is straightforward to check that $ (f,g) $ defined by
\begin{align}
    f = \sqrt{\mathfrak{q}} \frac{\tau_3 \tau_4}{\tau_1 \tau_2} \, , \quad
    g = q^{1/4} \sqrt{\mathfrak{q}} \frac{\tau_3 \overline{\tau_4}}{\tau_1 \overline{\tau_2}} \, ,
\end{align}
satisfy the non-linear $ q $-difference equations given by
\begin{align}
    f\overline{f} = q^{1/8} \sqrt{M_1} \frac{g(g+q^{1/8} \mathfrak{q} M_2^{-1/2})}{g+q^{1/8} \sqrt{M_1}} \, , \quad
    g\underline{g} = \frac{q^{1/8}}{\sqrt{M_1}} \frac{f(f+q^{1/8} \mathfrak{q}\sqrt{M_2})}{f+q^{1/8} M_1^{-1/2}} \, .
\end{align}
This system is the $ q\text{-}E_3^{(1)} $ equation given in \eqref{eq:qE3}.

\section{Conclusion} \label{sec:conclusion}

In this work, we studied bilinear relations of tau-functions for the elliptic and $ q $-Painlev\'e equations through their correspondence with 5d $ \mathcal{N}=1 $ supersymmetric gauge theories. This correspondence can be viewed as a five-dimensional analogue of the relation between differential Painlev\'e equations and 4d $ \mathcal{N}=2 $ theories. We defined the tau-functions as grand-canonical partition functions of 5d gauge theories on $ \hat{\mathbb{C}}^2 \times S^1 $ in the NS limit, and developed a systematic bootstrap method for deriving their bilinear relations. The bootstrap uses the effective prepotential, the perturbative partition function, and the global symmetry (or the leading-order behavior of the partition function) of the theory. In particular, the full non-perturbative instanton partition function is not needed as an input for constructing the bilinear relations. We applied this bootstrap method to the full sequence of 5d $ \mathrm{SU}(2)+N_f\mathbf{F} $ theories with $ N_f \leq 8 $. Starting from the E-string theory which UV-completes the 5d $ \mathrm{SU}(2)+8\mathbf{F} $ theory, we obtained the bilinear relations of the elliptic Painlev\'e equation. By considering the theories with fewer fundamental hypermultiplets, we derived the bilinear relations for the corresponding $ q $-Painlev\'e systems from the gauge theory data. We have further verified these bilinear relations using the explicit instanton partition functions.

We comment on the case of the 5d $ \mathcal{N}=1 $ rank-1 SCFT with trivial flavor symmetry, called the $ E_0 $ theory. This theory is a non-Lagrangian theory which does not admit a low-energy gauge theory description. It can be geometrically engineered by M-theory compactified on local $ \mathbb{P}^2 $, and can also be obtained by decoupling an instantonic hypermultiplet from the $ \mathrm{SU}(2)_\pi $ theory. The corresponding geometry appears in Sakai's classification of rational surfaces \cite{Sakai:2001}, as the final element in the degeneration sequence of the second row of Figure~\ref{fig:PainleveClassify}. However, the relevant Weyl group does not contain a non-trivial translation element generating a discrete time evolution, which indicates that there is no non-trivial $ q $-Painlev\'e dynamics associated with this case. This is consistent with the perspective of the 5d field theory. The partition function of the $ E_0 $ theory can be computed using the blowup equation \cite{Kim:2020hhh, Huang:2017mis}, so one may still define tau-functions by the discrete Fourier transformation \eqref{eq:tau}. Since this theory has no flavor symmetry, the tau-functions are labelled only by the fractional part $ r $ of the dynamical magnetic flux. The allowed values are $ r=\pm 1/6 $ and $ r=1/2 $, as discussed in \cite{Kim:2020hhh}. With this restricted set of tau-functions, one cannot construct a non-trivial three-term ansatz of the form \eqref{eq:ansatz} satisfying the total flux condition \eqref{eq:total-flux-cond-r}. Therefore, within the class of bilinear relations considered in this paper, the $ E_0 $ theory does not give a non-trivial bilinear relation.

There are several directions for future work. First, it would be interesting to extend the bootstrap method to the tau-functions associated with 4d $ \mathcal{N}=2 $ rank-1 theories. In particular, the correspondence between the Painlev\'e systems and 4d Minahan-Nemeschansky (MN) SCFTs with $ E_6 $, $ E_7 $, and $ E_8 $ flavor symmetries is still not as explicit as in the cases covered by the standard differential Painlev\'e-gauge correspondence given in Table~\ref{table:painleve-4d}. These are non-Lagrangian theories, and they do not admit a standard ADHM construction of the instanton moduli space. Nevertheless, the topological string approach has recently been used to compute the partition function of MN theories on the $ \Omega $-background \cite{Fucito:2023txg}. Combining these results with our bootstrap method may provide an explicit connection between the Painlev\'e equations and 4d MN theories.

Second, the present formulation should admit generalizations beyond the 5d $ \mathrm{SU}(2)+N_f\mathbf{F} $ theories. Bilinear relations associated with higher-rank pure gauge theories have been studied \cite{Bonelli:2022iob, Bershtein:2017swf,Bershtein:2018srt}. Since the bootstrap method uses only the prepotential, perturbative partition function, symmetry constraints, and the leading order information of the partition function, all of which are available for general 5d/6d SCFTs with UV completions \cite{Kim:2020hhh}, we expect that one can find non-trivial bilinear relations of tau-functions defined from the partition functions of general 5d and 6d theories. In particular, there are three more rank-1 UV-complete theories in five dimensions, namely $ \mathrm{SU}(2)_\theta + 1\mathbf{Adj} $ with $ \theta=0,\pi $, which are the circle compactification of the 6d $ \mathcal{N}=(2,0) $ $ A_1 $ theory and the twisted circle compactification of the 6d $ \mathcal{N}=(2,0) $ $ A_2 $ theory, respectively, and the local $ \mathbb{P}^2 + 1\mathbf{Adj} $ SCFT, obtained by a mass deformation of the $ \mathrm{SU}(2)_\pi + 1\mathbf{Adj} $ theory \cite{Bhardwaj:2019jtr, Kim:2020hhh}. In the 4d limit, the $ \mathrm{SU}(2)_0 + 1\mathbf{Adj} $ theory reduces to the 4d $ \mathcal{N}=2^* $ theory, whose connection to integrable systems has been studied in \cite{Bonelli:2019boe}. Finding the bilinear relations of their tau-functions and studying the associated integrable systems would be a promising direction.

Finally, the bilinear relations studied in this paper are formulated on the self-dual $ \Omega $-background, but there is a generalization of the Painlev\'e-gauge correspondence in the fully refined $ \Omega $-background $ \epsilon_1+\epsilon_2 \neq 0 $ \cite{Bonelli:2025rih, Bonelli:2025mfw}. For instance, the tau-function defined from the instanton partition function of the 4d $ \mathrm{SU}(2)+4\mathbf{F} $ theory can be identified with the Fourier series of the 4-point Virasoro conformal block at central charge $ c = 1+6\frac{(\epsilon_1+\epsilon_2)^2}{\epsilon_1\epsilon_2} $, from the perspective of the AGT correspondence. Generalizing the bilinear relations and the bootstrap method to the fully refined $ \Omega $-background would be an interesting open problem.

\acknowledgments

We would like to thank Gu Jie, Hee-Cheol Kim and Kimyeong Lee for valuable discussions and comments. The research of MK is supported by a KIAS Individual Grant QP097502. XW is supported by the Fundamental Research Funds for the Central Universities (Grants No. WK2030250140) and the National Natural Science Foundation of China (Grant No.12247103).

\appendix

\section{Proof of the lemma} \label{app:lemma}

In section~\ref{sec:elliptic}, we claim that two bilinear relations from the magnetic fluxes $ \{(\alpha_i, \alpha_i')\} $ and $ \{(\beta_i, \beta_i')\} $ are connected to each other by a Weyl transformation if $ \langle x_i, x_j \rangle = \langle y_i, y_j \rangle $ for $ x_i=\alpha_i-\alpha_i' $ and $ y_i=\beta_i-\beta_i' $. This claim is equivalent to the following mathematical lemma.
\begin{lemma}
    Let $ \{\alpha_{ia}\}_{i=1,2,3}^{a=1,2} $ and $ \{\beta_{ia}\}_{i=1,2,3}^{a=1,2} $ be two sets of nonzero $ E_8 $ root vectors satisfying the constant-sum conditions
    \begin{align}
        \alpha_{i1}+\alpha_{i2}=\lambda \, , \quad
        \beta_{i1}+\beta_{i2}=\sigma \, ,
    \end{align}
    for all $ i $, and set $ x_i = \alpha_{i1}-\alpha_{i2} $ and $ y_i=\beta_{i1}-\beta_{i2} $. If the Gram matrices of $ x_i $ and $ y_i $ agree, $ \langle x_i, x_j \rangle = \langle y_i, y_j \rangle $ for all $ i,j $, then there exists $ w \in W(E_8) $ such that $ w(\alpha_{ia}) = \beta_{ia} $.
\end{lemma}
\begin{proof}
    Since $ \alpha_{ia} $ and $ \beta_{ia} $ are nonzero $ E_8 $ roots, $ \lVert \alpha_{ia} \rVert^2 = \lVert \beta_{ia} \rVert^2 = 2 $, and  $ \langle \lambda, x_i \rangle = \langle \sigma, y_i \rangle = 0 $. The roots are decomposed as 
    \begin{align}\label{eq:E8root-decompose}
        \alpha_{i1} = \frac{\lambda+x_i}{2} \, , \quad
        \alpha_{i2} = \frac{\lambda-x_i}{2} \, , \quad
        \beta_{i1} = \frac{\sigma+y_i}{2} \, , \quad
        \beta_{i2} = \frac{\sigma-y_i}{2} \, .
    \end{align}
    Squaring \eqref{eq:E8root-decompose} gives $ \lVert \lambda \rVert^2 + \lVert x_i \rVert^2 = \lVert\sigma\rVert^2 + \lVert y_i \rVert^2 = 8 $. Combined with $ \lVert x_i \rVert^2 = \lVert y_i \rVert^2 $, this yields $ \lVert \lambda \rVert^2 = \lVert \sigma \rVert^2 $. Then a direct computation gives 
    \begin{align}\label{eq:app-inner}
       \langle \alpha_{ia}, \alpha_{jb} \rangle = \frac{\lVert \lambda \rVert^2 \pm \langle x_i, x_j \rangle}{4} = \frac{\lVert \sigma \rVert^2 \pm \langle y_i, y_j \rangle}{4} = \langle \beta_{ia}, \beta_{jb} \rangle \, ,
    \end{align}
    where the sign is plus if $ a=b $ and minus if $ a\neq b $. 

    Let $ \Sigma $ be the $ E_8 $ root system and $ \Xi_A, \Xi_B \subset \Sigma $ be the smallest root subsystems generated by $ \{\alpha_{ia}\} $ and $ \{\beta_{ia}\} $, respectively. Since $ \alpha_{ia} \in \operatorname{Span}\{\lambda, x_1, x_2, x_3 \} $ and $ \beta_{ia} \in \operatorname{Span}\{\sigma, y_1,y_2,y_3\} $, both root subsystems have rank at most four. By \eqref{eq:app-inner}, the assignment $ f(\alpha_{ia}) = \beta_{ia} $ extends to an isomorphism $ f : \Xi_A \to \Xi_B $ of root systems preserving the inner product. By the classification of root subsystems of $ E_8 $ \cite{Oshima:2007}, possible subsystems of rank at most four are
    \begin{align}
        A_1,\ A_2,\ A_3,\ A_4,\ D_4,\ 2A_1,\ 3A_1,\ 4A_1,\ A_2+A_1,\ A_2+2A_1,\ 2A_2,\ A_3+A_1 .
    \end{align}
    Let $ \iota_A : \Xi_A \hookrightarrow \Sigma $ and $ \iota_B : \Xi_B \hookrightarrow \Sigma $ be the inclusion maps. For every type on the list except $ 4A_1 $, all embeddings of that type into $ \Sigma $ lie in a single $ W(E_8) $-orbit \cite{Oshima:2007}. Consequently, provided $ \Xi_A \not\cong 4A_1 $, the two embeddings $ \iota_A $ and $ \iota_B \circ f $ of $ \Xi_A $ into $ \Sigma $ are $ W(E_8) $-related, i.e., there exists $ w \in W(E_8) $ with $ \iota_B \circ f = w \circ \iota_A $. Evaluating on the generators yields
    \begin{align}
        w(\alpha_{ia}) = \beta_{ia}
    \end{align}
    for all $ i $ and $ a $.

    It remains to exclude the case $ \Xi_A \cong 4A_1 $. Assuming this for
    contradiction, write
    \begin{align}
        \Xi_A = \{ \pm \gamma_1, \pm \gamma_2, \pm \gamma_3, \pm \gamma_4 \} \, ,
        \qquad \langle \gamma_j, \gamma_k \rangle = 2\delta_{jk} \, .
    \end{align}
    Each $ \alpha_{ia} $ is then of the form $ \pm \gamma_k $ for some $ k $, and the constant-sum condition $ \alpha_{i1}+\alpha_{i2}=\lambda $ leaves three possibilities for $ \lambda $: $ \lambda = 0 $, $ \pm 2\gamma_k $, or $ \pm\gamma_j \pm \gamma_k $ with $ j\neq k $. If $ \lambda = 0 $, then $ \alpha_{i2} = -\alpha_{i1} $ for each $ i $, so the six roots $ \{\alpha_{ia}\} $ span at most three dimensions, contradicting $ \rank\Xi_A = 4 $. If $ \lambda \neq 0 $, the pair $ \{\alpha_{i1}, \alpha_{i2}\} $ is uniquely determined by $ \lambda $, so the six vectors span at most two dimensions, again contradicting $ \rank\Xi_A = 4 $. Either way, $ \Xi_A \not\cong 4A_1 $, completing the proof.
\end{proof}

\section{Discrete Painlev\'e equations and bilinear relations} \label{app:bilinear}

In this appendix, we present some supplementary mathematical material on the elliptic and $ q $-Painlev\'e equations and their bilinear relations. We first list their explicit expressions, and review their bilinear relations in the subsequent subsections.

\subsection{Explicit expressions of discrete Painlev\'e equations} \label{app:PainleveExpression}

In this appendix, we present explicit expressions for the elliptic and $ q $-Painlev\'e equations. More detailed derivations and mathematical discussions can be found in \cite{Sakai:2007, Yamada:2014, Kajiwara:2015}. Let
\begin{alignat}{2}
    \begin{aligned}\label{eq:theta}
        &\theta_1(\tau,z) = -i \sum_{n\in \mathbb{Z}} (-1)^n q^{\frac{1}{2}(n+1/2)^2} y^{n+1/2} \, , \quad
        &&\theta_2(\tau,z) = \sum_{n\in \mathbb{Z}} q^{\frac{1}{2}(n+1/2)^2} y^{n+1/2} \\
        &\theta_3(\tau,z) = \sum_{n\in \mathbb{Z}} q^{\frac{1}{2}n^2} y^{n} \, , \quad
        &&\theta_4(\tau,z) = \sum_{n\in \mathbb{Z}} (-1)^n q^{\frac{1}{2}n^2} y^{n} \, ,
    \end{aligned}
\end{alignat}
be the Jacobi theta functions, where $ q=e^{-\tau} $ and $ y=e^{-z} $. We also use the multiplicative notation $ \tilde{\theta}_1(q,y) \equiv \theta_1(\tau,z) $ for convenience.

\paragraph{$ \boldsymbol{e\text{-}E_8^{(1)}} $}
Let
\begin{align}\label{eq:discrete-time}
   T(\kappa_1,\kappa_2,f,g) = (\kappa_1/q,q \kappa_2, \overline{f}, \overline{g}) \, , \quad
   q = \frac{\kappa_1^2 \kappa_2^2}{\prod_{i=1}^8 b_i} \, ,
\end{align}
be the discrete time evolution. The elliptic Painlev\'e equation is
\begin{align}
    \begin{aligned}\label{eq:eE8}
        \frac{\bigl(f-f(\frac{\kappa_2}{t}) \bigr) \bigl(\overline{f}-\overline{f}(\frac{\kappa_2}{t}) \bigr)}{\bigl(f-f(t) \bigr) \bigl(\overline{f}-\overline{f}(t) \bigr)} &= \frac{f_\alpha(t) \overline{f_\alpha}(t)}{f_\alpha(\frac{\kappa_2}{t}) \overline{f_\alpha}(\frac{\kappa_2}{t})} \frac{U(\frac{\kappa_2}{t})}{U(t)} \, , \\
        \frac{\bigl(g-g(\frac{\kappa_1}{s})\bigr) \bigl(\underline{g}-\underline{g}(\frac{\kappa_1}{s})\bigr)}{\bigl(g-g(s)\bigr)\bigl(\underline{g}-\underline{g}(s)\bigr)} &= \frac{g_\alpha(s) \underline{g_\alpha}(s)}{g_\alpha(\frac{\kappa_1}{s}) \underline{g_\alpha}(\frac{\kappa_1}{s})} \frac{U(\frac{\kappa_1}{s})}{U(s)} \, ,
    \end{aligned}
\end{align}
where $ s $ and $ t $ are variables such that $ f=f(s) $ and $ g=g(t) $, $ \alpha $ and $ \beta $ are arbitrary parameters, and
\begin{gather}
    \begin{gathered}
        f(z) = \frac{f_\beta(z)}{f_\alpha(z)} \, , \quad
        g(z) = \frac{g_\beta(z)}{g_\alpha(z)} \, , \quad
        U(z) = \prod_{i=1}^8 \tilde{\theta}_1(b_0, b_i/z) \, , \\
        f_\alpha(z) = \tilde{\theta}_1\Bigl(b_0, \frac{\alpha}{z}\Bigr) \tilde{\theta}_1\Bigl(b_0, \frac{\kappa_1}{\alpha z}\Bigr) \, , \quad
        g_\alpha(z) = \tilde{\theta}_1\Bigl(b_0, \frac{\alpha}{z}\Bigr) \tilde{\theta}_1\Bigl(b_0, \frac{\kappa_2}{\alpha z}\Bigr) \, .
    \end{gathered}
\end{gather}

\paragraph{$ \boldsymbol{q\text{-}E_8^{(1)}} $}
Let \eqref{eq:discrete-time} be the discrete time evolution. The $ q $-Painlev\'e equation $ q\text{-}E_8^{(1)} $ can be written as
\begin{align}
    \begin{aligned}\label{eq:qE8}
        \frac{\bigl(\overline{f}-g \bigr) \bigl(f-g\bigr) - \bigl(\frac{\kappa_1}{q}-\kappa_2\bigr)\bigl(\kappa_1-\kappa_2\bigr) \frac{1}{\kappa_2}}{\bigl( \frac{\overline{f}q}{\kappa_1} - \frac{g}{\kappa_2} \bigr) \bigl( \frac{f}{\kappa_1} - \frac{g}{\kappa_2} \bigr) - \bigl( \frac{q}{\kappa_1} - \frac{1}{\kappa_2} \bigr) \bigl( \frac{1}{\kappa_1} -\frac{1}{\kappa_2} \bigr) \kappa_2 } &= \frac{\kappa_1^2}{q} \frac{A(\kappa_2,g)}{B(\kappa_2,g)} \, , \\
        \frac{\bigl(f-g \bigr) \bigl(f-\underline{g}\bigr) - \bigl(\kappa_1-\kappa_2\bigr) \bigl(\kappa_1-\frac{\kappa_2}{q}\bigr) \frac{1}{\kappa_1} }{\bigl( \frac{f}{\kappa_1} - \frac{g}{\kappa_2} \bigr) \bigl( \frac{f}{\kappa_1} - \frac{\underline{g} q}{\kappa_2} \bigr) - \bigl( \frac{1}{\kappa_1}- \frac{1}{\kappa_2}\bigr) \bigl(\frac{1}{\kappa_1} - \frac{q}{\kappa_2}\bigr) \kappa_1} &= \frac{\kappa_2^2}{q} \frac{A(\kappa_1, f)}{B(\kappa_1, f)} \, .
    \end{aligned}
\end{align}
Here, the polynomials $ A $ and $ B $ are
\begin{align}
    \begin{aligned}
        A(h,z) &= c_0 z^4 - c_1 z^3 + \left( -3hc_0 + c_2 - \frac{c_8}{h^3} \right) z^2 \\
        &\quad + \left(2h c_1 - c_3 + \frac{c_7}{h^2} \right) z + \left( h^2 c_0 - h c_2 + c_4 - \frac{c_6}{h} + \frac{c_8}{h^2} \right) \, , \\
        B(h,z) &= \frac{c_8}{h^4} z^4 - \frac{c_7}{h^3} z^3 + \left( -h c_0 + \frac{c_6}{h^2} - \frac{3c_8}{h^3} \right) z^2 \\
        &\quad + \left( hc_1 - \frac{c_5}{h} + \frac{2c_7}{h^2} \right) z + \left( h^2 c_0 - h c_2 + c_4 - \frac{c_6}{h} + \frac{c_8}{h^2} \right) \, ,
    \end{aligned}
\end{align}
where $ c_i $ are defined as
\begin{align}
    \frac{1}{z^4} \prod_{i=1}^8 (z-b_i) = \frac{1}{z^4} \sum_{i=0}^8 (-1)^i c_{8-i} z^i \, .
\end{align}

\paragraph{$ \boldsymbol{q\text{-}E_7^{(1)}} $}
Let
\begin{align}\label{eq:discrete-time2}
    T(t, f, g) = (qt, \overline{f}, \overline{g}) \, , \quad
    q \prod_{i=1}^8 b_i = 1 \, ,
\end{align}
be the discrete time evolution. The $ q $-Painlev\'e $ q\text{-}E_7^{(1)} $ equation is
\begin{align}\label{eq:qE7}
    \frac{(fg-t^2) (\overline{f}g-t^2/q)}{(fg-1)(\overline{f}g-1)} = \frac{\prod_{i=5}^8 (g-b_it)}{\prod_{i=1}^4 (g-1/b_i)} \, , \quad
    \frac{(fg-t^2) (f\underline{g}-qt^2)}{(fg-1)(f\underline{g}-1)} = \frac{\prod_{i=5}^8 (f-t/b_i)}{\prod_{i=1}^4 (f-b_i)} \, .
\end{align}

\paragraph{$ \boldsymbol{q\text{-}E_6^{(1)}} $}
Let \eqref{eq:discrete-time2} be the discrete time evolution. The $ q\text{-}E_6^{(1)} $ equation is
\begin{align}\label{eq:qE6}
    \frac{(fg-1)(\overline{f}g-1)}{f\overline{f}} = \frac{\prod_{i=1}^4 (g-1/b_i)}{\prod_{i=5}^6 (g-t b_i)} \, , \quad
    \frac{(fg-1)(f\underline{g}-1)}{g\underline{g}} = \frac{\prod_{i=1}^4 (f-b_i)}{\prod_{i=7}^8 (f-t/b_i)} \, .
\end{align}

\paragraph{$ \boldsymbol{q\text{-}D_5^{(1)}} $}
For the time evolution \eqref{eq:discrete-time2}, the $ q\text{-}D_5^{(1)} $ equation is given by
\begin{align}\label{eq:qD5}
    f\overline{f} = b_3b_4 \frac{(g-tb_5)(g-tb_6)}{(g-1/b_1)(g-1/b_2)} \, , \quad
    g\underline{g} = \frac{1}{b_1b_2} \frac{(f-t/b_7)(f-t/b_8)}{(f-b_3)(f-b_4)} \, .
\end{align}

\paragraph{$ \boldsymbol{q\text{-}A_4^{(1)}} $}
For the time evolution \eqref{eq:discrete-time2}, the $ q\text{-}A_4^{(1)} $ equation is given by
\begin{align}\label{eq:qA4}
    f\overline{f} = -b_2b_3b_4 \frac{(g-t b_5)(g-t b_6)}{g-1/b_1} \, , \quad
    g\underline{g} = -\frac{1}{b_1b_2b_3} \frac{(f-t/b_7)(f-t/b_8)}{f-b_4} \, .
\end{align}

\paragraph{$ \boldsymbol{q\text{-}E_3^{(1)}} $}
For the time evolution \eqref{eq:discrete-time2}, the $ q\text{-}E_3^{(1)} $ equation is given by
\begin{align}\label{eq:qE3}
    f\overline{f} = -b_2b_3b_4 \frac{g(g-tb_5)}{g-1/b_1} \, , \quad
    g\underline{g} = -\frac{1}{b_1b_2b_3} \frac{f(f-t/b_8)}{f-b_4} \, .
\end{align}

\paragraph{$ \boldsymbol{q\text{-}E_2^{(1)}} $}
For the time evolution \eqref{eq:discrete-time2}, the $ q\text{-}E_2^{(1)} $ equation is given by
\begin{align}\label{eq:qE2}
    f\overline{f} = -b_2b_3b_4 \frac{g^2}{g-1/b_1} \, , \quad
    g\underline{g} = -\frac{1}{b_1b_2b_3} \frac{f(f-t/b_8)}{f-b_4} \, .
\end{align}

\paragraph{$ \boldsymbol{q\text{-}A_1^{(1)}} $}
For the time evolution \eqref{eq:discrete-time2}, the $ q\text{-}A_1^{(1)} $ equation is given by
\begin{align}\label{eq:qA1}
    f\overline{f} = b_1b_2b_3b_4 g^2 \, , \quad
    g\underline{g} = -\frac{1}{b_1b_2b_3} \frac{f(f-t/b_8)}{f-b_4} \, .
\end{align}

\paragraph{$ \boldsymbol{q\text{-}\tilde{A}_1^{(1)}} $}
For the time evolution \eqref{eq:discrete-time2}, the $ q\text{-}\tilde{A}_1^{(1)} $ equation is given by
\begin{align}\label{eq:qA1t}
    f\overline{f} = -b_2b_3b_4 g \, , \quad
    g\underline{g} = \frac{1}{b_1b_2b_3b_8} \frac{f t}{f-b_4} \, .
\end{align}

\subsection{Affine \texorpdfstring{$ E_8^{(1)} $}{E8(1)} algebra}

We next give a brief review of the $ E_8^{(1)} $ algebra, which is the affine extension of the $ E_8 $ Lie algebra. Let $ \alpha_i $ with $ i=0,1,\cdots,8 $ be the simple roots of $ E_8^{(1)} $. The Cartan matrix of $ E_8^{(1)} $ can be read from the Dynkin diagram given in Figure~\ref{fig:E8dynkin}. The root system of $ E_8^{(1)} $ is given by $ \{ \alpha + n \delta \mid \alpha\in \Delta(E_8), \, n \in \mathbb{Z} \} $, where $ \Delta(E_8) $ is the root system of the simple Lie algebra $ E_8 $, and $ \delta $ is the \emph{imaginary root} defined by
\begin{align}
    \delta = \alpha_0 + 2\alpha_1 + 4\alpha_2 + 6\alpha_3 + 5\alpha_4 + 4\alpha_5 + 3\alpha_6 + 2\alpha_7 + 3\alpha_8 \, ,
\end{align}
where the coefficients of the simple roots $ \alpha_i $ are the Coxeter labels (or marks). We parametrize the simple roots of $ E_8^{(1)} $ using the mass parameters $ \mu_l $ of the E-string theory which correspond to the orthonormal basis of the root system as
\begin{gather}
    \begin{gathered}\label{eq:E81-ortho}
        \alpha_1 = \mu_1 + \mu_8 - \frac{1}{2} \sum_{l=1}^8 \mu_l \, , \quad
        \alpha_i = -\mu_{i-1} + \mu_{i} \, , \quad (2 \leq i \leq 7) \, , \\
        \alpha_8 = \mu_1 + \mu_2 \, , \quad
        \alpha_0 = \delta - \mu_7 - \mu_8 \, .
    \end{gathered}
\end{gather}

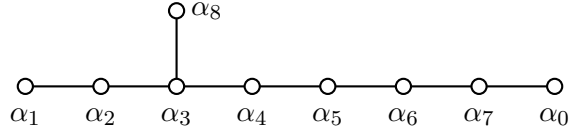
\begin{figure}
    \centering
    \begin{tikzpicture}
        \draw [thick] (0,0) -- (7,0);
        \draw [thick] (2,0) -- (2,1);
        \filldraw [thick, fill=white] (0,0) circle (0.1);
        \filldraw [thick, fill=white] (1,0) circle (0.1);
        \filldraw [thick, fill=white] (2,0) circle (0.1);
        \filldraw [thick, fill=white] (3,0) circle (0.1);
        \filldraw [thick, fill=white] (4,0) circle (0.1);
        \filldraw [thick, fill=white] (5,0) circle (0.1);
        \filldraw [thick, fill=white] (6,0) circle (0.1);
        \filldraw [thick, fill=white] (7,0) circle (0.1);
        \filldraw [thick, fill=white] (2,1) circle (0.1);
        \draw (0,-0.4) node {$ \alpha_1 $};
        \draw (1,-0.4) node {$ \alpha_2 $};
        \draw (2,-0.4) node {$ \alpha_3 $};
        \draw (3,-0.4) node {$ \alpha_4 $};
        \draw (4,-0.4) node {$ \alpha_5 $};
        \draw (5,-0.4) node {$ \alpha_6 $};
        \draw (6,-0.4) node {$ \alpha_7 $};
        \draw (7,-0.4) node {$ \alpha_0 $};
        \draw (2.4,1) node {$ \alpha_8 $};
    \end{tikzpicture}
    \caption{Dynkin diagram of affine $ E_8^{(1)} $. In much of the Painlev\'e equation literature, the notation where the labels of $ \alpha_0 $ and $ \alpha_8 $ are exchanged is used.} \label{fig:E8dynkin}
\end{figure}

The affine $ E_8^{(1)} $ algebra is closely related to the geometry of the eight-point blowup of $ \mathbb{P}^1 \times \mathbb{P}^1 $, or equivalently to the nine-point blowup of $ \mathbb{P}^2 $. This surface appears in Sakai's classification of Painlev\'e equations \cite{Sakai:2001}, as well as in the geometric engineering of the E-string theory. The Picard lattice of this surface is generated by the two curve classes $ H_1 $ and $ H_2 $ of $ \mathbb{P}^1 \times \mathbb{P}^1 $ and by the eight exceptional divisors $ E_{i=1,\cdots,8} $ introduced by the blowups. The intersection pairing of the curve classes is given by\footnote{\label{note:dP9}Equivalently, one can use the Picard lattice of the nine-point blowup of $ \mathbb{P}^2 $ generated by $ \varepsilon_0, \varepsilon_1, \cdots, \varepsilon_9 $. They are related to $ H_i $ and $ E_i $ by $ \varepsilon_0 = H_1+H_2-E_1 $, $ \varepsilon_1 = H_1-E_1 $, $ \varepsilon_2 = H_2-E_1 $ and $ \varepsilon_i=E_{i-1} $ for $ i\geq 3 $.}
\begin{align}\label{eq:geo-intersection}
    H_1^2 = H_2^2 = 0 \, , \quad
    H_1 \cdot H_2 = 1 \, , \quad
    E_i \cdot E_j = -\delta_{ij} \, , \quad
    H_i \cdot E_j = 0\,.
\end{align}
The simple roots $ \alpha_i $ of the $ E_8^{(1)} $ algebra can be parametrized using the curve classes as
\begin{gather}\label{eq:E8-geometry}
    \alpha_1 = H_1-H_2 \, , \quad
    \alpha_2 = H_2 - E_1 - E_2 \, , \quad
    \alpha_3 = E_2 - E_3 \, , \quad
    \alpha_4 = E_3 - E_4 \, , \\
    \alpha_5 = E_4 - E_5 \, , \quad
    \alpha_6 = E_5 - E_6 \, , \quad
    \alpha_7 = E_6 - E_7 \, , \quad
    \alpha_8 = E_1 - E_2 \, , \quad
    \alpha_0 = E_7 - E_8 \, ,  \nonumber
\end{gather}
where the inner product between two roots is given by minus of the geometric intersection pairing \eqref{eq:geo-intersection}, $ \langle \alpha_i, \alpha_j \rangle = -\alpha_i \cdot \alpha_j $. In this parametrization, the imaginary root $ \delta $ corresponds to the anti-canonical class of the surface:
\begin{align}
    \delta = 2H_1 + 2H_2 - \sum_{i=1}^8 E_i \, .
\end{align}

The affine Weyl group $ W = W(E_8^{(1)}) $ is generated by $ \langle s_0, s_1, \cdots, s_8 \rangle $, whose action on the parameters $ H_i $ and $ E_i $ is given by
\begin{gather}\label{eq:E8Weyl}
    s_0 \ : \ E_7 \leftrightarrow E_8 \, , \quad
    s_1 \ : \ H_1 \leftrightarrow H_2 \, , \quad
    s_j \ : \ E_{j-1} \leftrightarrow E_j \, , \ (3 \leq j \leq 7) \, , \quad
    s_8 \ : \ E_1 \leftrightarrow E_2 \, , \nonumber \\
    s_2 \ : \ E_1 \mapsto H_2-E_2 \, , \quad E_2 \mapsto H_2-E_1 \, , \quad H_1 \mapsto H_1+H_2-E_1-E_2 \, .
\end{gather}
Unlike the Weyl group of a simple Lie algebra, the Weyl group of an affine Lie algebra includes translations along the weight lattice and is therefore of infinite order. For a root $ \alpha $ of the affine Lie algebra, a translation element in the affine Weyl group, called the \emph{Kac translation} is defined as
\begin{align}\label{eq:Kac-traslation}
    T_\alpha(\lambda) = \lambda + \langle \delta, \lambda \rangle \alpha - \left( \frac{1}{2} \langle \alpha, \alpha \rangle \langle \delta, \lambda \rangle + \langle \alpha, \lambda \rangle \right) \delta \, .
\end{align}
This is a linear transformation satisfying
\begin{align}
    T_\alpha T_\beta = T_{\alpha+\beta} \, , \quad
    w\circ T_\alpha = T_{w(\alpha)} \circ w \, , \quad (w \in W) \, , \quad
    \langle T_\alpha(\lambda), T_\alpha(\mu) \rangle = \langle \lambda, \mu \rangle \, .
\end{align}
The Kac translation $ T_{\alpha_1} $ associated with the root $ \alpha_1 $ acts on the parameters $ H_i $ and $ E_i $ as
\begin{gather}
    \begin{gathered}\label{eq:HE-Kac}
        T_{\alpha_1}(H_1) = H_1 - 2(H_1-H_2) + \delta \, , \quad
        T_{\alpha_1}^{-1}(H_1) = H_1 + 2(H_1-H_2) + 3\delta \, , \\
        T_{\alpha_1}(H_2) = H_2 - 2(H_1-H_2) + 3\delta \, , \quad
        T_{\alpha_1}^{-1}(H_2) = H_2 + 2(H_1-H_2) + \delta \, , \\
        T_{\alpha_1}^{\pm 1}(E_i) = E_i \mp (H_1-H_2) + \delta \, .
    \end{gathered}
\end{gather}
This Kac translation will play the role of the time evolution of the discrete Painlev\'e equations, while the imaginary root $ \delta $ will be identified with the self-dual $ \Omega $-deformation parameter $ \epsilon $.

\subsection{Elliptic Painlev\'e equation} \label{app:elliptic}

We now relate the bilinear relations obtained in section~\ref{sec:elliptic} from the 5d $ \mathrm{SU}(2)+8\mathbf{F} $ theory to the elliptic Painlev\'e equation. We relabel the tau-functions $ \tau_{\mathrm{6d}}\genfrac[]{0pt}{1}{1/2}{0,\alpha} $ defined in section~\ref{sec:elliptic} as
\begin{alignat}{3}
    \begin{aligned}
        &\mathcal{T}(E_1) = \tauftnd{1/2}{0, -e_1+e_8} ,
        &&    \mathcal{T}(H_1-E_1) = \tauftnd{1/2}{0, \rho+e_1} ,
        &&    \mathcal{T}(H_2-E_1) = \tauftnd{1/2}{0, e_1+e_8} , \\
        &    \mathcal{T}(E_i) = \tauftnd{1/2}{0, e_i+e_8} ,
        &&    \mathcal{T}(H_1-E_i) = \tauftnd{1/2}{0, \rho-e_i} ,
        &&    \mathcal{T}(H_2-E_i) = \tauftnd{1/2}{0, -e_i+e_8}  , \\
        &    \mathcal{T}(E_8) = \tauftnd{1/2}{0, 0} ,
        &&    \mathcal{T}(H_1-E_8) = \tauftnd{1/2}{0, \rho+e_8} ,
        &&    \mathcal{T}(H_2-E_8) = \tauftnd{1/2}{0, 2e_8} , 
    \end{aligned}
\end{alignat}
where $ \rho = -e_1 + \sum_{i=1}^8 e_i/2 $, with $ e_i $ denoting the orthonormal basis of $ \mathbb{R}^8 $. This relabelling is compatible with the Weyl transformations of $ E_8^{(1)} $ given in \eqref{eq:E8Weyl}. The bilinear relations of the elliptic Painlev\'e equation can be written as
\begin{align}\label{eq:elliptic-bilinear-ref}
    \Phi_{H_1}(E_i, E_j) \xi_k + (i,j,k)\text{-cyclic} = 0 \, , \quad
    \Phi_{H_2}(E_i, E_j) \eta_k + (i,j,k)\text{-cyclic} = 0 \, ,
\end{align}
where
\begin{gather}
    \begin{gathered}\label{eq:elliptic-Phi}
        \Phi_{\kappa}(x,y) = \theta_1(\mu_0, x-y) \theta_1(\mu_0, \kappa-x-y) \, , \\
        \xi_k = \mathcal{T}(E_k) \mathcal{T}(H_1-E_k) \, , \quad
        \eta_k = \mathcal{T}(E_k) \mathcal{T}(H_2-E_k) \, .
    \end{gathered}
\end{gather}
By comparing the bilinear relations studied in section~\ref{sec:elliptic} and the parametrizations \eqref{eq:E81-ortho} and \eqref{eq:E8-geometry}, we identify the imaginary root $ \delta $ as the $ \Omega $-deformation parameter, $ \delta=-\epsilon $.

The bilinear relations \eqref{eq:elliptic-bilinear-ref} and their equivalence with the elliptic Painlev\'e equation are established in section~5 of \cite{Kajiwara:2015}; here we summarize the construction. The discrete time evolution of the elliptic Painlev\'e equation is given by a Kac translation \eqref{eq:Kac-traslation} with respect to a simple root. We choose $ T_{\alpha_1} $ as the time evolution, and denote $ \overline{f} \equiv T_{\alpha_1}(f) $ and $ \underline{f} \equiv T_{\alpha_1}^{-1}(f) $. Let us introduce a parameter $ \lambda = (H_1+H_2)/4 + n\delta $ for an arbitrary constant $ n $, which satisfies\footnote{This corrects a typo in (5.115) of \cite{Kajiwara:2015}.}
\begin{align}
    s_2(\lambda) = \lambda + \frac{H_2 - E_1-E_2}{4} \, , \quad
    s_j(\lambda) = \lambda \, , \quad (j \neq 2) \, ,
\end{align}
for the Weyl reflections \eqref{eq:E8Weyl}. A convenient parametrization of $ H_i $ and $ E_i $ is
\begin{align}
    \kappa_i = H_i - 2\lambda \, , \quad
    b_i = E_i - \lambda \, .
\end{align}
Instead of the additive notation, let us use the \emph{multiplicative notation}, in which all parameters are exponentiated, so that
\begin{align}\label{eq:kap-v-parameter}
    \kappa_i = \frac{H_i}{\lambda^2} \, , \quad
    b_i = \frac{E_i}{\lambda} \, ,
\end{align}
where $ H_i $ and $ E_i $ are now defined as the exponentials of the original additive parameters, $ H_i^{\mathrm{new}} = e^{-H_i^{\mathrm{old}}} $, etc. From \eqref{eq:HE-Kac}, the time evolutions of the new parameters are given by
\begin{align}\label{eq:Kac-kappa}
    \overline{\kappa_1} = \frac{\kappa_1}{q} \, , \quad
    \overline{\kappa_2} = \kappa_2 q \, , \quad
    \overline{b_i} = b_i \, , \quad
    \overline{\lambda} = \frac{\lambda \kappa_2 q}{\kappa_1} \, ,
\end{align}
where $ q=e^{-\delta} $. We next define $ f $ and $ g $ as
\begin{align}
    f(z) = \frac{\Phi_{\kappa_1}(\beta,z)}{\Phi_{\kappa_1}(\alpha,z)} \, , \quad
    g(z) = \frac{\Phi_{\kappa_2}(\beta,z)}{\Phi_{\kappa_2}(\alpha,z)} \, ,
\end{align}
where $ \alpha $ and $ \beta $ are arbitrary fixed parameters. Note that $ f(z) $ and $ g(z) $ are elliptic functions of order 2, satisfying $ f(\kappa_1/z)=f(z) $ and $ g(\kappa_2/z)=g(z) $. We now define variables $ (Z,W) $, which play the role of the dependent variables of the elliptic Painlev\'e equation, as
\begin{align}\label{eq:ZW}
    \frac{\Phi_{\kappa_1}(\alpha,b_i)}{\Phi_{\kappa_1}(\alpha,b_j)} \frac{Z-f(b_i)}{Z-f(b_j)} = \frac{\xi_i}{\xi_j} \, , \quad
    \frac{\Phi_{\kappa_2}(\alpha,b_i)}{\Phi_{\kappa_2}(\alpha,b_j)} \frac{W-g(b_i)}{W-g(b_j)} = \frac{\eta_i}{\eta_j} \, ,
\end{align}
using the products of tau-functions $ \xi_i $ and $ \eta_i $. Here, the elliptic function $ \Phi $ defined in \eqref{eq:elliptic-Phi} should be understood as the multiplicative notation, as $ \Phi_{\kappa}(x,y) = \tilde{\theta}(\mu_0,\frac{x}{y}) \tilde{\theta}(\mu_0,\frac{\kappa}{xy}) $. Let $ s $ and $ t $ be variables such that $ Z = f(s) $ and $ W=g(t) $. Then the time evolutions of $ Z $ and $ W $ satisfy the elliptic Painlev\'e equation \eqref{eq:eE8}. We refer to \cite{Kajiwara:2015} for the detailed derivation.

\subsection{\texorpdfstring{$ q\text{-}E_8^{(1)} $}{q-E8(1)} from bilinear relations} \label{app:qE8}

The $ q $-Painlev\'e equation $ q\text{-}E_8^{(1)} $ can be obtained by replacing the elliptic functions appearing in the elliptic Painlev\'e equation with their trigonometric counterparts. The relation between the mass parameters $ \mu_l $ and the parameters $ \{H_i, E_i\} $ can be found by comparing \eqref{eq:E81-ortho} and \eqref{eq:E8-geometry}, where $ \mu_8 $ and the imaginary root $ \delta $ are identified as $ \mu_8 = -2m_0 $ and $ \delta=-\epsilon $. We consider tau-functions $ \mathcal{T} $ defined in terms of the tau-functions used in section~\ref{sec:7F} as
\begin{alignat}{3}
    \begin{aligned}
        &\mathcal{T}(E_1) = \tauftn{1/2}{-1/2, -e_1} \, , \!
        &&\mathcal{T}(H_1-E_1) = C_1 \tauftn{0}{-1/4, \rho+e_1} \, , \!
        &&\mathcal{T}(H_2-E_1) = \tauftn{1/2}{-1/2,e_1} \, , \\
        &\mathcal{T}(E_i) = \tauftn{1/2}{-1/2, e_i} \, ,  \!
        &&\mathcal{T}(H_1-E_i) = C_i \tauftn{0}{-1/4, \rho-e_i} \, , \!
        &&\mathcal{T}(H_2-E_i) = \tauftn{1/2}{-1/2,-e_i} \, , \\
        &\mathcal{T}(E_8) = \tauftn{1/2}{0, 0} \, , \!
        &&\mathcal{T}(H_1-E_8) = C_0 \tauftn{0}{-3/4, \rho} \, ,  \!
        &&\mathcal{T}(H_2-E_8) = \tauftn{1/2}{-1,0} \, ,
    \end{aligned}
\end{alignat}
where $ 2\leq i \leq 7 $, $ \rho = -e_1+\sum_{l=1}^7 e_l/2 $, and the normalization constants $ C_i $ are given in \eqref{eq:7F-C}. Then the bilinear relations can be expressed in the same form as \eqref{eq:elliptic-bilinear-ref}:
\begin{align}
    \begin{aligned}
        \Phi_{H_1}(E_i, E_j) \xi_k + (i, j, k)\text{-cyclic} = 0 \, , \quad
        \Phi_{H_2}(E_i, E_j) \eta_k + (i, j, k)\text{-cyclic} = 0 \, ,
    \end{aligned}
\end{align}
but $ \Phi_\kappa(x,y) $ in \eqref{eq:elliptic-Phi} is replaced by the trigonometric expression 
\begin{align}\label{eq:Phi-tri}
    \Phi_{\kappa}(x,y) = 4\sinh\Big(\frac{x-y}{2}\Big) \sinh\Big(\frac{\kappa-x-y}{2}\Big) \, .
\end{align}
The equivalence of these bilinear relations and the $ q $-Painlev\'e equation $ q\text{-}E_8^{(1)} $ is discussed in \cite{Masuda:2011, Kajiwara:2015}. Let us define the parameters $ \kappa_i $ and $ b_i $ as \eqref{eq:kap-v-parameter}, in the multiplicative notation. The discrete time evolution of $ q\text{-}E_8^{(1)} $ is defined as the Kac translation $ T_{\alpha_1} $, which acts on the parameters $ \kappa_i $ and $ v_i $ as \eqref{eq:Kac-kappa}. We now define
\begin{align}
    f(z) = z + \frac{\kappa_1}{z} \, , \quad
    g(z) = z + \frac{\kappa_2}{z} \, .
\end{align}
As in \eqref{eq:ZW}, we define $ (f,g) = (f(s), g(t)) $ such that
\begin{align}
    \frac{f-f(b_j)}{f-f(b_i)} = \frac{\xi_j}{\xi_i} \, , \quad
    \frac{g-g(b_j)}{g-g(b_i)} = \frac{\eta_j}{\eta_i} \, .
\end{align}
Then the time evolutions of $ f $ and $ g $ satisfy the $ q\text{-}E_8^{(1)} $ equation \eqref{eq:qE8}.

\subsection{\texorpdfstring{$ q\text{-}E_7^{(1)} $}{q-E7(1)} from bilinear relations}  \label{app:E7}

\begin{figure}
    \centering
    \begin{tikzpicture}
        \draw [thick] (0,0) -- (6,0);
        \draw [thick] (3,0) -- (3,1);
        \filldraw [thick, fill=white] (0,0) circle (0.1);
        \filldraw [thick, fill=white] (1,0) circle (0.1);
        \filldraw [thick, fill=white] (2,0) circle (0.1);
        \filldraw [thick, fill=white] (3,0) circle (0.1);
        \filldraw [thick, fill=white] (4,0) circle (0.1);
        \filldraw [thick, fill=white] (5,0) circle (0.1);
        \filldraw [thick, fill=white] (6,0) circle (0.1);
        \filldraw [thick, fill=white] (3,1) circle (0.1);
        \draw (0,-0.4) node {$ \alpha_0 $};
        \draw (1,-0.4) node {$ \alpha_1 $};
        \draw (2,-0.4) node {$ \alpha_2 $};
        \draw (3,-0.4) node {$ \alpha_3 $};
        \draw (4,-0.4) node {$ \alpha_4 $};
        \draw (5,-0.4) node {$ \alpha_5 $};
        \draw (6,-0.4) node {$ \alpha_6 $};
        \draw (3.4,1) node {$ \alpha_7 $};
    \end{tikzpicture}
    \caption{Dynkin diagram of affine $ E_7^{(1)} $} \label{fig:E7dynkin}
\end{figure}
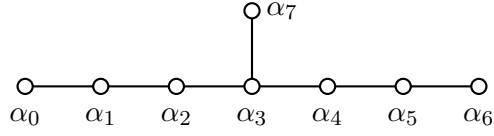

In this appendix, we present the relation between the bilinear relations in section~\ref{sec:6F} obtained from the 5d $ \mathrm{SU}(2)+6\mathbf{F} $ theory and those in the literature \cite{Masuda:2009}. Let us use the parametrization $ \{\varepsilon_i\}_{i=0}^9 $ of the nine-point blowup of $ \mathbb{P}^2 $ discussed in footnote~\ref{note:dP9}, instead of using $ H_i $ and $ E_i $. The Dynkin diagram of the $ E_7^{(1)} $ algebra is given in Figure~\ref{fig:E7dynkin}, and the simple roots are parametrized by
\begin{align}
    (\alpha_0,\alpha_1,\cdots,\alpha_7) = (\varepsilon_{89} , \varepsilon_{78}, \varepsilon_{127} , \varepsilon_{23}  , \varepsilon_{34}  , \varepsilon_{45}  , \varepsilon_{56}  , \varepsilon_{12} ) \, ,
\end{align}
where $ \varepsilon_{ij} = \varepsilon_i-\varepsilon_j $ and $ \varepsilon_{ijk} = \varepsilon_0-\varepsilon_i-\varepsilon_j-\varepsilon_k $. We associate the mass parameters $ m_i $ of the $ \mathrm{SU}(2)+6\mathbf{F} $ theory with these parameters of the Picard lattice as follows:
\begin{gather}
    \begin{gathered}\label{eq:qE7-mapping}
        \varepsilon_0 = 3m_0 + \mathfrak{m} \, , \quad
        \varepsilon_i = m_0 + m_i \, , \ (1 \leq i \leq 6) \, , \quad
        \varepsilon_7 = -\delta \, , \\
        \varepsilon_8 = 2m_0 \, , \quad
        \varepsilon_9 = m_0 + \mathfrak{m} \, , \quad
        \mathfrak{m} = \frac{1}{2} \sum_{i=1}^6 m_i \, .
    \end{gathered}
\end{gather}
Let us denote the tau-functions used in \cite{Masuda:2009} by $ \mathcal{T} $. Then the bilinear relations in \eqref{eq:6F-bilinear1} and \eqref{eq:6F-bilinear3} can be mapped to the bilinear relations given in equations (2.2)-(2.4) of \cite{Masuda:2009} using 
\begin{gather}
    \tauftn{1/2}{1/2,e_1} = \mathcal{T}(\varepsilon_1) \, , \quad
    \tauftn{1/2}{1/2,e_i} = \mathcal{T}(\varepsilon_i) \, , \quad
    \tauftn{1/2}{0,0} = \mathcal{T}(\varepsilon_7) \, , \quad
    \tauftn{1/2}{1,0} = \mathcal{T}(\varepsilon_8) \, , \nonumber \\
    \tauftn{0}{1/2,\rho+e_1} = \mathcal{T}(\varepsilon_9) \, , \quad
    C_7^{(2)}\tauftn{0}{1,\rho} = \mathcal{T}(\varepsilon_0-\varepsilon_1-\varepsilon_7) \, , \quad
    C_8^{(2)}\tauftn{0}{0,\rho} = \mathcal{T}(\varepsilon_0-\varepsilon_1-\varepsilon_8) \, , \nonumber \\
\tauftn{1/2}{1/2,-e_1} = \mathcal{T}(\varepsilon_0-\varepsilon_1-\varepsilon_9) \, , \quad
    \tauftn{1/2}{1/2,-e_i} = \mathcal{T}(\varepsilon_0-\varepsilon_i-\varepsilon_9) \, , \nonumber \\
    \tauftn{0}{1/2,\rho-e_i} = \mathcal{T}(\varepsilon_0-\varepsilon_1-\varepsilon_i) \, , \quad
    \tauftn{0}{1/2,\lambda-e_i-e_j} = \mathcal{T}(\varepsilon_0-\varepsilon_i-\varepsilon_j)
\end{gather}
where $i$ and $j$ are distinct integers $ 2 \leq i,  j \leq 6 $, and the normalization constants $ C_n^{(2)} $ are given in \eqref{eq:6F-tildem}. Using the parameter mapping \eqref{eq:qE7-mapping} together with $ \delta=-\epsilon $, one can reproduce the coefficients of the bilinear relations in \cite{Masuda:2009} from the result of section~\ref{sec:6F}. Thus, the bilinear relations obtained from the 5d $ \mathrm{SU}(2)+6\mathbf{F} $ theory are equivalent to the $ q\text{-}E_7^{(1)} $ Painlev\'e equation, whose explicit form is given in \eqref{eq:qE7}, where the discrete time evolution is defined by the Kac translation $ T_{\alpha_7} \in W(E_7^{(1)}) $. The details can be found in \cite{Masuda:2009}.

\subsection{\texorpdfstring{$ q\text{-}E_6^{(1)} $}{q-E6(1)} from bilinear relations} \label{app:E6}

We next connect the bilinear relations \eqref{eq:5F-bilinear} found from the $ \mathrm{SU}(2)+5\mathbf{F} $ theory with the $ q\text{-}E_6^{(1)} $ equation. Let us denote the tau-functions used in \cite{Tsuda:2008} by $ \mathcal{T}_n $ for $ 1 \leq n \leq 9 $. The tau-functions $ \tau_i^{(j)} $ in \eqref{eq:5F-tau} are identified with $ \mathcal{T}_n $ as
\begin{align}
    \tau_i^{(j)} = \mathfrak{q}^{p_1(B)} \prod_{k=1}^5 M_k^{\frac{2}{3}b_0 b_k} q^{p_2(B)} \mathcal{T}_{j+3i-3} \, ,
\end{align}
where $ B=(b_i)_{i=0}^5 $ is the background magnetic flux labelling the tau-function $ \tau_i^{(j)} $, and
\begin{align}
    \begin{aligned}
        p_1(B) &= \frac{4}{9} b_0 \left( b_0 + b_3 - b_4 + b_5 \right) + \frac{2}{3} b_1 b_2 \, , \\
        p_2(B) &= -\frac{4}{27} b_0^3 + \frac{4}{9} b_0^2 \left( -1 + b_3 - b_4 + b_5 \right) + \frac{2}{9} b_0 \left( -b_1 + b_2 + b_3 + 2b_5 + 3b_1b_2 \right) \, .
    \end{aligned}
\end{align}
Note that \cite{Tsuda:2008} further introduces $ U_{1,2,3} \sim \mathcal{T}_{1,4,7} $, $ V_{1,2,3} \sim \mathcal{T}_{2,5,8} $, $ W_{1,2,3} \sim \mathcal{T}_{3,6,9} $, up to normalization factors. Furthermore, the parameters in \cite{Tsuda:2008} are matched by
\begin{gather}
    a_0 = \frac{q^{1/3}}{\mathfrak{q}^{1/3}} \biggl(\frac{M_2M_4M_5}{M_1M_3}\biggr)^{1/6} \, , \quad
    a_1 = \biggl(\frac{M_4}{M_5}\biggr)^{\!1/3} \, , \quad
    a_2 = \frac{1}{(M_1M_4)^{1/3}} \, , \quad
    a_3 = (M_1M_2)^{1/3} \, , \nonumber \\
    a_4 = \mathfrak{q}^{1/3} \biggl(\frac{M_4M_5}{M_1M_2M_3}\biggr)^{1/6} \, , \quad
    a_5 = \frac{1}{\mathfrak{q}^{1/3}} \biggl(\frac{M_1}{M_2M_3M_4M_5}\biggr)^{1/6} \, , \quad
    a_6 = \biggl(\frac{M_3}{M_2}\biggr)^{1/3} \, ,
\end{gather}
together with the identification of the $ q $-variable $ q_T $ in \cite{Tsuda:2008} as $ q_T=q^{1/3} $, so that the normalization $ a_0 a_1 a_2^2 a_3^3 a_4^2 a_5 a_6^2=q_T $ is satisfied. The three maps $ \ell_j $ in the bilinear relation \eqref{eq:5F-bilinear} are identified with the translations $ \ell_j $ in \cite{Tsuda:2008} with same labels. One of the transformations $ \ell_j $ can be regarded as the discrete time evolution. Thus, the bilinear relations obtained from the 5d $ \mathrm{SU}(2)+5\mathbf{F} $ theory are equivalent to the $ q\text{-}E_6^{(1)} $ Painlev\'e equation.

\section{Instanton partition functions of 5d SU(2) gauge theories} \label{app:adhm}

In this appendix, we review instanton partition functions of 5d $ \mathrm{SU}(2) = \mathrm{Sp}(1) $ gauge theories coupled to $ N_f $ fundamental hypermultiplets. For $ k $ instantons, the 1d ADHM quantum mechanics has an $ \mathrm{O}(k) $ gauge group \cite{Shadchin:2005mx, Kim:2012gu, Hwang:2014uwa}. The $ \mathrm{O}(k) $ group has two disconnected components $ \mathrm{O}(k)_\pm $ consisting of the elements with determinant $ \pm 1 $, and this disconnectedness must be taken into account in the computation of the instanton partition functions. By denoting the contributions of the two disconnected components to the instanton partition function as $ Z_k^\pm $, the $ k $-instanton partition function can be written as
\begin{align}\label{eq:ADHM}
    Z_k = (\pm 1)^k \frac{Z_k^+ \pm Z_k^-}{2} \, , \quad
    Z_k^\pm = \frac{1}{|W^\pm|} \oint Z_{\mathrm{vec}}^\pm Z_{\mathrm{fund}}^{\pm} \, ,
\end{align}
where the $ \pm $ sign in $ Z_k $ corresponds to the choice of the discrete theta angle, which is irrelevant in the presence of fundamental hypermultiplets. We always choose the plus sign for $ N_f > 0 $ in this paper. Each $ Z_k^\pm $ can be represented as the contour integral of the 1-loop determinants of vector and hypermultiplets, and the integration contour is chosen by the Jeffrey-Kirwan (JK) residue prescription discussed in \cite{Benini:2013nda, Benini:2013xpa}. By denoting $ k=2n+\chi $ for $ \chi=0,1 $, the Weyl group factors $ |W^\pm| $ for $ \mathrm{O}(k)_\pm $ are given by
\begin{align}
    |W^+| = \left\{ \begin{array}{ll}
            \displaystyle {2^{n-1}n!} & \quad (\chi=0) \\[1em]
            \displaystyle {2^{n}n!} & \quad (\chi=1) \, ,
        \end{array}\right. \qquad
    |W^-| = \left\{ \begin{array}{ll}
            \displaystyle {2^{n-1}(n-1)!} & \quad (\chi=0) \\[1em]
            \displaystyle {2^{n}n!} & \quad (\chi=1) \, .
        \end{array}\right. \qquad
\end{align}
The 1-loop determinant of the vector multiplet is
\begin{align}
    Z_{\mathrm{vec}}^+ &= \left( \frac{1}{2\sinh(\frac{\epsilon_{1,2}}{2}) 2\sinh(\frac{\pm \phi + \epsilon_+}{2})} \right)^\chi \prod_{I<J}^n \frac{2\sinh(\frac{\pm u_I \pm u_J}{2}) 2\sinh(\frac{\pm u_I \pm u_J + 2\epsilon_+}{2})}{2\sinh(\frac{\pm u_I \pm u_J + \epsilon_{1,2}}{2})} \\
    &\!\!\!\!\!\! \times \prod_{I=1}^n \left( \frac{2\sinh(\frac{\pm u_I}{2}) 2\sinh( \frac{\pm u_I + 2\epsilon_+}{2})}{2\sinh(\frac{\pm u_I + \epsilon_{1,2}}{2})} \right)^\chi  \frac{du_I}{2\pi i} \frac{2\sinh(\epsilon_+)}{2\sinh(\frac{\epsilon_{1,2}}{2}) 2\sinh(\frac{\pm 2u_I + \epsilon_{1,2}}{2}) 2\sinh(\frac{\pm u_I \pm \phi + \epsilon_+}{2})}, \nonumber
\end{align}
for $ \mathrm{O}(k)_+ $;
\begin{align}
    \begin{aligned}
        Z_{\mathrm{vec}}^- &= \frac{1}{2\sinh(\frac{\epsilon_{1,2}}{2}) 2\cosh(\frac{\pm \phi + \epsilon_+}{2})} \prod_{I<J}^n \frac{2\sinh(\frac{\pm u_I \pm u_J}{2}) 2\sinh(\frac{\pm u_I \pm u_J + 2\epsilon_+}{2})}{2\sinh(\frac{\pm u_I \pm u_J + \epsilon_{1,2}}{2})} \\
        &\times \prod_{I=1}^n \frac{du_I}{2\pi i} \frac{2\cosh(\frac{\pm u_I}{2}) 2\cosh(\frac{\pm u_I+2\epsilon_+}{2}) 2\sinh(\epsilon_+)}{2\cosh(\frac{\pm u_I+\epsilon_{1,2}}{2}) 2\sinh(\frac{\epsilon_{1,2}}{2}) 2\sinh(\frac{\pm 2u_I + \epsilon_{1,2}}{2}) 2\sinh(\frac{\pm u_I \pm \phi + \epsilon_+}{2})} ,
    \end{aligned}
\end{align}
for $ \mathrm{O}(k)_- $ with $ \chi=1 $;
\begin{align}
    Z_{\mathrm{vec}}^- &= \frac{2\cosh(\epsilon_+)}{2\sinh(\frac{\epsilon_{1,2}}{2}) 2\sinh(\epsilon_{1,2}) 2\sinh(\pm \phi + \epsilon_+)} \prod_{I<J}^{n-1} \frac{2\sinh(\frac{\pm u_I \pm u_J}{2}) 2\sinh(\frac{\pm u_I \pm u_J + 2\epsilon_+}{2})}{2\sinh(\frac{\pm u_I \pm u_J + \epsilon_{1,2}}{2})} \nonumber \\
    &\times \prod_{I=1}^{n-1} \frac{du_I}{2\pi i} \frac{2\sinh(\pm u_I) 2\sinh(\pm u_I + 2\epsilon_+) 2\sinh(\epsilon_+)}{2\sinh(\frac{\epsilon_{1,2}}{2}) 2\sinh(\frac{\pm 2u_I + \epsilon_{1,2}}{2}) 2\sinh(\frac{\pm u_I \pm \phi + \epsilon_+}{2}) 2\sinh(\pm u_I + \epsilon_{1,2})} ,
\end{align}
for $ \mathrm{O}(k)_- $ with $ \chi=0 $. Here, we denote $ \phi $ as the Coulomb branch parameter of the 5d $ \mathrm{Sp}(1) $ gauge theory and we use the shorthand notation $ 2\sinh(x \pm y) = 2\sinh(x+y) 2\sinh(x-y) $ and $ 2\sinh(\epsilon_{1,2}) = 2\sinh(\epsilon_1) 2\sinh(\epsilon_2) $, etc. The 1-loop determinant of the hypermultiplets is
\begin{align}
    Z_{\mathrm{fund}}^+ = \prod_{l=1}^{N_f} \left( 2\sinh\Big(\frac{m_l}{2}\Big) \right)^\chi \prod_{I=1}^n 2\sinh\Big(\frac{\pm u_I+m_l}{2}\Big) ,
\end{align}
for $ \mathrm{O}(k)_+ $;
\begin{align}
    Z_{\mathrm{fund}}^- = \prod_{l=1}^{N_f} 2\cosh\Big(\frac{m_l}{2}\Big) \prod_{I=1}^n 2\sinh\Big(\frac{\pm u_I+m_l}{2}\Big) ,
\end{align}
for $ \mathrm{O}(k)_- $ with $ \chi=1 $; and
\begin{align}
    Z_{\mathrm{fund}}^- = \prod_{l=1}^{N_f} 2\sinh(m_l) \prod_{I=1}^{n-1} 2\sinh\Big(\frac{\pm u_I+m_l}{2}\Big) ,
\end{align}
for $ \mathrm{O}(k)_- $ with $ \chi=0 $.

We comment for the cases $ N_f \geq 6 $. When $ N_f=6 $, the 2-instanton partition function contains a term which does not depend on the dynamical parameter $ \phi $. This term is an extra factor that decouples from the 5d theory. Since we assume $ Z_k \sim \mathcal{O}(e^{-\phi}) $ to bootstrap bilinear relations of tau-functions, one needs to factor out this extra factor to check the validity of the ansatz at the level of instanton corrections. The extra factor can be factored out as
\begin{align}
    Z_{N_f=6} = Z_{N_f=6}^{\mathrm{ADHM}} \cdot \PE\left[ \frac{e^{-2m_0} \cosh(\epsilon_+)}{2\sinh(\epsilon_1/2) 2\sinh(\epsilon_2/2)} \right] \, ,
\end{align}
where $ Z^{\mathrm{ADHM}} $ is the partition function computed using \eqref{eq:ADHM}.

When $ N_f = 7 $ and $ 8 $, the residue integral \eqref{eq:ADHM} contains higher-degree poles at the infinities $ u_I = \pm \infty $. Unfortunately, a systematic method to treat such poles is unknown. Nevertheless, it is possible to compute instanton partition functions of $ \mathrm{SU}(2)+7\mathbf{F} $ and $ \mathrm{SU}(2)+8\mathbf{F} $ theories by coupling them to an auxiliary antisymmetric hypermultiplet. The $ k $-instanton partition function can be written as
\begin{align}
    Z_k^\pm = \frac{1}{|W^\pm|} \oint Z_{\mathrm{vec}}^\pm Z_{\mathrm{fund}}^{\pm} Z_{\mathrm{anti}}^\pm \, ,
\end{align}
where the 1-loop determinant of the antisymmetric hypermultiplet is given by
\begin{align}
    Z_{\mathrm{anti}}^+ &= \left( \frac{2\sinh(\frac{m_A\pm \phi}{2})}{2\sinh(\frac{m_A\pm \epsilon_+}{2})} \prod_{I=1}^n \frac{2\sinh(\frac{\pm u_I \pm m_A - \epsilon_-}{2})}{2\sinh(\frac{\pm u_I \pm m_A - \epsilon_+}{2})} \right)^\chi \! \prod_{I=1}^n \frac{2\sinh(\frac{\pm m_A - \epsilon_-}{2}) 2\sinh(\frac{\pm u_I \pm \phi - m_A}{2})}{2\sinh(\frac{\pm m_A-\epsilon_+}{2}) 2\sinh(\frac{\pm 2u_I \pm m_A - \epsilon_+}{2})} \nonumber \\
    &\quad \times \prod_{I<J}^n \frac{2\sinh(\frac{\pm u_I \pm u_J \pm m_A - \epsilon_-}{2})}{2\sinh(\frac{\pm u_I \pm u_J \pm m_A - \epsilon_+}{2})} \, ,
\end{align}
for $ \mathrm{O}(k)_+ $;
\begin{align}
    \begin{aligned}
        Z_{\mathrm{anti}}^- &= \frac{2\cosh(\frac{m_A\pm \phi}{2})}{2\sinh(\frac{m_A\pm \epsilon_+}{2})} \prod_{I=1}^n \frac{ 2\cosh(\frac{\pm u_I \pm m_A - \epsilon_-}{2}) 2\sinh(\frac{\pm m_A - \epsilon_-}{2}) 2\sinh(\frac{\pm u_I \pm \phi - m_A}{2})}{2\cosh(\frac{\pm u_I \pm m_A - \epsilon_+}{2}) 2\sinh(\frac{\pm m_A-\epsilon_+}{2}) 2\sinh(\frac{\pm 2u_I \pm m_A - \epsilon_+}{2})} \\
        &\quad \times \prod_{I<J}^n \frac{2\sinh(\frac{\pm u_I \pm u_J \pm m_A - \epsilon_-}{2})}{2\sinh(\frac{\pm u_I \pm u_J \pm m_A - \epsilon_+}{2})} \, ,
    \end{aligned}
\end{align}
for $ \mathrm{O}(k)_- $ with $ \chi=1 $; and
\begin{align}
        Z_{\mathrm{anti}}^- &= \frac{2\cosh(\frac{\pm m_A - \epsilon_-}{2}) 2\sinh(m_A\pm \phi)}{2\sinh(\frac{m_A \pm \epsilon_+}{2}) 2\sinh(m_A\pm \epsilon_+)} \prod_{I=1}^{n-1} \frac{2\sinh(\pm u_I \pm m_A - \epsilon_-) 2\sinh(\frac{\pm m_A-\epsilon_-}{2}) }{2\sinh(\pm u_I \pm m_A - \epsilon_+) 2\sinh(\frac{\pm m_A-\epsilon_+}{2})} \nonumber \\
        &\quad \times \prod_{I=1}^{n-1} \frac{2\sinh(\frac{\pm u_I \pm \phi - m_A}{2})}{2\sinh(\frac{\pm 2u_I \pm m_A - \epsilon_+}{2})} \cdot \prod_{I<J}^{n-1} \frac{2\sinh(\frac{\pm u_I \pm u_J \pm m_A - \epsilon_-}{2})}{2\sinh(\frac{\pm u_I \pm u_J \pm m_A - \epsilon_+}{2})}
\end{align}
for $ \mathrm{O}(k)_- $ with $ \chi=0 $. The integration contour is chosen by the JK-residue prescription, and the poles from the antisymmetric hypermultiplet also contribute. For instance, at the 2-instanton order, one can choose the poles of $ Z_{\mathrm{vec}}^+ $ and $ Z_{\mathrm{anti}}^+ $ located at
\begin{align}
    2u_1 + \epsilon_{1,2} = 0 , \, 2\pi i \, , \quad
    u_1 \pm \phi + \epsilon_+ = 0  , \quad
    2u_1 \pm m_A - \epsilon_+ = 0 , \, 2\pi i \, .
\end{align}
After computing the integral, we take the $ m_A \to \infty $ limit to decouple the auxiliary antisymmetric hypermultiplet. This yields the instanton partition functions of the $ \mathrm{SU}(2)+7\mathbf{F} $ and $ 8\mathbf{F} $ theories.

\bibliographystyle{JHEP}
\bibliography{refs}

\end{document}